\documentclass[letterpaper,12pt]{article}

\usepackage{setspace}
\usepackage{arxiv}
\usepackage{booktabs}
\mathtoolsset{showonlyrefs}
\allowdisplaybreaks   

\usepackage{lineno}
\usepackage{tikz}
\usetikzlibrary{positioning,arrows.meta}
\newcommand{\indsim}{\overset{\mathrm{ind}}{\sim}}
\newcommand{\pto}{\overset{\mathrm{p.}}{\longrightarrow}}
\newcommand{\ADE}{\mathrm{ADE}}
\newcommand{\AIE}{\mathrm{AIE}}
\newcommand{\MPE}{\mathrm{MPE}}
\newcommand{\TOT}{\mathrm{TOT}}
\newcommand{\PC}{\mathrm{PC}}
\newcommand{\rmL}{\mathrm{L}}
\newcommand{\rmG}{\mathrm{G}}

\newcommand{\RCT}{\mathrm{RCT}}
\newcommand{\oracle}{\mathrm{oracle}}
\newcommand{\functional}{\mathrm{func}}

\newcommand{\spillovertype}[1]{\hyperlink{#1}{$\mathsf{(#1)}$}}

\usepackage{float}

\usepackage{graphicx}
\usepackage{caption}
\usepackage{subcaption}
\usepackage{float} 
\usepackage{array}

\begin{document}


\title{\Large Decomposition of Spillover Effects Under Misspecification:\\
Pseudo-True Estimands and a Local-Global Extension\thanks{\setlength{\baselineskip}{4mm}We thank Isaiah Andrews and Davide Viviano for generous support. We also benefited from helpful comments by Boan Chen, Raj Chetty, Michael Leung, Shuangning Li, Evan Munro, Daniel Nevo, Jesse Shapiro, Elie Tamer, Austin Zheng, and Liang Zhong, as well as from discussions with participants at the Harvard Graduate Student Workshop in Econometrics and the Statistics Seminar at Harvard University. Xiaodong Yang was supported by his advisor Subhabrata Sen using Grants FA9950-23-1-0429 from AFOSR and N00014-23-1-2489 from ONR. Yechan Park acknowledges support from the Harvard Griffin Fund in Economics. Stefan Nicov provided excellent research assistance. All remaining errors are ours.\smallskip}}

\author{
Yechan Park\thanks{
Department of Economics, Harvard University.
Littauer Center, 1805 Cambridge St, Cambridge, MA 02138.
Email: \texttt{yechanpark@fas.harvard.edu}
}
\and
Xiaodong Yang\thanks{
Department of Statistics, Harvard University.
Maxwell Dworkin 242, 33 Oxford St, Cambridge, MA 02138.
Email: \texttt{xyang@g.harvard.edu}
}
\thanks{The authors contributed equally to this work.}
}

\date{}

\maketitle
\begin{abstract}\setlength{\baselineskip}{5.6mm}


To measure spillovers, researchers often summarize who else was treated with a simple measure, such as the share of treated neighbors. We study what the researcher estimates when that summary is misspecified. We show that the researcher estimates the best approximation to the true policy effect among all functionals of the chosen summary, and that the usual direct and spillover estimates are exactly the components of this approximation. We then specialize this framework to a setting common in applications such as cash transfers: treatment spills over both globally, through market equilibrium, and locally, through network externalities, while the researcher models only one channel. The network estimator still recovers the network spillover, and the equilibrium estimator the equilibrium spillover, even when the other channel is ignored. We illustrate with a simulation calibrated to a large cash-transfer experiment.

\end{abstract}

\section{Introduction}

In many empirical settings, one unit’s treatment affects the outcomes of others. Vaccination programs change not only the health of vaccinated individuals, but also the infection risks of their contacts \citep{HudgensHalloran2008}. Large-scale anti-poverty programs alter local markets and prices, with consequences that propagate through space \citep{egger2022ge}. The challenge in these settings is that the assignment vector acts as a single, system-wide shock rather than a collection of independent unit-level treatments: the outcome for each unit can depend on many, or even all, components of the treatment vector. From the researcher’s perspective, this means that one realized assignment from a fixed design must be used to learn about this complex pattern of dependence.

Applied work typically approaches these settings through \emph{exposure mappings} \citep{AronowSamii2017}: low-dimensional summaries of the underlying interference structure, such as the fraction of treated neighbors in a network \citep{cai2015social}, a spatial ring kernel \citep{egger2022ge}, or market prices \citep{munro2021treatment}. Outcomes are modeled as depending on own treatment and this exposure measure, and empirical work reports “direct effects’’ and “spillover effects’’ defined within this reduced description of the interference structure.

A central difficulty is that exposure mappings are often inevitably misspecified \citep{savje2024causal}. They compress a rich pattern of interference into a simple index that omits many potentially relevant details about who is treated and how spillovers operate. This raises a basic question: when the exposure mapping is only an approximation to the true interference structure, what policy object are exposure-based estimands actually targeting, and how should we interpret their direct and spillover components relative to the underlying policy question?

This paper answers that question by taking a primitive policy object as the target and then characterizing what the exposure-based estimands recover relative to it. We focus on the effect of marginally changing the treatment assignment rule, often called the \textit{marginal policy effect} in the recent literature \citep{carneiro2010evaluating,munro2021treatment}. This quantity is defined directly from the experimental or quasi-experimental design, without reference to any exposure mapping. It has a natural interpretation as a social multiplier that captures the aggregate impact of a small change in treatment intensity, and it is often more tractable than other counterfactual quantities \citep{munro2021treatment}. Many recent theoretical contributions analyze marginal policy effects of this form \citep[e.g.,][]{li2022random,munro2021treatment,arkhangelsky2025evaluating,hu2022average}, and they have a number of practical applications, for example as regression coefficients in linear regressions with both an own-treatment indicator and an exposure variable \citep{egger2022ge,muralidharan2023general}; see also the three examples in \citet{hu2022average}.

Given this policy primitive, we then study what happens when the analyst commits to using an exposure mapping chosen on the basis of domain knowledge. We show that this restriction induces a pseudo-true outcome model: among all outcome models that depend on the assignment vector only through the chosen exposure, there is a unique model that provides the best mean-squared approximation to the true outcomes. This pseudo-true model in turn induces marginal-policy, direct, and spillover effects. This extends the decomposition of \citet{hu2022average}, who show that under correct specification the marginal policy effect splits exactly into a direct and a spillover effect, to the pseudo-true objects induced by a misspecified exposure mapping. Thus, our contribution is not to claim that an arbitrary misspecified exposure mapping is automatically informative about the oracle policy effect, but rather to characterize the canonical target implied by the maintained exposure restriction. We then quantify when these pseudo-true estimands are close to their oracle counterparts. Accordingly, Section~\ref{sec:misspecified estimand} provides a general framework for interpreting what exposure-based procedures target under misspecification.

So far, we have intentionally been agnostic about the detailed structure of interference. As emphasized by \citet{leung2024discussion} and \citet{auerbach2024discussion}, without additional structure one should not expect exposure-based analyses to recover finer channels beyond what the maintained restriction encodes. At the same time, many applied settings feature multiple distinct spillover mechanisms, which raises the question of when the general pseudo-true objects above admit a sharper interpretation. To study this, we introduce a structured model class that nests many important empirical contexts \citep[e.g.,][]{egger2022ge,AngelucciDeGiorgi2009} and theoretical work on interference \citep[e.g.,][]{li2022random,munro2021treatment}. Specifically, we focus on environments in which local network spillovers and global spillovers, such as equilibrium prices, wages, or epidemic states, operate simultaneously.

This local-global model should be viewed as a structured specialization of our general misspecification framework, in which the pseudo-true interpretation sharpens into a statement about oracle channel-specific components. In this class of models, the oracle marginal policy effect admits an asymptotic three-way decomposition into a direct effect, a local spillover effect, and a global spillover effect. Specifically, a researcher who uses only a local exposure mapping can still be viewed as targeting the local component, while a researcher who uses only a global exposure mapping targets the global component, even though each omits the other first-order channel. More generally, Section~\ref{sec:misspecified estimand} continues to allow for additional omitted channels beyond these maintained local and global ones; when those residual channels are asymptotically negligible, the corresponding pseudo-true estimands remain close to the local-global oracle targets, and when they are not, they still retain their interpretation as the best exposure-based $L^2$ approximations.

An important implication is that many existing methods are more robust than previously understood once we reinterpret their targets as channel-specific components of this pseudo-true estimand. In particular, network estimators of Li--Wager type remain consistent for the local spillover component even in the presence of global spillovers. With additional sources of variation, such as augmented randomization schemes or instrumental-variable perturbations of global state variables, the global spillover component can also be separately recovered. We illustrate this idea through a semi-synthetic experiment calibrated to real data from the large-scale cash-transfer experiment studied by \citet{filmer2023cash}.

\paragraph{Roadmap.}
Figure~\ref{fig:intro_bridge} summarizes the relation between the paper’s two main conceptual components. Section~\ref{sec:misspecified estimand} develops a general pseudo-true framework for misspecified exposure mappings and studies the resulting direct, indirect, and marginal policy effects. Section~\ref{sec:example expo} then discusses motivating exposure mappings and examples in which local and global spillover channels may coexist. Building on this, Section~\ref{sec:local vs global example} specializes the general framework to a structured local-global environment, in which the pseudo-true objects from Section~\ref{sec:misspecified estimand} admit a sharper channel-specific interpretation and can be estimated using procedures adapted to the corresponding channel. Section~\ref{sec:gen_simu} evaluates these ideas in simulation and semi-synthetic designs. Section~\ref{sec:conclusion} concludes.

\begin{figure}[!htbp]
\centering
\footnotesize
\begin{minipage}[t]{0.43\textwidth}
\centering
\begin{tikzpicture}[
    >=Latex,
    font=\scriptsize,
    box/.style={
        draw,
        rectangle,
        rounded corners=1.5pt,
        minimum width=1.15cm,
        minimum height=0.58cm,
        align=center,
        inner sep=1.5pt
    }
]
\node[box] (w) at (0,0) {$\mathbf W$};
\node[box] (d) at (1.75,0) {$d_i(\mathbf W)$};
\node[box] (y) at (3.50,0) {$y_i(\mathbf W)$};
\draw[->] (w) -- (d);
\draw[->] (d) -- (y);
\draw[->, dashed, red!80, thick]
    (w.south east) .. controls +(0.55,-0.5) and +(-0.55,-0.5) .. (y.south west);
\node[red!80, font=\tiny] at (1.75,-1.0) {omitted channel(s)};
\end{tikzpicture}
\par\vspace{0.4em}
{\scriptsize \textbf{(a) General framework} (Section~\ref{sec:misspecified estimand})}\par
\end{minipage}
\hspace{0.02\textwidth}
\begin{minipage}[t]{0.52\textwidth}
\centering
\begin{tikzpicture}[
    >=Latex,
    font=\scriptsize,
    box/.style={
        draw,
        rectangle,
        rounded corners=1.5pt,
        minimum width=1.2cm,
        minimum height=0.58cm,
        align=center,
        inner sep=1.5pt
    }
]
\node[box] (w)  at (0,0) {$\mathbf W$};
\node[box] (dL) at (1.95,0.95) {$d_i^{L}(\mathbf W)$};
\node[box] (dG) at (1.95,-0.95) {$d_i^{G}(\mathbf W)$};
\node[box] (y)  at (3.90,0) {$y_i(\mathbf W)$};
\draw[->] (w) -- (dL);
\draw[->] (w) -- (dG);
\draw[->] (dL) -- (y);
\draw[->] (dG) -- (y);
\draw[->, dashed, red!80, thick]
    (w.east) .. controls +(0.85,-0.16) and +(-0.85,-0.16) .. (y.west);
\node[red!80, font=\tiny, align=center] at (1.95,-0.44)
    {extra omitted\\ channel(s)};
\node[font=\tiny, align=center] at (1.95,1.72) {local only};
\node[font=\tiny, align=center] at (1.95,-1.72) {global only};
\draw[->, dashed] (1.95,1.54) -- (dL.north);
\draw[->, dashed] (1.95,-1.54) -- (dG.south);
\end{tikzpicture}
\par\vspace{0.4em}
{\scriptsize \textbf{(b) Structured local-global model} (Section~\ref{sec:local vs global example})}\par
\end{minipage}
\caption{\footnotesize Panel (a) illustrates the general framework (Section \ref{sec:misspecified estimand}), in which a researcher summarizes the assignment vector \(\mathbf W\) through an exposure mapping \(d_i(\mathbf W)\), while outcomes may still depend on additional omitted channels. The resulting estimand is the pseudo-true outcome model: the best mean-squared approximation among functions that depend on \(\mathbf W\) only through \(d_i(\mathbf W)\). Panel (b) illustrates the structured local-global environment (Section \ref{sec:local vs global example}), in which outcomes depend on both a local exposure \(d_i^L(\mathbf W)\) and a global exposure \(d_i^G(\mathbf W)\), even though the researcher may model only one channel. In this setting, the general pseudo-true framework sharpens to a channel-specific interpretation under the local-global asymptotic structure. The dashed red arrow denotes possible additional omitted channels beyond the maintained local and global ones; these are asymptotically harmless when their residual contribution after conditioning is negligible (Corollary~\ref{corollary:approximate sufficiency}).}
\label{fig:intro_bridge}
\end{figure}

\vspace*{2mm}\noindent\textbf{Related literature}

\textbf{Misspecification in spillover estimation and policy-relevant primitives.}
A large literature now treats exposure mappings and randomization-based designs as the basic language for analyzing interference \citep[e.g.][]{Aronow2012Detect,AronowSamii2017,AtheyEcklesImbens2018,ogburn2024causal}. Against this backdrop, a more recent line of work takes misspecification of spillover structures seriously \citep{loomba2025policy,shuangning_additional_draft,weinstein2026causal,Park2026Exposure}. In particular, \citet{SavjeAronowHudgens2021} show that ADE can be estimated even under unknown interference, while \citet{savje2024causal} treats exposure mappings as researcher-defined summaries and derives conditions for consistent estimation under misspecification, prompting discussion of the policy content of the resulting exposure effects \citep{auerbach2024discussion,leung2024discussion}. Relatedly, \citet{leung2022causal} formalizes approximate neighborhood interference, under which standard exposure-based estimators remain well behaved even when distant treatments matter, and \citet{menzel2025fixed} defines conditional-on-assignment estimands that remain identified under very general interference and can be recovered by inverse-probability weighting in single-network experiments. We introduce a pseudo-true estimand perspective and formally construct it using a two-copy conditioning device. We establish approximation guarantees. This parallels classic pseudo-true parameter ideas in econometrics and finance, where maximum likelihood under misspecification converges to a Kullback--Leibler projection \citep{white1982maximum} and Hansen--Jagannathan distance selects the stochastic discount factor that minimizes a pricing-error norm \citep{hansen1997assessing}.

\textbf{Spillover decompositions and a local-global extension}
A separate literature uses decompositions of overall policy effects into ``direct'' and ``indirect'' components to organize mechanisms. A long line of work \citep{Sobel2006,HudgensHalloran2008}
has formalized direct and indirect effects under partial interference, with extensions to general exposure mappings in \citet{AronowSamii2017} and design-averaged estimands under unknown interference in \citet{SavjeAronowHudgens2021}. Within this tradition, \citet{hu2022average} define average direct and indirect effects under general exposure mappings and show that, in Bernoulli trials, their sum coincides with the effect of an infinitesimal increase in the treatment probability, an approach adopted in structured settings such as the market-equilibrium model of \citet{munro2021treatment}. Much of this work effectively treats the indirect component as a single channel. Our analysis shows, first, that this basic direct--indirect decomposition survives misspecification once exposure effects are interpreted as pseudo-true components of a marginal policy effect. We obtain a further sharper decomposition under a structured local-global framework, where we nest the local network framework of \citet{li2022random} and the equilibrium-spillover framework of \citet{munro2021treatment}; see also related global-state settings in \citet{arkhangelsky2025evaluating,halloran1991direct,lin2024sir}\footnote{Other complementary work includes \citet{bhattacharya2025causal}, who use mean-field methods to study global treatment effects.}. We show how to interpret these path-specific settings under other forms of misspecification, and we show that contamination from the other channel is asymptotically negligible, admitting a separable decomposition into the direct, local, and global indirect effects. Recent complementary work by \citet{Ritzwoller2025Spillovers} shows that regressions on proximity-weighted treatments blend multiple channels unless the proximity measure is residualized.



\section{Average effects with a misspecified exposure}\label{sec:misspecified estimand}

Consider a sample of $n$ units indexed by
$\{1,\ldots,n\}$, where each unit is assigned one of two possible treatments $\{0,1\}$. The
collection of all (potentially counterfactual) assignments is thus denoted as
$\bw=(w_1,\ldots,w_n)\in\{0,1\}^n$. A (possibly randomized) function
$y_i:\{0,1\}^n\to\RR$ gives the potential outcome for unit $i$ under a specific assignment.
We impose no additional structure on $y_i(\cdot)$ until
Section~\ref{sec:local vs global example}.

Throughout, we focus on experimental designs where the actual assignment vector $\bW\in\{0,1\}^n$ is generated randomly. In particular, we consider the simplest design, a \emph{randomized controlled trial} with a homogeneous treatment probability $\pi\in(0,1)$.
\begin{assumption}\label{assump:RCT}
    Draw $\bW=(W_1,\ldots,W_n)\sim\RCT(\pi)$, i.e.\ each $W_i$ independently satisfies
    $\PP(W_i=1)=1-\PP(W_i=0)=\pi$.
\end{assumption}

Our benchmark target is the oracle marginal policy effect
\begin{equation}\label{eq:oracle estimand}
    \tau_{\MPE}^{\oracle}(\pi) = \frac{1}{n}\sum_{j=1}^n\sum_{i=1}^n\EE_{\bW\sim\RCT(\pi)}\sbr{ y_j\rbr{w_i=1,\bW_{-i}} - y_j\rbr{w_i=0,\bW_{-i}} }.
\end{equation}
This estimand is defined without reference to any exposure mapping and serves as the oracle benchmark. Since $y_i(\cdot)$ is defined on the exponentially large assignment space $\{0,1\}^n$, directly estimating \eqref{eq:oracle estimand} is generally intractable.

Sections~\ref{sec:pseudo-true outcome} and \ref{sec:estimands defined via conditioning} therefore construct pseudo-true outcome models by conditioning on a researcher-chosen exposure mapping and use them to define tractable marginal-policy, direct, and indirect effects. These estimands admit an intervention-based interpretation as components of the induced pseudo-true model.

Misspecification arises whenever the researcher-chosen exposure mapping \(d_i(\bW)\) is not a sufficient summary of the assignment vector for the outcome \(y_i(\bW)\). This can happen for many reasons, but two broad motivations are especially relevant for our analysis. First, the analyst may use a low-dimensional exposure mapping as a tractable approximation to a richer interference structure, for example because the true mechanism is too complex to model directly or because tuning details of the exposure are uncertain. In this case, Section~\ref{sec:error to oracle} shows that the resulting pseudo-true estimands remain close to the oracle targets when the chosen exposure mapping is sufficiently informative about outcomes. Second, the analyst may intentionally work with only one exposure channel in order to obtain a more interpretable decomposition of spillovers. In that case, the resulting estimands should be viewed as channel-specific pseudo-true effects rather than as attempts to recover the full oracle object; Section~\ref{sec:local vs global example} shows that this interpretation is especially sharp in a local-global environment.

\subsection{Pseudo-true outcome model}\label{sec:pseudo-true outcome}
A large empirical literature works with \emph{exposure mappings} that summarize the features of the assignment vector $\bw$ that practitioners believe to be most relevant for unit $i$. Formally, the analyst specifies
$d_i:\{0,1\}^n\to\cD_i$\footnote{See, among many others,
\citet{HudgensHalloran2008,tchetgen2012causal,AronowSamii2017,hu2022average,li2022random,munro2021treatment}
for examples in epidemiology, statistics, and economics.}. Choosing $d_i$ is problem specific and requires domain expertise. We offer several examples in Section~\ref{sec:example expo}. In the well-specified setting, one posits that $y_i(\bw)$ depends on $\bw$ only through $d_i(\bw)$\footnote{Technically, one typically assumes that $y_i(\bw)$ depends on $\bw$ only through
$(w_i,d_i(\bw))$. Since we can redefine the exposure mapping as
$\tilde d_i(\bw) := (w_i,d_i(\bw))$, this is without loss of generality; in what follows we
often treat $d_i$ as already including the own-treatment component.}

However, in realistic environments with rich local and global spillovers, misspecification of exposure mappings is hard to avoid. Under such circumstances, we aim for tractable alternatives for the oracle estimands in~\eqref{eq:oracle estimand}.  A large body of empirical and methodological work postulates that potential outcomes depend on $\bw$ only through an exposure mapping $d_i(\bw)$, and then estimates causal effects by working with outcome models of the form $h_i\bigl(d_i(\bw)\bigr)$—for instance by pooling outcomes across units with the same or similar exposure, or by fitting flexible regressions of $Y_i$ on $d_i(\bW)$; see, e.g., \citet{AronowSamii2017,auerbach2021local,zivich2022tmle}. 

In the same spirit, we build exposure–based outcome models of the form $h_i\bigl(d_i(\bw)\bigr)$ and use them to define alternative estimands of interest under interference.
A natural criterion is to optimize over $h_i$ so that the following square loss is minimized:
\begin{equation}\label{eq:square loss criterion}
    \min_{h_i} \, \EE_{\bW\sim\RCT(\pi)} 
        \bigl[\, y_i(\bW) - h_i\bigl(d_i(\bW)\bigr) \,\bigr]^2.
\end{equation}
The solution is the conditional expectation:
\begin{equation}\label{eq:hi-star-def}
    h_i^{\ast}(d;\pi)
    \;=\;
    \EE_{\bW\sim\RCT(\pi)}\bigl[\,y_i(\bW)\,\big|\,d_i(\bW)=d\,\bigr],
    \quad \text{for any }d\in\cD_i.
\end{equation}
Here we slightly abuse notation by incorporating $\pi$ as another argument of
$h_i^\ast$, simply to emphasize that this optimal solution is \emph{design-induced}. We then define the corresponding \emph{exposure-based outcome model}
\begin{equation}\label{eq:exposure-based outcome proxy}
    \tilde{y}_i(\bw;\pi)
    \;:=\;
    h_i^{\ast}\bigl(d_i(\bw);\pi\bigr)
    \;=\;
    \EE_{\bW^{(2)}\sim\RCT(\pi)}
    \bigl[\,y_i(\bW^{(2)})\,\big|\,d_i(\bW^{(2)})=d_i(\bw)\,\bigr],
\end{equation}
where $\bW^{(2)}$ is an independent copy of the treatment vector, introduced to average out omitted interference conditional on the exposure. Among all outcome models that depend on $\bw$ only through $d_i(\bw)$, $\tilde{y}_i$ is the unique solution that minimizes the mean-squared discrepancy from the true $y_i$ under the design. We therefore call it \emph{pseudo-true}, following the misspecification literature on pseudo-true parameters in minimum-distance, likelihood, and related settings \citep{white1982maximum,hall2003large,hansen1997assessing,andrews2024misspecified}.

An important practical feature of~\eqref{eq:hi-star-def}--\eqref{eq:exposure-based outcome proxy}
is that, once the exposure mapping $\{d_i\}$ is fixed, the pseudo-true outcomes are
functions only of the joint distribution of $(Y_i,d_i(\bW))$ under $\RCT(\pi)$. In particular,
any flexible estimator of a conditional expectation—including classical inverse–probability-weighted and regression estimators, as well as modern machine–learning methods for nuisance functions—can be used to approximate $h_i^\ast$ and hence
$\tilde y_i(\cdot;\pi)$ without modeling the full interference structure; see, for example, \citet{chernozhukov2018double,wager2018estimation}.

\subsection{Estimands built on conditioning}\label{sec:estimands defined via conditioning}

We now use the pseudo-true outcome models~\eqref{eq:exposure-based outcome proxy} to
define average effects of interest. Our estimands extend the familiar ADE/AIE objects
studied under correctly specified exposure mappings in
\cite{hu2022average,li2022random,munro2021treatment}. In the correctly specified case,
the celebrated result of \citet{hu2022average} shows that, under the Bernoulli design
$\RCT(\pi)$, the marginal policy effect $\tau_{\MPE}(\pi)$ admits a principled
decomposition into an average direct effect and an average indirect effect. This
decomposition has been used both in recent theoretical analysis
(e.g., \citet{munro2021treatment,arkhangelsky2025evaluating,loomba2025policy}) and in applied work (e.g., \citet{behaghel2022encouraging}). Here we extend it to the misspecified case by replacing oracle outcomes
$y_i(\cdot)$ with the pseudo-true outcomes $\tilde y_i(\cdot;\pi)$.

\vspace*{2mm}\noindent\textbf{Marginal policy effect.}
Recall from~\eqref{eq:oracle estimand} and the results of \cite{hu2022average}, the oracle marginal effect can be written as
\begin{equation}
    \tau_{\MPE}^{\oracle}(\pi) = \frac{1}{n} \frac{\partial}{\partial\pi}\sum_{i=1}^n \EE\sbr{ y_i(\bW)} = \frac{1}{n} \frac{\partial}{\partial\pi}\sum_{i=1}^n \EE_{\bW\sim\RCT(\pi)}\sbr{ h_i^{\ast}(d_i(\bW);\pi)}.
\end{equation}
To isolate the role of $d_i$ from other unknown interference mechanisms, consider two independent assignments
$\bW^{(1)}\sim\RCT(\pi_1)$ and $\bW^{(2)}\sim\RCT(\pi_2)$. Under this two copy setup, we can freely differentiate with respect to an infinitesimal change in the law of $d_i(\bW^{(1)})$. Equivalent to the conditioning idea in
\eqref{eq:exposure-based outcome proxy}, we define
\begin{align}            
    \mu(\pi_1,\pi_2)
    \;&:=\;
    \frac{1}{n} \sum_{i=1}^n
    \EE_{\bW^{(1)}\sim\RCT(\pi_1)}
    \Bigl\{
      \EE_{\bW^{(2)} \sim\RCT(\pi_2)}
      \bigl[
        y_i\bigl(\bW^{(2)}\bigr)
        \,\big|\,
        d_i\bigl(\bW^{(2)}\bigr)=d_i\bigl(\bW^{(1)}\bigr)
      \bigr]
    \Bigr\} \\
    &= \; \frac{1}{n} \sum_{i=1}^n
    \EE_{\bW^{(1)}\sim\RCT(\pi_1)}
    \Bigl\{
      h_i^{\ast}(d_i(\bW^{(1)});\pi_2)
    \Bigr\}.
\end{align}
Then the marginal policy effect under exposure mappings $\{d_i\}$ and $\RCT(\pi)$ is given by
\begin{equation}
    \tau_{\MPE}(\pi)
    \;:=\;
    \frac{\partial}{\partial\pi_1} \mu(\pi_1,\pi_2)  \Big|_{\pi_1=\pi_2=\pi}.
\end{equation}

\vspace*{2mm}\noindent\textbf{Direct and indirect effects.}
The direct and indirect treatment effects under exposure mappings $\{d_i,i\in[n]\}$ and $\RCT(\pi)$ are defined analogously:
\begin{align}
    \tau_{\ADE}(\pi)
    &:= \frac{1}{n}\sum_{i=1}^n
    \EE_{\bW\sim\RCT(\pi)}
    \bigl[
      \tilde{y}_i\bigl(w_i=1,\bW_{-i};\pi\bigr)
      - \tilde{y}_i\bigl(w_i=0,\bW_{-i};\pi\bigr)
    \bigr],
    \label{eq:ADE-def}\\
    \tau_{\AIE}(\pi)
    &:= \frac{1}{n}\sum_{j=1}^n\sum_{i\neq j}
    \EE_{\bW\sim\RCT(\pi)}
    \bigl[
      \tilde{y}_j\bigl(w_i=1,\bW_{-i};\pi\bigr)
      - \tilde{y}_j\bigl(w_i=0,\bW_{-i};\pi\bigr)
    \bigr].
    \label{eq:AIE-def}
\end{align}

When exposures are correctly specified, these estimands coincide with their oracle counterparts in \eqref{eq:oracle estimand}.
Our first result states that these exposure-based estimands admit exactly the same decomposition as in \citet{hu2022average}. The proof is simple and is deferred to the appendix.

\begin{theorem}\label{thm:mpe equal direct plus indirect}
    Under $\RCT(\pi)$ (Assumption~\ref{assump:RCT}), it holds that \[\tau_{\MPE}(\pi)
        \;=\;
        \tau_{\ADE}(\pi)
        \;+\;
        \tau_{\AIE}(\pi).\]
\end{theorem}


\citet{savje2024causal} also studies misspecified exposure mappings via conditioning, but defines causal effects as contrasts between exposure labels. Motivated by the subsequent discussions in \citet{auerbach2024discussion,leung2024discussion}, we take a different route. We begin from the oracle marginal policy effect in \eqref{eq:oracle estimand} and ask what object is induced when the analyst restricts attention to outcome models that depend on the assignment vector only through a chosen exposure mapping. The resulting estimands \(\tau_{\MPE}(\pi)\), \(\tau_{\ADE}(\pi)\), and \(\tau_{\AIE}(\pi)\) are therefore best viewed as the marginal-policy, direct, and spillover components of the induced pseudo-true outcome model, rather than as arbitrary contrasts between exposure labels. In this sense, our contribution is not to claim that an arbitrary misspecified exposure mapping is automatically policy-informative, but to characterize the canonical target of exposure-based procedures under misspecification. Section~\ref{sec:error to oracle} then quantifies the distance between these induced objects and their oracle counterparts, while Section~\ref{sec:local vs global example} shows that in a structured local-global environment this interpretation sharpens into channel-specific approximations to the corresponding oracle components.

\subsection{Approximation to oracle estimands}\label{sec:error to oracle}
Instead of directly adopting the pseudo-true outcome models~\eqref{eq:exposure-based outcome proxy}, one can more generally use any collection $f=\{f_i\}_{i\in[n]}$ with $f_i:\{0,1\}^n\to\RR$ as a candidate approximation to the oracle estimands in~\eqref{eq:oracle estimand}. Specifically, define the functionals $\tau_{\MPE}^{\functional}(f;\pi) = \tau_{\ADE}^{\functional}(f;\pi) + \tau_{\AIE}^{\functional}(f;\pi)$ by
\begin{equation}\label{eq:functional estimand}
\begin{aligned}
    \tau_{\ADE}^{\functional}(f;\pi) &= \frac{1}{n}\sum_{i=1}^n\EE_{\bW\sim\RCT(\pi)}\sbr{ f_i\rbr{w_i=1,\bW_{-i}} - f_i\rbr{w_i=0,\bW_{-i}} }, \\
    \tau_{\AIE}^{\functional}(f;\pi) &= \frac{1}{n}\sum_{j=1}^n\sum_{i\neq j}\EE_{\bW\sim\RCT(\pi)}\sbr{ f_j\rbr{w_i=1,\bW_{-i}} - f_j\rbr{w_i=0,\bW_{-i}} }.
\end{aligned}
\end{equation}
The following proposition shows that these functionals are Lipschitz in $f$ under the $L^2$ norm. The same decomposition also holds for the oracle estimand, $\tau_{\MPE}^{\oracle}=\tau_{\ADE}^{\oracle}+\tau_{\AIE}^{\oracle}$, with
\begin{equation}\label{eq:oracle estimand decomposition}
\begin{aligned}
    \tau_{\ADE}^{\oracle}(\pi) &= \frac{1}{n}\sum_{i=1}^n\EE_{\bW\sim\RCT(\pi)}\sbr{ y_i\rbr{w_i=1,\bW_{-i}} - y_i\rbr{w_i=0,\bW_{-i}} }, \\
    \tau_{\AIE}^{\oracle}(\pi) &= \frac{1}{n}\sum_{j=1}^n\sum_{i\neq j}\EE_{\bW\sim\RCT(\pi)}\sbr{ y_j\rbr{w_i=1,\bW_{-i}} - y_j\rbr{w_i=0,\bW_{-i}} }.
\end{aligned}
\end{equation}
\begin{proposition}\label{prop:functional estimand being Lipschitz}
    There exists a constant $C>0$ that depends only on $\pi$ such that, for any collection $f=\{f_i\}_{i\in[n]}$ with $f_i:\{0,1\}^n\to\RR$ and any $\star\in\cbr{\MPE,\ADE,\AIE}$,
    \begin{equation}\label{eq:Lipschitz functional estimand}
        \abr{ \tau_{\star}^{\functional}(f;\pi) - \tau_{\star}^{\oracle}(\pi) }^2 \le C \sum_{i=1}^n \EE_{\bW\sim\RCT(\pi)}\sbr{ f_i(\bW) - y_i(\bW) }^2.
    \end{equation}
    Moreover, taking $y_i(\bw)=\sum_{j=1}^n (w_j-\pi)$ and $f_i(\bw)=2y_i(\bw)$ shows that the bound is tight.
\end{proposition}
The proof of this proposition is deferred to Appendix~\ref{sec:reflecting sq loss criterion}. Thus, the square-loss criterion~\eqref{eq:square loss criterion} searches for the optimal $\{f_i\}_{i\in[n]}$ subject to the compositional restriction $f_i=h_i \odot d_i$ by minimizing the right-hand side of~\eqref{eq:Lipschitz functional estimand}.
Plugging the conditional-expectation formula for $h_i^{\ast}$ in~\eqref{eq:hi-star-def} into~\eqref{eq:Lipschitz functional estimand}, we immediately obtain the following corollary.
\begin{corollary}\label{corollary:approximate sufficiency}
    There exists a constant $C>0$ that depends only on $\pi$ such that,
    \begin{equation}\label{eq:oracle-approx-bound}
        \max_{\star\in\{\MPE,\ADE,\AIE\}} \; \bigl|
          \tau_{\star}
          - \tau_{\star}^{\oracle}
        \bigr|
        \;\le\;
        C \; \bigg\{ \sum_{i\in[n]}
        \EE_{\bW}
        \bigl[
            \Var(y_i(\bW)|d_i(\bW))
        \bigr] \bigg\}^{1/2},
    \end{equation}
    where $\bW\sim\RCT(\pi)$. If the exposure mappings are sufficiently informative that
    \[
        \sum_{i\in[n]}
            \EE_{\bW}
            \bigl[
                \Var(y_i(\bW)|d_i(\bW))
            \bigr]
        \;=\; o\!\left(1\right),
    \]
    then the pseudo-true estimands are asymptotically close to their oracle counterparts.
\end{corollary}
This corollary makes precise in what sense our estimands approximate the oracle targets.
When there is little residual variation in outcomes after conditioning on the exposure
mapping, the pseudo-true direct, indirect, and total effects are necessarily close to
their oracle counterparts. In other words, within the class of outcome models that depend
on treatment only through the chosen exposure mapping, any method that fits individual
outcomes well also delivers a good approximation to the marginal policy effect, and our
pseudo-true model is the optimal such approximation in that class.

\section{Examples of coexisting spillover channels}\label{sec:example expo}

This section provides examples of exposure mappings that fit naturally within the general framework of Section~\ref{sec:misspecified estimand}. Our goal is not to be exhaustive, but to highlight two broad classes of spillover mechanisms that frequently arise in practice: local network spillovers and global equilibrium spillovers. These examples motivate the structured local-global model studied in Section~\ref{sec:local vs global example}, where both channels coexist.

\subsection{Local and global spillovers}

\vspace{2mm}
\noindent\textbf{Local network spillovers.}
A large recent literature studies spillovers that operate through local network structure, including applications in epidemiology, peer effects, spatial externalities, and informal insurance \citep{halloran1991direct,HudgensHalloran2008,lin2024sir,ogburn2024causal,cai2015social,fafchamps2003risk}. In many such settings, units are embedded in a network encoding pairwise relationships such as social ties, geographic proximity, or technological links. Let $E \in \{0,1\}^{n\times n}$ be a symmetric adjacency matrix, where $E_{ij}=1$ indicates that units $i$ and $j$ are connected. A natural exposure mapping is
\(
d_i(\bw) := \{ w_j : j \neq i,\; E_{ij}=1 \},
\)
which records the treatment assignments of unit $i$'s neighbors. In practice, researchers often work with lower-dimensional summaries of this vector, such as the proportion of treated neighbors or an indicator for whether at least one neighbor is treated \citep{li2022random,cai2015social}.

\vspace{2mm}
\noindent\textbf{Global equilibrium spillovers.}
Other forms of interference operate through aggregate or equilibrium mechanisms that affect all units simultaneously, including herd immunity, market-clearing prices, and centralized allocation rules \citep{halloran1991direct,lin2024sir,egger2022ge,munro2025designed,arkhangelsky2025evaluating}. In such settings, an individual’s outcome may depend on the full treatment assignment only through a low-dimensional global state. This motivates exposure mappings of the form
\(
d_i(\bw) := P_n(\bw), \, i \in [n],
\)
where $P_n(\bw)$ is a scalar or low-dimensional summary induced by the full assignment $\bw$, such as an epidemic threshold or an equilibrium price.

\subsection{Multiple coexisting spillovers}\label{sec:local-global-examples}

Our main interest is in environments where these two channels coexist. In such cases, individual outcomes can be written as
\(
y_i(\bw)
=
y_i\bigl(w_i,\, S_i(\bw),\, P_n(\bw)\bigr),
\)
where $w_i$ is an individual treatment, $S_i(\bw)$ is a local network exposure, and $P_n(\bw)$ is a global state induced by the assignment $\bw$.

     \textbf{1. Vaccination on networks with herd immunity.} Consider a susceptible--infected--removed (SIR) model on a contact network. Vaccination of neighbors reduces unit $i$'s infection risk through local transmission channels, which can be summarized by a network exposure such as the fraction of vaccinated neighbors \citep{HudgensHalloran2008}. At the same time, aggregate vaccination levels determine whether the population crosses a herd-immunity threshold, altering infection risk for all individuals through a global channel \citep{halloran1991direct,lin2024sir}. Exposure mappings that focus only on local network structure therefore confound these two mechanisms whenever global epidemic conditions also matter, echoing recent work on misspecified exposure mappings and equilibrium causal estimands \citep{savje2024causal,menzel2025fixed}.
    
     \textbf{2. Market equilibrium with network externalities.} A similar structure arises in market environments with both equilibrium spillovers and local interactions. Individual treatments, such as subsidies or cash transfers, may affect outcomes globally through equilibrium prices determined by market clearing, but also locally through peer effects, information transmission, or technological complementarities.\footnote{\citet{munro2021treatment} write: ``One unit’s treatment impacts another’s outcomes only through the treatment’s impact on the equilibrium price, which rules out peer effects or other forms of network-type interference.'' \citet{munro2025designed} similarly note: ``There are two possible sources of interference from an information treatment; the first is spillovers through the mechanism due to capacity constraints, and the second is network spillovers. The estimates in Table 4 only account for the first type of spillover.''} Local network exposures capture peer interactions holding prices fixed, while global exposures summarize equilibrium adjustments operating at the economy-wide level. Large-scale cash-transfer experiments in Kenya illustrate general-equilibrium spillovers on non-recipients via changes in local demand and prices \citep{egger2022ge}, whereas \cite{AngelucciDeGiorgi2009} show the cash-transfer program in Mexico generates local network externalities through gifts, loans, and informal risk-sharing with little evidence of local price changes.  Related patterns appear in other domains: job-placement programs can raise employment for treated workers but displace untreated job seekers in the same labor markets \citep{crepon2013displacement}, with referrals through social networks mediating access to jobs \citep{beaman2012referral}; and school-choice reforms affect aggregate sorting and housing markets \citep{hsieh2006choice} while classroom peer composition generates local externalities \citep{carrell2010externalities}. These examples underscore that similar interventions can trigger either or both types of spillovers depending on scale and context.

These examples illustrate that interference often arises through multiple, conceptually distinct channels. In the next section, we formalize a structured local-global model, motivated by the market-equilibrium-with-network-externalities example, in which these channels can be analyzed jointly.

\section{Analysis of local and global interference}\label{sec:local vs global example}

We now specialize the general pseudo-true framework of Section~\ref{sec:misspecified estimand} to a structured environment in which two first-order spillover channels coexist: local network interference and global market interference. This setting is motivated by the examples in Section~\ref{sec:local-global-examples} and nests the two well-specified benchmark environments studied separately in \citet{li2022random} and \citet{munro2021treatment}. Our goal is to show that, in this structured model, the pseudo-true objects from Section~\ref{sec:estimands defined via conditioning} admit a sharper channel-specific interpretation, and that existing methods continue to estimate the corresponding local and global components.
\begin{enumerate}

    \item[\hypertarget{NET}{$\mathsf{(NET)}$}]  (Local \textbf{network} interference) With latent variables $Q_i\in\cQ$, the graph is generated by a graphon model $E_{ij}\sim\mathrm{Bernoulli}\rbr{G_n(Q_i,Q_j)}$ independently for $i<j$. In the following, we will let $N_i=\sum_{j\neq i}E_{ij}$ be each unit's degree, and $M_i = \sum_{j\neq i} E_{ij} W_j$ be the total number of a unit's treated neighbors. Then $S_i=M_i/N_i$ represents the proportion of treated neighbors.

    \item[\hypertarget{MAR}{$\mathsf{(MAR)}$}] (Global \textbf{market} interference) Suppose outcomes are also affected by the prices \(p\in\RR^J\) of \(J\) products. For each \(i\in[n]\), let \(z_i(w_i,p)\in\RR^J\) denote unit \(i\)'s excess demand at price \(p\) when assigned treatment \(w_i\in\{0,1\}\). The realized price \(P_n(\bW)\) is determined by approximately solving
\[
\frac{1}{n}\sum_{i=1}^n z_i(W_i,P_n(\bW)+U_i)\approx 0
\]
subject to individualized price perturbations \(U_i\). We additionally observe the realized excess demands \(Z_i=z_i(W_i,P_n(\bW))\) for all individuals.
\end{enumerate}
The environment also generates implicit functionals that output the observed outcomes
\begin{equation}\label{eq: outcome model}
    Y_i=y_i(W_i,S_i(\bW),P_n(\bW)),
\end{equation}

In this section, we show that the oracle marginal policy effect \(\tau_{\MPE}^{\oracle}\) admits an interpretable asymptotic decomposition into direct, local, and global components. These components are defined using the general pseudo-true framework from Section~\ref{sec:estimands defined via conditioning}, specialized here to local and global exposure mappings; see Section~\ref{sec:local vs global example estimand}. Moreover, the existing methods of \citet{li2022random} and \citet{munro2021treatment} continue to provide valid estimators for the corresponding local and global components; see Section~\ref{sec:local vs global example estimator}. Relative to those earlier analyses, our contribution is to show that these methods remain valid in a joint local-global environment, using in particular a higher-order expansion of Z-estimators for the empirical price variable \(P_n(\bW)\).

\subsection{Model setup}
\begin{assumption}\label{assump:super-population for units}
    We assume that each $(Q_i,z_i,y_i)$ is drawn independently from the same joint distribution for all $i\in[n]$.
\end{assumption}

\begin{remark}
    One could relax \eqref{eq: outcome model} and allow \(y_i\) to depend more generally on the full assignment vector \(\bW\). By Corollary~\ref{corollary:approximate sufficiency}, our results should continue to hold provided that: (i) the joint exposures \((S_i(\bW),P_n(\bW))\) are approximately sufficient in the sense that
    \begin{equation}\label{eq:approximate sufficiency}
        \sum_{i\in[n]}
            \EE_{\bW}
            \bigl[
                \Var(y_i(\bW)\mid W_i,S_i(\bW),P_n(\bW))
            \bigr]
        \;=\; o(1),
    \end{equation}
    and (ii) the corresponding pseudo-true outcome models
    \[
        \tilde y_i(\bw;\pi)
        =
        \E\!\left[
            y_i(\bW)
            \mid W_i,
            S_i(\bW)=S_i(\bw),\,
            P_n(\bW)=P_n(\bw)
        \right]
    \]
    satisfy the remaining assumptions of the section. We impose \eqref{eq: outcome model} in order to state the main results under cleaner conditions. The extension under \eqref{eq:approximate sufficiency} is conceptually straightforward and omitted for brevity.
\end{remark}
We now define several population quantities explicitly. Let $p_\pi^{\ast}$ be the unique population-clearing price, as in \citet[Assumption 3]{munro2021treatment}, which solves
\begin{equation}\label{eq:population-clearing price}
    \EE\sbr{\pi z_i(1,p_{\pi}^{\ast}) + (1-\pi) z_i(0,p_{\pi}^{\ast})} = 0.
\end{equation}
In addition, define the \emph{population} gradients which are evaluated at $p_{\pi}^{\ast}$,
\begin{align}
    \xi_z &:= \mathbb{E}\sbr{\pi \nabla_p z_i(1,p_{\pi}^{\ast}) + (1-\pi) \nabla_p z_i(0,p_{\pi}^{\ast})} \in\RR^{J\times J}, \\
    \xi_y &:= \mathbb{E}\sbr{\pi \nabla_p y_i(1,\pi,p_{\pi}^{\ast}) + (1-\pi) \nabla_p y_i(0,\pi,p_{\pi}^{\ast})} \in\RR^{J}.
\end{align}
We also impose the structural condition that the graphon model in \spillovertype{NET} is of low rank. The statistical network-analysis literature has studied the spectral decay of sparse graphon models \citep{gao2015rate,chen2025minimax}.
\begin{condition}[Sparse and low-rank graphon sequence]
\label{assump:graphon}
    Assume $G_n(u,v)=\min\cbr{1,\rho_nG(u,v)}$ for some fixed non-negative symmetric bi-variate function $G$. We further require $\rho_n=cn^{-\kappa}$ for some fixed $\frac{1}{3}<\kappa<\frac{1}{2}$. We also assume the graphon model to be low rank: for some $r\ge 1$, there exists
    \begin{equation}
        G(Q_i,Q_j)=\sum_{k=1}^r \lambda_k\psi_k(Q_i)\psi_k(Q_j),
    \end{equation}
    such that
    \begin{equation}
        |\lambda_1|\ge|\lambda_2|\ge\cdots\ge|\lambda_r|>0,\quad\EE\left[\psi_k(Q_i)^2\right]=1,
    \end{equation}
    and $\EE\left[\psi_k(Q_i)\psi_l(Q_i)\right]=0$ for any $k\neq l$. We write
    \(
        g(q):=\EE_{Q'}[G(q,Q')],
    \)
    where \(Q'\) is an independent draw from the same distribution as \(Q_i\). Finally, we assume that the expected connectivity is uniformly bounded away from zero and infinity: there exist constants \(0<c_g\le C_g<\infty\) such that \(c_g\le g(q)\le C_g\) for all \(q\in\cQ\). This ensures that unit degrees diverge uniformly with high probability, so that degree-normalized quantities such as \(S_i=M_i/N_i\) and \(1/g(Q_i)\) are well behaved.
\end{condition}
In addition to the graphon structure in \spillovertype{NET}, we assume access to the
same augmented randomized trial that provides instrumental variables-like variation in
\spillovertype{MAR}.
\begin{condition}[Augmented randomized trial]\label{assump: augmented trial}
    In addition to generating the treatments, the experimenter can generate
    individualized perturbations $U_i\in\RR^J$ to the global equilibrium factor
    $P_n(\bW)$. These perturbations satisfy
    $U_{ij}\,\indsim\,\mathrm{Unif}(\{\pm h_n\})$ for $h_n = c n^{-\alpha}$ with
    $\frac{1}{4}<\alpha<\frac{1}{2}$. 
    Then the actual global factor $P_n(\bW)$ is defined by solving
    \begin{equation}
    \frac{1}{n}\sum_{i=1}^n z_{ij}(W_i,P_n(\bW)+U_i) \approx 0,\quad \forall j\in[J].
    \end{equation}
    The notation ``$\approx 0$'' is formalized in Assumption~\ref{assump:empirical price}. 
\end{condition}

We treat $(\rho_n,h_n)$ jointly as parameters of the environment.

\subsection{Treatment effect estimands}\label{sec:local vs global example estimand}

Consistent with Section~\ref{sec:misspecified estimand}, our primary object remains the oracle marginal policy effect
\begin{equation}\label{eq: define total effect}
    \tau_{\MPE}^{\oracle} = \frac{1}{n} \sum_{i=1}^n \sum_{j=1}^n \EE_{\bW\sim\RCT(\pi)} \sbr{y_j(w_i=1;\bW_{-i})-y_j(w_i=0;\bW_{-i})}.
\end{equation}
This notion is well-defined regardless of the interference structure. In addition, we consider three related components.
\begin{itemize}
    \item[(i)] \textbf{(Oracle) Direct effect:} This estimand characterizes how each $w_i$ affects its own outcome,
    \begin{equation}\label{eq: define direct effect}
        \tau_{\ADE}^{\oracle} = \frac{1}{n}\sum_{i=1}^n \EE_{\pi}\sbr{y_i(w_i=1;\bW_{-i})-y_i(w_i=0;\bW_{-i})}.
    \end{equation}

    \item[(ii)] \textbf{Local spillover effect:} This estimand characterizes how each $w_i$ affects the outcomes of that unit's neighbors. We will be choosing $d_i^{\rmL}(\bw) = \{ w_{j}: j\in[n], j\neq i, E_{ij}=1 \}$, which induces
    \begin{align}
        \tilde{y}_i^{\rmL}(\bw;\pi) &= \EE_{\bW^{(2)}\sim\mathrm{RCT}(\pi)}\cbr{ y_i\rbr{ \bW^{(2)} } \Big| W^{(2)}_{j}=w_{j},\quad\forall j\text{ with }E_{ij}=1 }, \\
        &= \EE_{\bW^{(2)}}\cbr{ y_i\rbr{ w_i, \frac{\sum_{j\neq i}E_{ij}w_j}{\sum_{j\neq i}E_{ij}}, P_n\rbr{\bw_{\cN_i},\bW^{(2)}_{-\cN_i}} }  }, \\
        \tau_{\AIE}^{\rmL}(\pi) &= \frac{1}{n}\sum_{j=1}^n\sum_{i\neq j}\EE_{\bW\sim\RCT(\pi)}\sbr{ \tilde{y}_j^{\rmL}\rbr{w_i=1,\bW_{-i};\pi} - \tilde{y}_j^{\rmL}\rbr{w_i=0,\bW_{-i};\pi} } \\
        &= \frac{1}{n}\sum_{j=1}^n\sum_{i\in\cN_ j}\EE_{\bW\sim\RCT(\pi)}\sbr{ \tilde{y}_j^{\rmL}\rbr{w_i=1,\bW_{-i};\pi} - \tilde{y}_j^{\rmL}\rbr{w_i=0,\bW_{-i};\pi} }
        \label{eq:local spillover estimand}
    \end{align}

    \item[(iii)] \textbf{Global spillover effect:} This estimand characterizes how each $w_i$ affects the outcomes of all units through the equilibrium mechanism. The exposure mapping is simply the equilibrium variable $d_i^{\rmG}(\bw) = P_n(\bw)$ for every unit $i\in[n]$. This choice induces the following objects:
    \begin{align}\label{eq:global spillover estimand}
        \tilde{y}_i^{\rmG}(\bw;\pi) &= \EE_{\bW^{(2)}\sim\mathrm{RCT}(\pi)}\cbr{ y_i\rbr{ \bW^{(2)} } \Big| P_n\rbr{\bW^{(2)}} \approx P_n(\bw) }, \\
        &= \EE_{\bW^{(2)}}\cbr{ y_i\rbr{ w_i,\frac{\sum_{j\neq i}E_{ij}W_j^{(2)}}{\sum_{j\neq i}E_{ij}}, P_n\rbr{\bw} } \Big| P_n\rbr{\bW^{(2)}} \approx P_n(\bw) },\\
        \tau_{\AIE}^{\rmG}(\pi) &= \frac{1}{n}\sum_{j=1}^n\sum_{i\neq j}\EE_{\bW\sim\RCT(\pi)}\sbr{ \tilde{ y}_j^{\rmG}\rbr{w_i=1,\bW_{-i};\pi} - \tilde{y}_j^{\rmG}\rbr{w_i=0,\bW_{-i};\pi} }.
    \end{align}
    We will formalize the meaning of ``$\approx$" later in Section~\ref{sec: global spillover estimand proof}.
\end{itemize}
In finite samples, there is no reason to expect \(\tau_{\MPE}^{\oracle}\) to decompose exactly into these three components. Our next result shows that such a decomposition turns out to emerge asymptotically.
\begin{theorem}\label{thm: asymptotic limit of estimands}
    As $n\to\infty$, all finite-sample estimands converge in probability to population limits:
    \begin{align}
        \tau_{\ADE}^{\oracle} &\pto \tau_{\ADE}^{\oracle,\ast} := \EE\sbr{y_i(1,\pi,p_{\pi}^{\ast})-y_i(0,\pi,p_{\pi}^{\ast})}, \\
        \tau_{\AIE}^{\rmL} &\pto \tau_{\AIE}^{\rmL,\ast} := \EE\sbr{\pi \nabla_s y_i(1,\pi,p_{\pi}^{\ast})+(1-\pi)\nabla_s y_i(0,\pi,p_{\pi}^{\ast})}, \\
        \tau_{\AIE}^{\rmG} &\pto \tau_{\AIE}^{\rmG,\ast} := - \xi_y^{\top}\xi_z^{-1} \EE\sbr{z_i(1,p_{\pi}^{\ast})-z_i(0,p_{\pi}^{\ast})}.
    \end{align}
    Moreover, the total treatment effect also has a finite asymptotic limit, which is the sum of the three components above:
    \begin{equation}\label{eq:decomposition of MPE into 3 parts}
        \tau_{\MPE}^{\oracle} \pto \tau_{\MPE}^{\oracle,\ast} := \tau_{\ADE}^{\oracle,\ast} + \tau_{\AIE}^{\rmL,\ast} + \tau_{\AIE}^{\rmG,\ast}.
    \end{equation}
    The convergence in probability here is with respect to all sources of randomness, including the draws of the latent functions $\{(y_i,z_i):i\in[n]\}$ and the random network $\bE$.
\end{theorem}


Within the model \eqref{eq: outcome model}, the local and global spillover channels are asymptotically decoupled. Intuitively, this holds because the global and local channels correspond to fluctuations of the assignment vector in very different directions. Specifically, the global spillover channel operates through a low-dimensional, ``consensus'' statistic of the assignment, while the local channel operates through high-dimensional ego exposures; under the Bernoulli design these directions fluctuate at order $n^{-1/2}$ and are asymptotically uncorrelated, so only the separate local and global components contribute to the welfare derivative at first order, and their interaction is second order.
Appendix~\ref{sec:app-th4.5} presents the proof in several steps.

It is also worth noting that \(\tau_{\AIE}^{\rmL,\ast}\) coincides with the estimand in \citet{li2022random} when the equilibrium price is fixed at \(p_\pi^\ast\), whereas \(\tau_{\AIE}^{\rmG,\ast}\) coincides with the estimand in \citet{munro2021treatment} when local interference is fixed at its benchmark level \(\pi\).

\subsection{Estimators}\label{sec:local vs global example estimator}
After defining several notions of treatment effects, this section presents corresponding estimators that are consistent for the asymptotic estimands. We also derive sharp convergence rates to assess the statistical efficiency of the proposed methods.

To recap the basic setup, we observe a network $\bE\in\{0,1\}^{n \times n}$, a randomized assignment $\bW\in\{0,1\}^n$ generated according to Assumption~\ref{assump:RCT}, and individualized perturbations $\bU\in\RR^{n \times J}$ generated according to Assumption~\ref{assump: augmented trial}. We then observe realized outcomes $\bY\in\RR^{n}$ and excess demands $\bZ\in\RR^{n \times J}$. Any valid estimator must be constructed only from these observables.

\medskip\noindent\textbf{Direct effect.} 
Consistent with the practice in \cite{li2022random,munro2021treatment}, we employ the usual Horvitz--Thompson estimator for $\tau_{\mathrm{ADE}}^{\oracle}$, which is automatically unbiased under the RCT design (Assumption~\ref{assump:RCT}):
\begin{equation}\label{eq:HT_estimator}
    \hat{\tau}_{\mathrm{ADE}}^{\oracle}
      =\frac{1}{n}\sum_{i=1}^n
         \rbr{
            \frac{W_i }{\pi}
            -\frac{1-W_i}{1-\pi}
         }Y_i .
\end{equation}
Because our model contains two distinct spillover mechanisms, the asymptotic variance of this estimator differs from the standard benchmark. We derive it explicitly in the following theorem.
\begin{theorem}\label{thm:direct effect estimator}
    Under assumptions detailed in Section~\ref{sec:assump}, the Horvitz-Thompson estimator $\hat{\tau}_{\ADE}$ has a limiting Gaussian distribution around the asymptotic average direct effect estimand $\tau_{\ADE}^{\oracle,\ast}$,
\begin{equation}
    \sqrt{n}\rbr{ \hat{\tau}_{\mathrm{ADE}}^{\oracle}-\tau_{\ADE}^{\oracle,\ast} }
    \Rightarrow
    \cN\!\left(
        0,\,
        \Var\sbr{\mathcal{V}^{(1)}+\mathcal{V}^{(2)}} + \pi(1-\pi)\EE\sbr{\rbr{V^{(1)}+V^{(2)}+V^{(3)}}^2}
    \right)
\end{equation}
where we can represent the limiting variance through
\begin{align}
    \mathcal{V}^{(1)} &= y_1(1,\pi,p_\pi^\ast) - y_1(0,\pi,p_\pi^\ast), \\
    \mathcal{V}^{(2)} &= -\nabla_p\sbr{\EE y(1,\pi,p_\pi^\ast)-\EE y(0,\pi,p_\pi^\ast)}^\top \xi_z^{-1} \sbr{\pi z_1(1,p_\pi^\ast)+(1-\pi)z_1(0,p_\pi^\ast)}, \\
    V^{(1)} &= \frac{y_1(1,\pi,p_\pi^\ast)}{\pi}+\frac{y_1(0,\pi,p_\pi^\ast)}{1-\pi}, \\
    V^{(2)} &= \EE_{Q_2,y_2}\sbr{\frac{G(Q_1,Q_2)\sbr{\nabla_s y_2(1,\pi,p_\pi^\ast) - \nabla_s y_2(0,\pi,p_\pi^\ast)}}{g(Q_2)}\bigg|Q_1}, \\
    V^{(3)} &= -\nabla_p\sbr{\EE y(1,\pi,p_\pi^\ast)-\EE y(0,\pi,p_\pi^\ast)}^\top \xi_z^{-1} \sbr{z_1(1,p_\pi^\ast)-z_1(0,p_\pi^\ast)}.
\end{align}
\end{theorem}

\medskip\noindent\textbf{Local spillover effect.} 
Because we adopt the same setup as \citet{li2022random} for local interference, it is natural to use the same estimator. Start by forming a vector $\bnu\in\RR^n$ of raw weights:
\begin{equation}
    \nu_i = \frac{M_i}{\pi}-\frac{N_i-M_i}{1-\pi} = \sum_{j\in\cN_i} \rbr{ \frac{W_j}{\pi}-\frac{1-W_j}{1-\pi} }.
\end{equation}
Compute $\hat{\bPsi}\in\RR^{n\times r}$ as the normalized top-$r$ eigenvectors of the observed adjacency matrix $\bE=(E_{ij})$ with $\hat{\bPsi}^\top\hat{\bPsi}=\Ib_r$. The PC-balancing estimator is then defined by
\begin{equation}
    \hat{\tau}_{\AIE}^{\rmL} = \frac{1}{n} \bnu^{\top} \rbr{\Ib_n-\hat{\bPsi}\hat{\bPsi}^\top}\bY \in\RR.
\end{equation}
The following remark explains the intuition behind this estimator.
\begin{remark}
    The empirical average $\bnu^\top\bY/n$ is already a natural estimator, since its expectation
    \begin{align}
        \EE_{\bW}\sbr{\frac{1}{n}\sum_{i=1}^n \nu_iY_i} = \frac{1}{n}\sum_{j=1}^n\sum_{i\in\cN_ j}\EE_{\bW}\sbr{ y_j\rbr{w_i=1,\bW_{-i}} - y_j\rbr{w_i=0,\bW_{-i}} },
    \end{align}
    is already quite close to our estimand $\tau_{\AIE}^{\rmL}$ in~\eqref{eq:local spillover estimand}. It differs from $\tau_{\AIE}^{\rmL}$ only because it replaces the conditioned $\tilde{y}_j^{\rmL}$ with the unconditioned $y_j$. In \citet{li2022random}, where only local interference is present, this estimator is indeed unbiased.
    
    However, projecting both $\bY$ and $\bnu$ onto the graphon principal components $\bPsi=\cbr{\psi_k(Q_i)}_{i,k}\in\RR^{n\times r}$ yields a pathological term $\bnu^{\top} \bPsi\bPsi^{\top} \bY/n$, because $\bnu^{\top} \bPsi$ has nonzero mean whereas $\bPsi^{\top} \bY/n$ has exploding variance. Section 4.2 of \citet{li2022random} illustrates this weakness using a stochastic block model. The authors therefore propose the methodology above, which mitigates this issue by projecting onto the subspace orthogonal to $\hat{\bPsi}$ (as a proxy for $\bPsi$).
\end{remark}


Departing from the existing theory in \citet{li2022random}, the next theorem deepens our understanding of the PC-balancing estimator $\hat{\tau}_{\AIE}^{\rmL}$. It shows that the estimator is robust to additional unspecified market interference \spillovertype{MAR}. In particular, it still targets $\tau_{\AIE}^{\rmL,\ast}$ with the same convergence rate. The limiting variance is also similar to that in \citet{li2022random} and is therefore omitted from the main text.

\begin{theorem}\label{thm:local spillover effect estimator informal}
    Under assumptions detailed in Section~\ref{sec:assump}, the PC-balancing estimator $\hat{\tau}_{\AIE}^{\rmL}$ has a limiting Gaussian distribution around the asymptotic local spillover estimand $\tau_{\AIE}^{\rmL,\ast}$,
    \begin{equation}
        \frac{1}{\sqrt{\rho_n}} \rbr{ \hat{\tau}_{\AIE}^{\rmL} - \tau_{\AIE}^{\rmL,\ast} } \Rightarrow \cN\rbr{0, \mathsf{V}_{\rmL}},
    \end{equation}
    where the variance $\mathsf{V}_{\rmL}$ is given in Section~\ref{sec: estimator proof local}.
\end{theorem}

\medskip\noindent\textbf{Global spillover effect.}
As shown in Theorem~\ref{thm: asymptotic limit of estimands}, the global spillover estimand $\tau_{\AIE}^{\rmG}$ depends asymptotically on the price-elasticity vector $\gamma:=\xi_z^{-\top}\xi_y$ and the direct effect $\tau_z:=\EE\sbr{z_i(1,p_{\pi}^{\ast})-z_i(0,p_{\pi}^{\ast})}$ on excess demand. We can therefore estimate these two objects separately and then combine them to obtain a valid estimator of $\tau_{\AIE}^{\rmG,\ast}=-\gamma^{\top}\tau_z$.

Price elasticities have long been a central topic in econometrics \citep{houthakker1969income,chetty2009sufficient}. They can be estimated using instrumental variables \citep{angrist1996identification,berry2021foundations}. For conciseness, we follow the approach of \citet{munro2021treatment}, in which the experimenter creates instrumental variables by augmenting the experimental design with individualized price perturbations. The construction is stated formally in Condition~\ref{assump: augmented trial}.

Equipped with $\bU\in\{\pm h_n\}^{n \times J}$, we estimate price elasticities by
\(
    \hat{\gamma} = \rbr{\bU^\top\bZ}^{-1}\rbr{\bU^\top\bY}.
\)
After constructing a Horvitz–Thompson estimator for the treatment effect of excess demands
\(    \hat{\tau}_z = \frac{1}{n}\sum_{i=1}^n \rbr{\frac{W_i}{\pi}-\frac{1-W_i}{1-\pi}}Z_i,
\)
the final estimator is $\hat{\tau}_{\AIE}^{\rmG}=-\hat{\gamma}^{\top} \hat{\tau}_z$. Our next theorem studies the theoretical performance of this estimator. It consistently targets the asymptotic limit $\tau_{\AIE}^{\rmG,\ast}$, with the same convergence rate, even under unspecified local network interference \spillovertype{NET}.

\begin{theorem}\label{thm:global spillover effect estimator informal}
    Under assumptions detailed in Section~\ref{sec:assump}, the estimator $\hat{\tau}_{\AIE}^{\rmG}$ has a limiting Gaussian distribution around the asymptotic global spillover estimand $\tau_{\AIE}^{\rmG,\ast}$,
    \begin{equation}
        h_n\sqrt{n} \rbr{ \hat{\tau}_{\AIE}^{\rmG} - \tau_{\AIE}^{\rmG,\ast} } \Rightarrow \cN\rbr{0, \mathsf{V}_{\rmG}},
    \end{equation}
    where the variance $\mathsf{V}_{\rmG}$ is given in Section~\ref{sec: estimator proof global}.
\end{theorem}

Taken together, Theorems~\ref{thm:local spillover effect estimator informal} and~\ref{thm:global spillover effect estimator informal} show that the PC-balancing estimator from \cite{li2022random} and the augmented-trial estimator in \cite{munro2021treatment} remain asymptotically valid in the full local-global environment. Each consistently recovers the corresponding local or global component of $\tau_{\MPE}^{\oracle,\ast}$ singled out by our decomposition, even though the underlying exposure mapping omits the other first-order channel.

Moreover, by Theorem~\ref{thm: asymptotic limit of estimands}, we can consistently estimate the oracle marginal policy effect by summing the corresponding estimators,
\begin{equation}
    \hat{\tau}_{\MPE}^{\oracle}
    =
    \hat{\tau}_{\ADE}^{\oracle}
    +
    \hat{\tau}_{\AIE}^{\rmL}
    +
    \hat{\tau}_{\AIE}^{\rmG}.
\end{equation}
Since each component estimator is consistent, it follows that $\hat{\tau}_{\MPE}^{\oracle}\pto\tau_{\MPE}^{\oracle,\ast}$ as $n\to\infty$. Its convergence rate, however, is governed by the slowest component estimator.

\begin{corollary}
    The convergence rate of $\hat{\tau}_{\MPE}^{\oracle}$ depends on whether $\kappa+2\alpha$ is greater or less than $1$. If $\kappa+2\alpha<1$, then the local AIE estimator dominates the error and $\hat{\tau}_{\MPE}^{\oracle}$ converges to $\tau_{\MPE}^{\oracle,\ast}$ at rate $n^{-\kappa/2}$; if $\kappa+2\alpha>1$, then the global AIE estimator dominates the error and $\hat{\tau}_{\MPE}^{\oracle}$ converges to $\tau_{\MPE}^{\oracle,\ast}$ at rate $n^{-1/2+\alpha}$.
\end{corollary}

\section{Numerical study}\label{sec:gen_simu}

\subsection{Simulation example: A fixed-index model}\label{sec:simulation}
Our first simulation setup is a fixed-index model, where the outcome \(Y_i\) of each unit ultimately depends on one aggregated index \(\eta_i\). With \(\bw\in\{0,1\}^n\) being a potential assignment, the detailed model generating process is given as below.

\begin{description}
    \item[(a)] \textbf{Local network:} For any \(i\neq j\), their connection \(E_{ij}\) is drawn independently from \(\mathrm{Bern}(\rho)\). For each unit \(i\), write \(N_i(\bw)=\sum_{j\neq i}E_{ij}\) and \(M_i(\bw)=\sum_{j\neq i}E_{ij}w_j\), and define the proportion of treated neighbors as \(S_i(\bw)=M_i(\bw)/\max\{1,N_i(\bw)\}\).

    \item[(b)] \textbf{Excess demand functional:} Suppose that there is only \(J=1\) product under consideration. The excess demand functional is \(z_i(w_i,p)=(1-\theta_p w_i)-p\). Solving \(\sum_{i=1}^n z_i(w_i,p)=0\) yields the equilibrium price \(P_n(\bw)=1-\theta_p n^{-1}\sum_{i=1}^n w_i\).

    \item[(c)] \textbf{Linear implicit index:} Define \(\eta_i(\bw)=\theta_w w_i+(1-u)\theta_\ell S_i(\bw)+u\theta_g P_n(\bw)\), where \(u\in[0,1]\) is a mixing parameter that interpolates between local and global spillovers. Lastly, through a (possibly non-linear) link function \(g\), the environment outputs \(y_i(\bw)=g(\eta_i(\bw))\).
\end{description}

Henceforth, this model is described by linear coefficients $(\theta_p,\theta_\ell,\theta_g,\theta_w)=(0.5,0.5,0.8,1)$, parameters \((\rho,u)\) and a link function \(g(\cdot)\). We consider five canonical choices of \(g\), including the linear link \(g(x)=x\), a quadratic link \(g(x)=x+x^2\), a cosine link \(g(x)=\cos(x)\), a logarithmic link \(g(x)=\log(1+x^2)\), and a higher-order polynomial link \(g(x)=x+x^2+x^3\). These are denoted as \(\{\mathtt{linear},\,\mathtt{quad},\,\mathtt{cos},\,\mathtt{log},\,\mathtt{poly}\}\) later. To carry out experiments, we additionally choose the treatment assigning rate \(\pi\), individualized price perturbation size \(h\), and the rank \(r\) in the PC-balancing step.

\medskip\noindent\textbf{Monte Carlo experiments with fixed sample size.} Throughout this part, we set \(n=1000\). When computing the estimators, we always use individualized perturbation size \(h=0.1\) and correctly specified rank \(r=1\). Figure~\ref{fig:linear-index finite-sample} depicts the performance of our estimators with the \(\mathtt{linear}\) link function, and varying \((u,\pi,\rho)\). In each panel we report the average direct effect (ADE) together with the local and global components of the average indirect effect (AIE): solid curves show Monte Carlo averages, and dashed curves show the corresponding oracle quantities.

\begin{figure}
  \centering

  \begin{subfigure}{.32\textwidth}
    \centering
    \includegraphics[width=\linewidth]{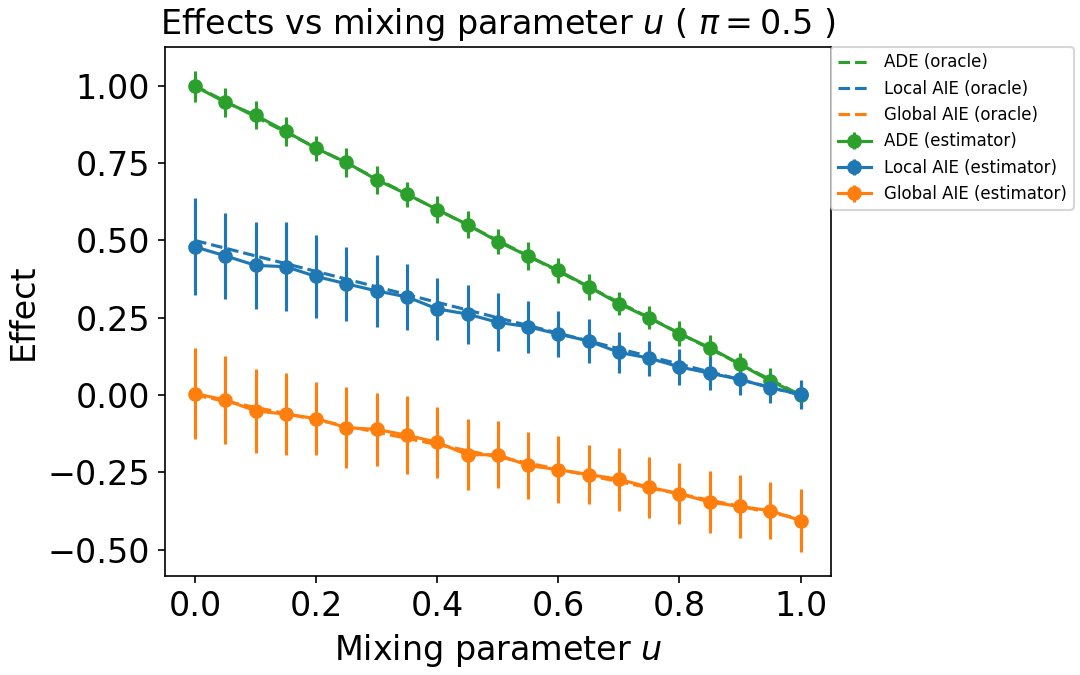}
    \caption{varying mixing parameter \(u\) (\(\pi=0.5,\,\rho=0.01\)).}
    \label{fig:lin-u}
  \end{subfigure}
  \hfill
  \begin{subfigure}{.32\textwidth}
    \centering
    \includegraphics[width=\linewidth]{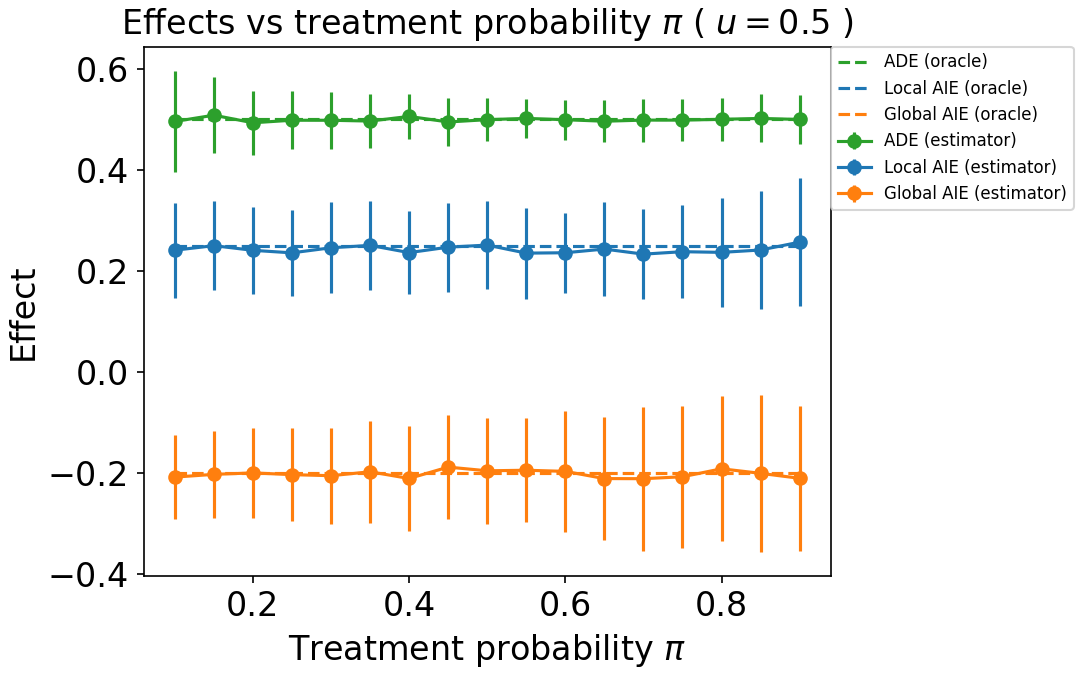}
    \caption{varying assignment rate \(\pi\) (\(u=0.5,\,\rho=0.01\)).}
    \label{fig:lin-pi}
  \end{subfigure}
  \hfill
  \begin{subfigure}{.32\textwidth}
    \centering
    \includegraphics[width=\linewidth]{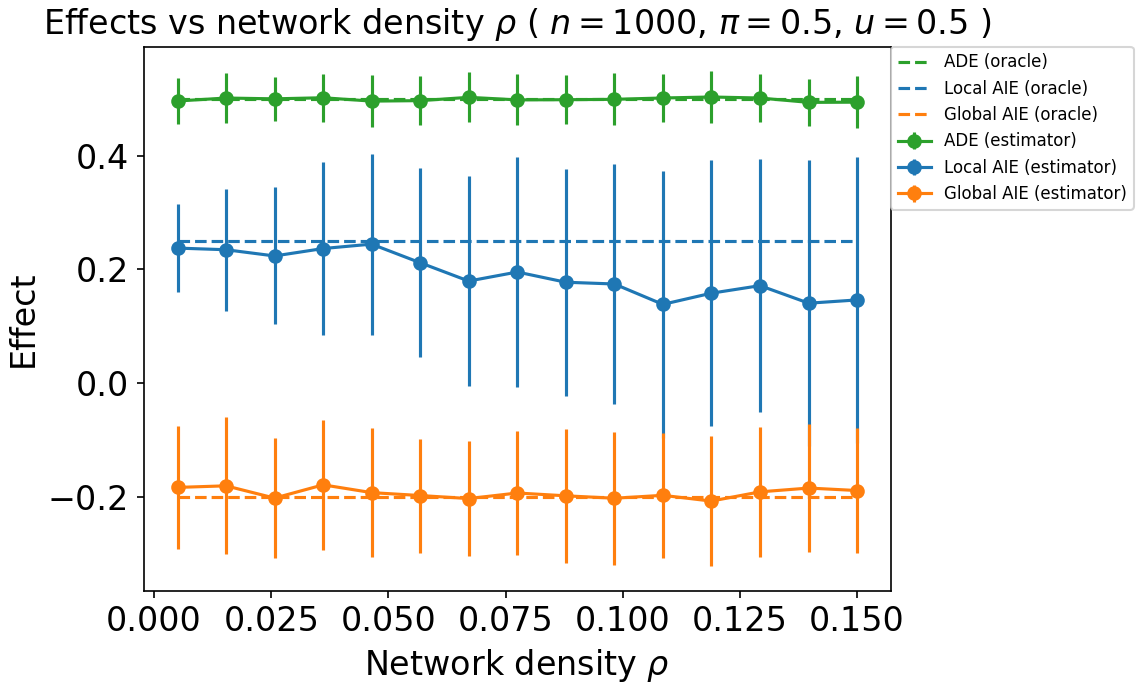}
    \caption{varying ER density \(\rho\) (\(\pi=u=0.5\)).}
    \label{fig:lin-density}
  \end{subfigure}

  \caption{Index outcome model with \(\mathtt{linear}\) link function and fixed sample size \(n=1000\).
  Dashed curves indicate oracle ADE / local AIE / global AIE, while solid curves are Monte Carlo averages.}
  \label{fig:linear-index finite-sample}
\end{figure}

Apart from the case of \(\mathtt{linear}\) link function, Figure~\ref{fig:nonlinear-index finite-sample} shows the performance of our estimators in finite samples with several \emph{non-linear} link functions. This time we set \(\pi=0.5\) and \(\rho=0.01\) throughout, and only vary the mixing parameter \(u\). For each choice of link in \(\{\mathtt{quad},\mathtt{cos},\mathtt{log},\mathtt{poly}\}\), the figure displays ADE, local AIE, and global AIE separately as functions of \(u\); dashed curves indicate oracle values, and solid curves indicate Monte Carlo averages.

\begin{figure}[p]
  \centering

  \begin{subfigure}{.40\linewidth}
    \centering
    \includegraphics[width=\linewidth]{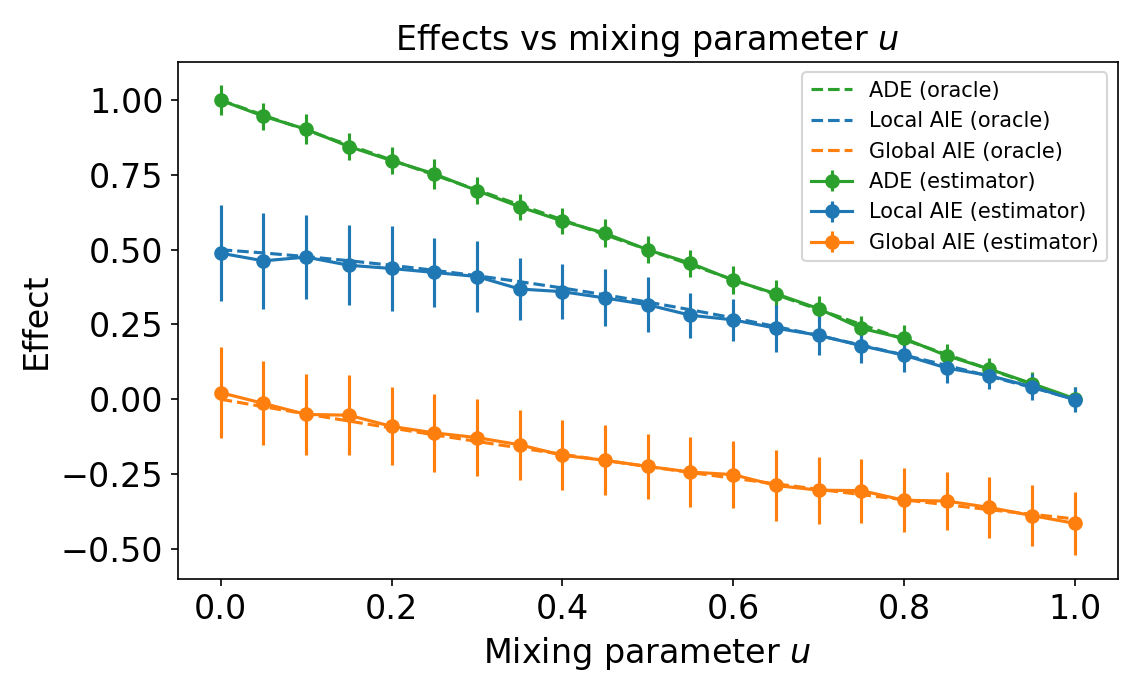}
    \caption{\(\mathtt{quad}\) link.}
    \label{fig:mpe-u-quadratic}
  \end{subfigure}
  \hfill
  \begin{subfigure}{.40\linewidth}
    \centering
    \includegraphics[width=\linewidth]{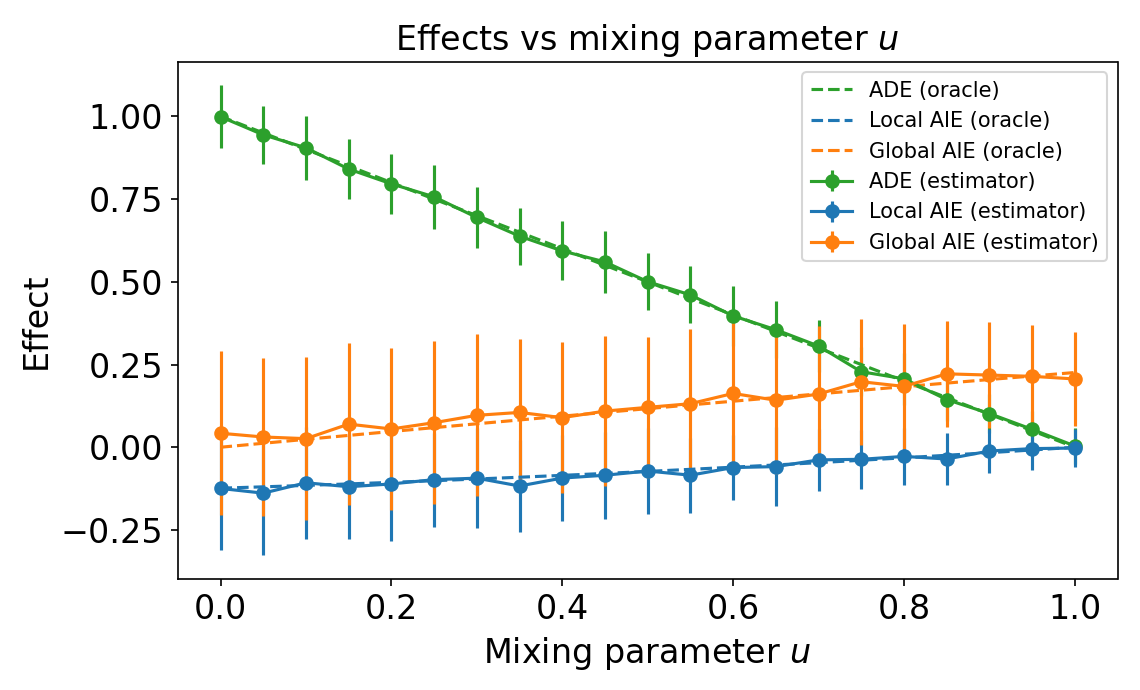}
    \caption{\(\mathtt{cos}\) link.}
    \label{fig:mpe-u-cos}
  \end{subfigure}

  \vspace{0.2em}

  \begin{subfigure}{.40\linewidth}
    \centering
    \includegraphics[width=\linewidth]{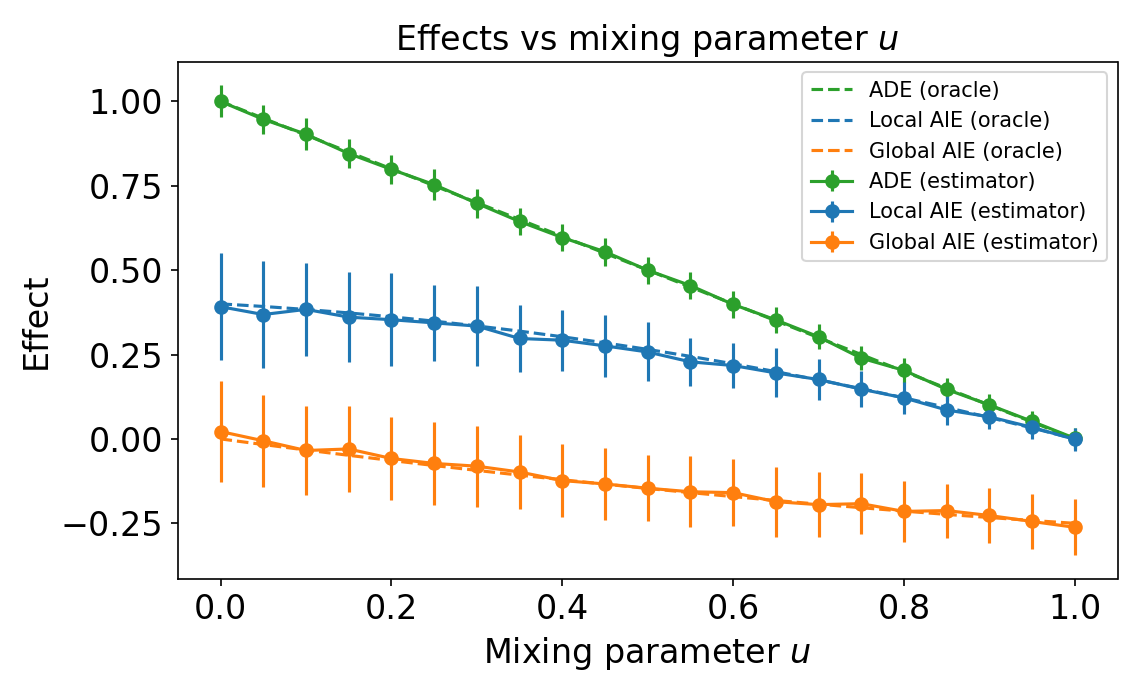}
    \caption{\(\mathtt{log}\) link.}
    \label{fig:mpe-u-log}
  \end{subfigure}
  \hfill
  \begin{subfigure}{.40\linewidth}
    \centering
    \includegraphics[width=\linewidth]{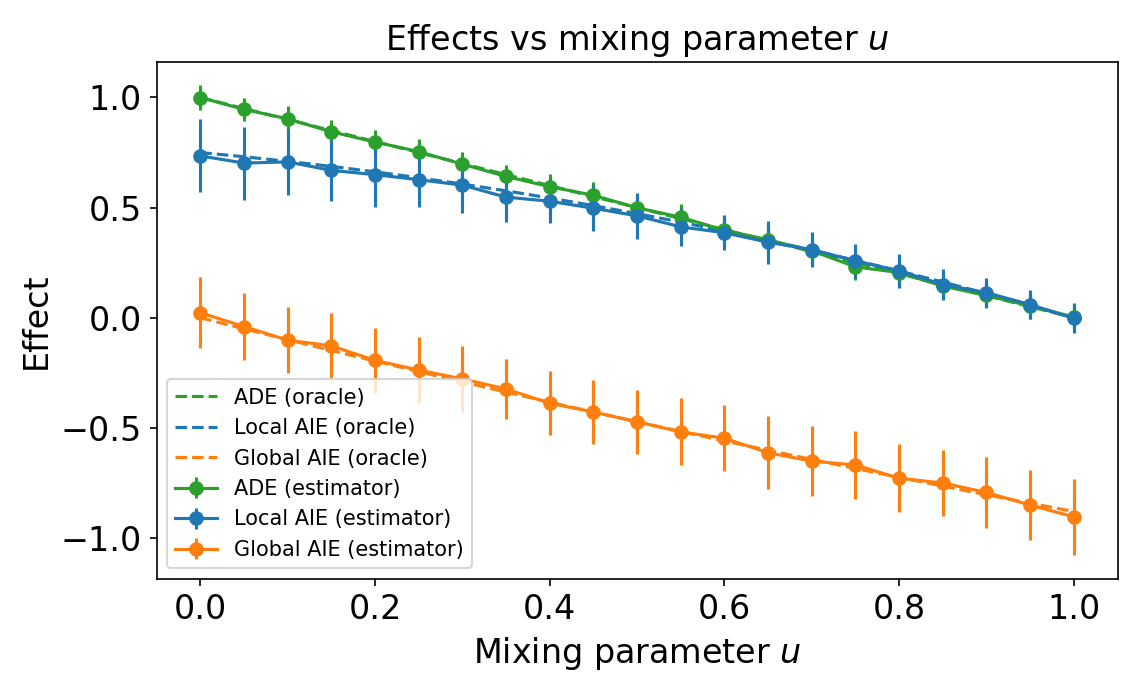}
    \caption{\(\mathtt{poly}\) link.}
    \label{fig:mpe-u-poly}
  \end{subfigure}

  \caption{Index outcome model with non-linear link functions and fixed sample size \(n=1000\).
  ADE, local AIE, and global AIE are plotted as functions of the mixing parameter \(u\). Dashed curves indicate oracle values, while solid curves are Monte Carlo averages.}
  \label{fig:nonlinear-index finite-sample}
\end{figure}

All the experiments so far suggest that our estimators can approximate the limiting estimands well enough in finite samples.

\medskip\noindent\textbf{Monte Carlo experiments of growing sample size.} Now we increase the magnitude of \(n\) to numerically check the asymptotic convergence rates shown in Section~\ref{sec:local vs global example estimator}. In Figure~\ref{fig:linear-index growing n}, we take \(n\in\{100,\ldots,10000\}\) with \(h_n=0.75\,n^{-\alpha}\) and \(\rho_n=0.75\,n^{-\kappa}\). The pair \((\kappa,\alpha)\) takes values in \(\{(0.34,0.26),\,(0.49,0.40)\}\). We set \(u=\pi=0.5\) and plot the MSE of our ADE estimator and the local and global AIE estimators against \(n\) in log-log scale. We consider the linear DGP and the cosine DGP.

\begin{figure}[!htbp]
  \centering
  \begin{subfigure}{.40\linewidth}
    \centering
    \includegraphics[width=\linewidth]{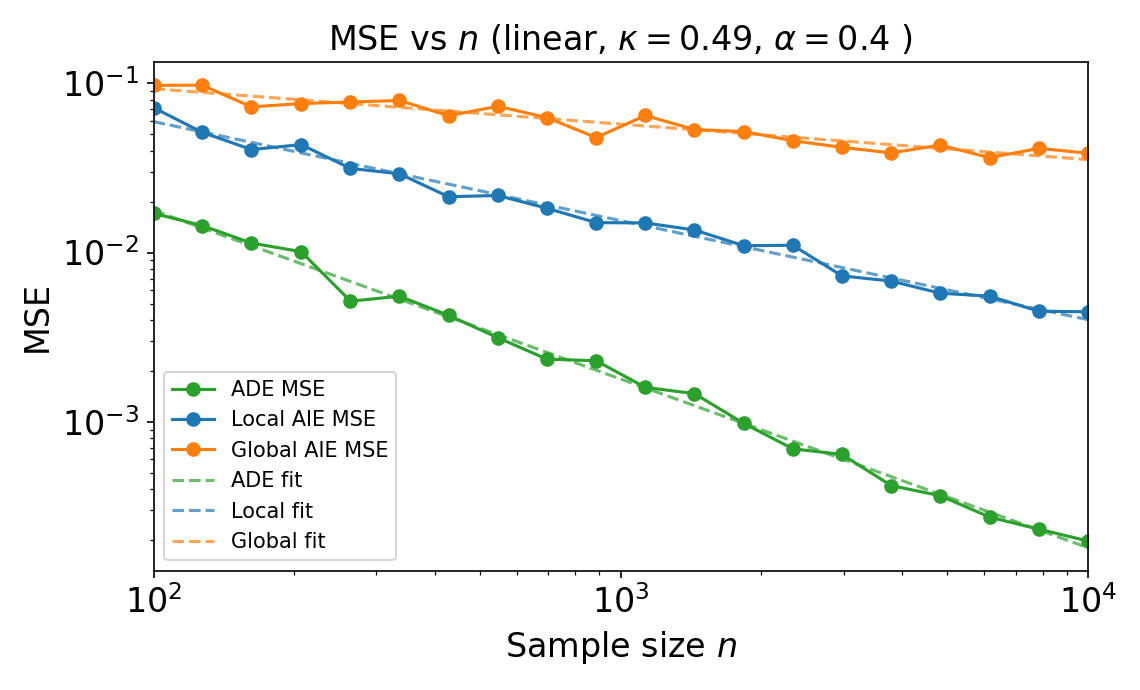}
    \caption{\((\kappa,\alpha)=(0.49,0.40)\).}
  \end{subfigure}
  \hfill
  \begin{subfigure}{.40\linewidth}
    \centering
    \includegraphics[width=\linewidth]{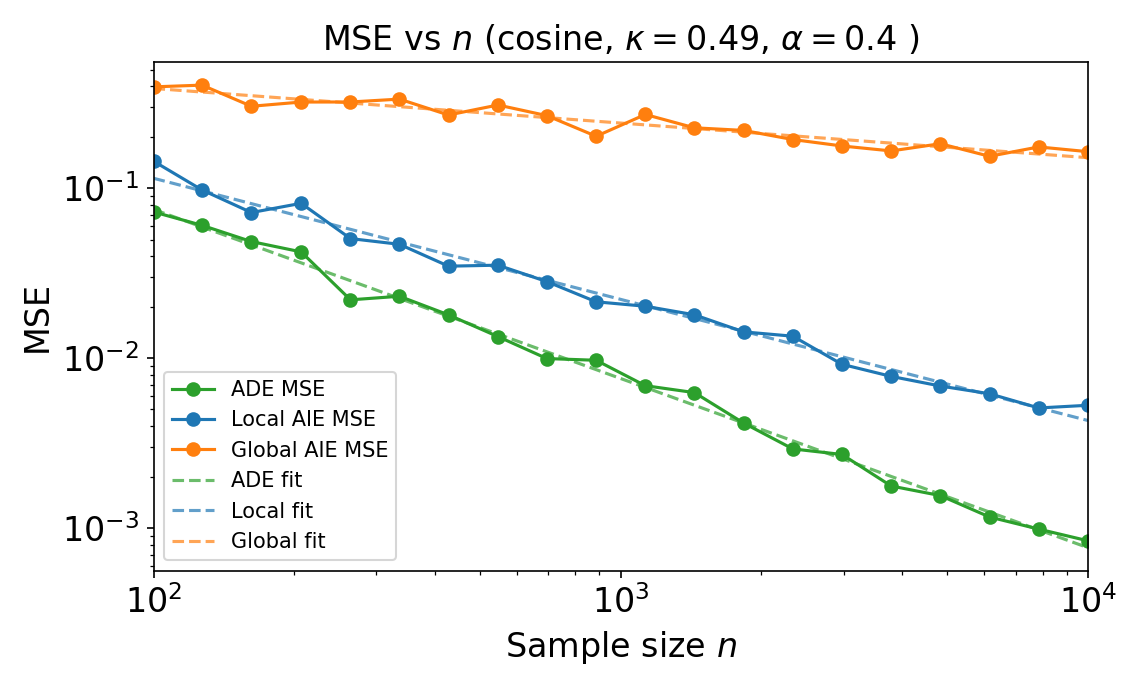}
    \caption{\((\kappa,\alpha)=(0.49,0.40)\).}
  \end{subfigure}

  \vspace{0.2em}

  \begin{subfigure}{.40\linewidth}
    \centering
    \includegraphics[width=\linewidth]{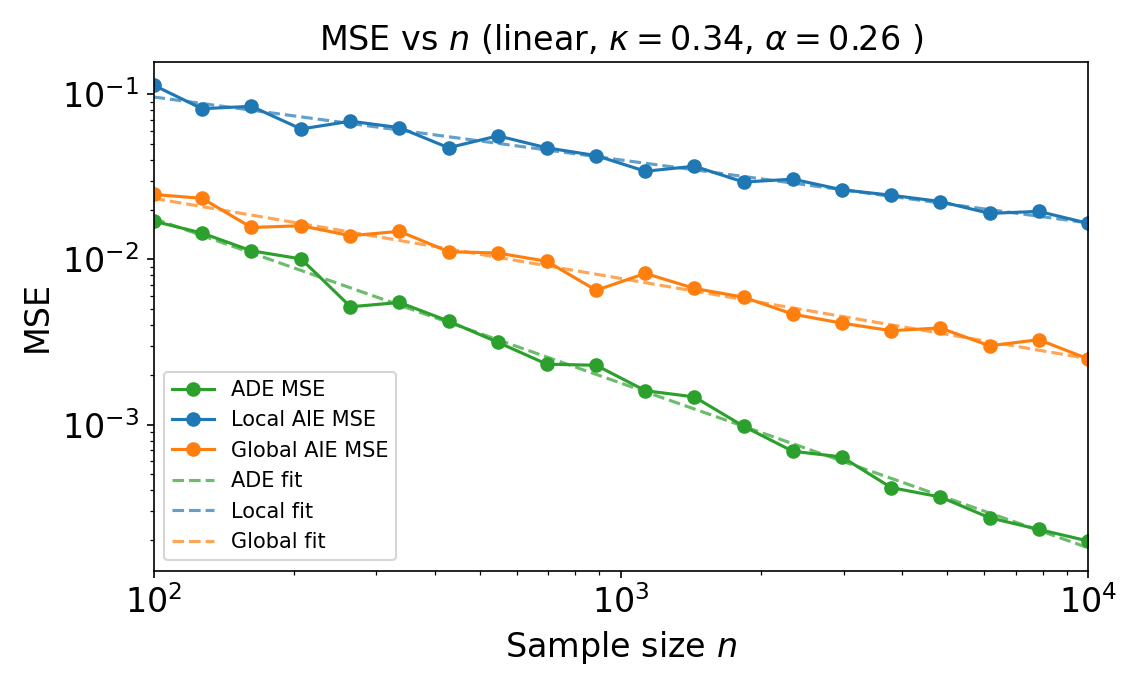}
    \caption{\((\kappa,\alpha)=(0.34,0.26)\).}
  \end{subfigure}
  \hfill
  \begin{subfigure}{.40\linewidth}
    \centering
    \includegraphics[width=\linewidth]{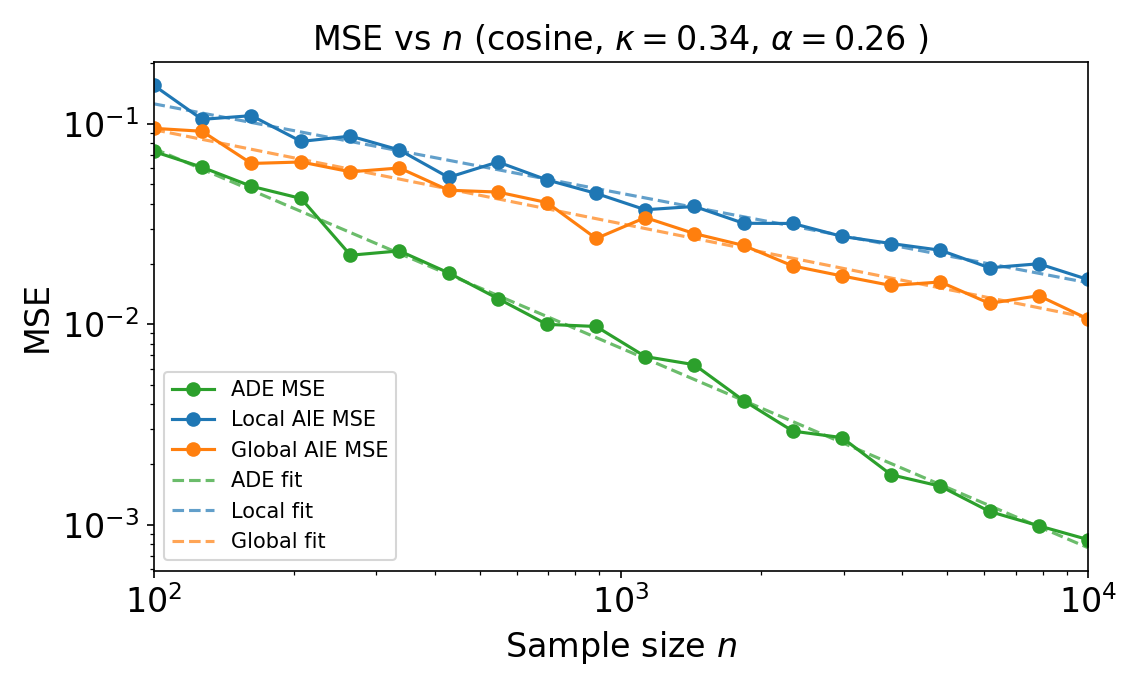}
    \caption{\((\kappa,\alpha)=(0.34,0.26)\).}
  \end{subfigure}

  \caption{MSE of the ADE and the local and global AIE estimators against \(n\) in log-log scale, across two \((\kappa,\alpha)\) regimes for linear and nonlinear (cosine) outcome DGP.}
  \label{fig:linear-index growing n}
\end{figure}




\subsection{Semi-synthetic application: cash transfers in a village economy}\label{sec:filmer-sim}

We next consider a semi-synthetic design calibrated to the cash-transfer experiment in Philippine villages studied by \cite{filmer2023cash}. A previous work by \citet{munro2021treatment} studies the spillover effect caused by a \emph{global} market-equilibrium channel operating through village-level egg prices.  We extend their calibration by further adding a \emph{local} network channel operated through household neighbors, so that the marginal policy effect can be decomposed into direct,
local, and global components within a realistic structural environment.

In the experiments of \cite{filmer2023cash}, each household (of several individuals) is randomized to receive a cash transfer with probability $\pi$ at the
household level, which exactly satisfies our Assumption~\ref{assump:RCT}. The target outcome is the height-for-age Z-score for children between 0 and 5 years of age.

To conduct our synthetic experiments, a concrete parametric model is specified for the outcomes, demands/supplies and network. Then we use the real dataset from \cite{filmer2023cash} to estimate the involved parameters. Subsequently, fixing $\pi$ at the level of the empirical experiments, we are able to simulate multiple experiments and thus assess the ability of our estimators. \cite{munro2021treatment} adopts exactly the same semi-synthetic paradigm for their calibrated evaluation.

\medskip\noindent\textbf{Parametric model.} Besides the following brief description, Appendix~\ref{app:filmer-details} provides additional details.

\medskip\noindent\emph{(a) Market equilibrium model.} We use the same setup as in \cite{munro2021treatment} which incorporates the market of eggs into consideration. Conditional on treatment, each individual has linear demand and supply schedules.  The equilibrium price is then determined by market clearing.

\medskip\noindent\emph{(b) Block network model.}
To introduce local spillovers, we build a network between households whose link probabilities combine a geography-based block component (barangay and municipality codes), homophily in housing and socioeconomic characteristics (roof and wall quality, education, assets, water/sanitation, school-age presence, income, and household size), and a triadic-closure adjustment.  The probability matrix is rescaled to match a target density $\rho$ before drawing edges. Related network constructions appear in \cite{fafchamps2003risk} as well.

\medskip\noindent\emph{(c) Final outcome model.}
The outcome of each child depends linearly on eligibility of the household for cash transfer, own treatment, a scalar
equilibrium price of eggs, and the share of treated neighbor
households. Then the outcome of a house is the average of its children.

\medskip\noindent\textbf{Fitting approach.}
Following \cite{munro2021treatment}, all the coefficients relating to the market are set equal to the structural estimates from the remote-village subsamples in
\cite{filmer2023cash}. To do so, $10$ moments are computed from the real dataset. Additionally, the coefficient of local exposure is calibrated from a
partial-linear regression of child height-for-age on the share of treated neighbors,
controlling for own treatment and village prices.  A detailed description can be found in Appendix~\ref{app:filmer-details}.

\medskip\noindent\textbf{Simulation procedure.}
For each Monte Carlo
replication, we: (i) subsample $2{,}000$ households and draw household treatments from the Bernoulli design, (ii) resample the
household network from the categorical features of the subsampled households, (iii) solve for the
market-clearing price and apply a small scalar perturbation, (iv) simulate individual demand,
supply, and child outcomes under the calibrated structural system, and (v) aggregate to
household-level excess demand and child outcomes.  We then compute the Horvitz--Thompson
estimator of the average direct effect, the PC-balanced estimator of the local component,
and the augmented-trial IV estimator of the global component.
This process is repeated for $2{,}000$ times.
Closed-form expressions
for the corresponding oracle targets
$\tau_{\ADE}^\ast,\tau_{\AIE}^{\rmL,\ast},\tau_{\AIE}^{\rmG,\ast}$ are
derived in Appendix~\ref{app:filmer-details}.

\medskip\noindent\textbf{Results.}
Table~\ref{tab:decomp-filmer} summarizes our findings.  The “Truth” column
reports the analytical limits from the structural model; the remaining columns report
Monte Carlo means, biases, and standard deviations of the estimators.  Appendix
Figure~\ref{fig:filmer-hist} provides histograms of the Monte Carlo sampling
distributions for the ADE, local AIE, and global AIE estimators.

\begin{table}[H]
\centering
\begin{tabular}{lrrrr}
\toprule
\textbf{Estimator} & \textbf{Truth} & \textbf{Mean} & \textbf{Bias} & \textbf{SD} \\
\midrule
ADE           & 0.3514   & 0.3151   & $-0.0363$ & 0.1522 \\
AIE (Local)   & $-1.0871$ & $-1.0732$ & $0.0139$ & 0.5731 \\
AIE (Global)  & $-0.1333$ & $-0.1320$ & $0.0013$ & 0.0973 \\
\bottomrule
\end{tabular}
\caption{Decomposition of the marginal policy effect in the Filmer-calibrated
semi-synthetic design.  ``Truth'' reports the analytical population targets; the other
columns report Monte Carlo means, biases, and standard deviations across
$2{,}000$ replications. Appendix~\ref{app:filmer-details} has a histogram showing the distribution of these Monte Carlo repetitions.}
\label{tab:decomp-filmer}
\end{table}

The three components are well recovered by their corresponding estimators: Monte Carlo
means lie close to the oracle targets.  In this calibration, the local component is sizeable and
negative, reflecting that treated neighbors reduce the marginal gains from one’s own
transfer, while the global component captures the additional negative effect of
equilibrium price changes.  The decomposition shows that looking only at the average
direct effect would obscure these indirect channels, even in a setting closely matched to
a real cash-transfer experiment.

\medskip\noindent\textbf{Econometric insights.}
Ex ante, the sign
of $\tau_{\mathrm{TOT}}^{\ast}$ is ambiguous. On the one hand, transfers can relax liquidity constraints and
increase gifts and informal loans within the network, raising non-labor income and consumption
for both treated and untreated households \citep[e.g.][]{fafchamps2003risk}.
On the other hand, higher transfer intensity can generate congestion in non-priced local amenities
and adverse price or equilibrium effects in thin markets, so that the indirect components
$\tau_{\AIE}^{\rmL,\ast} $ and $\tau_{\AIE}^{\rmG,\ast} $ may be negative and potentially dominate the direct
gain. Recent work on large-scale public programs documents that such general-equilibrium
spillovers can be first-order and even larger than the direct effect
\citep{muralidharan2023general}, while market-based models of status and consumption in
networks show that equilibrium adjustments can overturn the sign of aggregate welfare effects
\citep{ghiglino2010keeping}. Our calibration should therefore be viewed as one
plausible configuration in which $\tau_{\mathrm{ADE}}^{\ast} > 0$ but $\tau_{\mathrm{TOT}}^{\ast} < 0$ because negative
local and global spillovers dominate the positive direct effect.
\section{Conclusion and future directions}\label{sec:conclusion}
This paper develops a framework for interpreting exposure-based estimands under interference when exposure mappings may be misspecified, and for relating those estimands to underlying policy objects. For any analyst-chosen exposure mapping, we define a pseudo-true outcome model as the best mean-squared approximation to the true potential-outcome model within the class of models that depend on treatment only through that mapping. This induces corresponding pseudo-true direct, spillover, and marginal policy effects, and the usual decomposition of the marginal policy effect into direct and spillover components continues to hold exactly. When the chosen exposure mapping is sufficiently informative so that residual outcome variation after conditioning on exposure is asymptotically negligible, these pseudo-true estimands are asymptotically close to their oracle counterparts.

We then specialize this framework to a structured environment with both local network and global equilibrium spillovers. In this setting, the oracle marginal policy effect admits an asymptotic decomposition into direct, local spillover, and global spillover components. This sharpens the interpretation of misspecified exposure mappings: procedures based only on local exposures can still be viewed as targeting the local component, while procedures based only on global exposures target the global component, even when each omits the other first-order channel. More generally, even when additional channels are omitted, the induced estimands remain the best exposure-based \(L^2\) approximations to the corresponding oracle policy effects, and are close to those oracle targets whenever the remaining omitted variation is asymptotically negligible.

Taken together, our results suggest a disciplined way to think about decomposition of spillover effects under misspecification. They should not be understood only as a source of misspecification, but also as a disciplined way to isolate interpretable channels of policy transmission and to define well-defined approximations when the full interference structure is too complex to model directly.

Interesting directions for future work include broadening the analysis from local marginal policy changes to larger design shifts and other policy objects \citep[e.g.,][]{munro2025designed} and developing experimental designs that better allocate power across direct, local, and global spillover channels \citep{UganderEtAl2013}.

\bibliographystyle{apalike}
\bibliography{reference}

\newpage

\appendix
\tableofcontents

\section{Additional details for the cash-transfer calibration}
\label{app:filmer-details}

\subsection{Environment and treatment assignment}
\label{subsec:filmer-setup}

There are $n$ individuals indexed by $i\in[n]$, grouped into $n_h$ households
$h\in[n_h]$.  Let $h:[n]\to[n_h]$ map each individual to its household $h(i)$.  Treatment
is assigned at the \emph{household} level: for each $h$, the indicator
$W_h\in\{0,1\}$ denotes whether household $h$ receives the cash transfer.  We collect the
assignments into
\[
  \bw = (w_1,\ldots,w_{n_h}) \in \{0,1\}^{n_h},
  \qquad
  \bW = (W_1,\ldots,W_{n_h}) \sim \RCT(\pi),
\]
where $\RCT(\pi)$ is the Bernoulli randomized design with
$\PP(W_h=1)=\pi$ independently across households.

Households differ in program eligibility.  Let $E_h\in\{0,1\}$ indicate whether household
$h$ is eligible according to the original study’s targeting rule, and let
$\mu_{\mathrm{eli}} := \EE[E_h]$ denote the population share of eligible households.

\subsection{Demand, supply, excess demand, outcomes, and aggregation}

\paragraph{Individual demand, supply, and excess demand.}
For an individual $i$ in household $h(i)$, demand and supply for eggs are specified as
linear functions of the household treatment $w_{h(i)}$ and a scalar price $p$:
\begin{align}
  \mathsf{demand}_i(w_{h(i)}, p)
  &=
  \theta_{d01} E_{h(i)}
  +
  \theta_{d00} (1-E_{h(i)})
  +
  \theta_{dw} w_{h(i)} E_{h(i)}
  +
  \theta_{dp} p
  +
  \epsilon_{d,h(i)}
  +
  \nu_{d,i}(w_{h(i)}),
  \\
  \mathsf{supply}_i(w_{h(i)}, p)
  &=
  \theta_{s0}
  +
  \theta_{sp} p,
  \\
  z_i(w_{h(i)}, p)
  &=
  \mathsf{demand}_i(w_{h(i)}, p)
  -
  \mathsf{supply}_i(w_{h(i)}, p),
\end{align}
where $\epsilon_{d,h}$ is a household-level demand shock and $\nu_{d,i}(w_{h(i)})$ is an
idiosyncratic disturbance.  The quantity $z_i(w_{h(i)},p)$ is the individual excess
demand, matching the notation of Section~\ref{sec:local vs global example} with a
one-dimensional global state ($J=1$).

\paragraph{Household network and local exposure.}
In addition to the trading market on eggs, spillovers also operate through a network on households.  Let
\[
  \bE
  =
  \bigl(E_{hh'}\bigr)_{h,h'\in[n_h]}
  \in \{0,1\}^{n_h\times n_h}
\]
be an undirected adjacency matrix, so $E_{hh'}=E_{h'h}=1$ if households $h$
and $h'$ are neighbors.  Two individuals $i,j$ are considered network neighbors whenever
their households are linked, i.e.\ $E_{h(i)h(j)}=1$. For each individual $i$, we define the local exposure as the treated share among
neighboring households:
\begin{equation}
    S_i(\bw)
    =
    \frac{
    \sum_{j\neq i} E_{h(i)h(j)}\,w_{h(j)}
    }{
    \sum_{j\neq i} E_{h(i)h(j)}
    },
    \qquad
    i\in[n].
    \label{eq:filmer-local-exposure}
\end{equation}
The next section discusses in detail how the network is generated.

\paragraph{Individual outcomes.}
Child outcomes depend linearly on eligibility, demand, own treatment, the local exposure, and
household-level outcome shocks:
\begin{align}
    y_i\bigl(w_{h(i)}, s_i(\bw), p\bigr)
    &=
    \theta_{y01} E_{h(i)}
    +
    \theta_{y00} (1-E_{h(i)})
    +
    \theta_{yd}\,\mathsf{demand}_i\bigl(w_{h(i)}, p\bigr)
    +
    \theta_{yw} w_{h(i)}
    \nonumber\\
    &\quad
    +
    \theta_{ys} S_i(\bw)
    +
    \epsilon_{y,h(i)}
    +
    \nu_{y,i}(w_{h(i)}), \label{eq:real data outcome model}
\end{align}
where $\epsilon_{y,h}$ and $\nu_{y,i}(w_{h(i)})$ are random household- and individual-level
outcome noise.  In the simulation we draw all noise terms independently as mean-zero
Gaussians with standard deviations
$(\sigma_{d1},\sigma_{d0},\sigma_{y1},\sigma_{y0},\sigma_{dh},\sigma_{yh})
=(1/3,\,1/3,\,1,\,1,\,1/3,\,1)$ for
$(\epsilon_{d,i}(1),\epsilon_{d,i}(0),\epsilon_{y,i}(1),\epsilon_{y,i}(0),\epsilon_{d,h},\epsilon_{y,h})$.

\paragraph{Aggregation to households.}
Let $A^h\subset[n]$ be the set of all members of household $h$ and $C^h\subset[n]$ the
subset of children used in the outcome analysis.  We define household-level excess demand
and child outcomes as
\begin{align}
  \underline{z}_h(w_h, p)
  &=
  \frac{n_h}{n}
  \sum_{i\in A^h}
  z_i(w_h, p),
  \\
  \underline{y}_h(w_h, p)
  &=
  \frac{n_h}{n_c}
  \sum_{i\in C^h}
  y_i\bigl(w_h, s_i(\bw), p\bigr),
\end{align}
where $n_c=\sum_{h=1}^{n_h}|C^h|$ is the total number of children.  In the simulation we
treat $\underline{z}_h$ and $\underline{y}_h$ as household-level variables and apply the
estimators of Section~\ref{sec:local vs global example estimator} at the household level.

In total, our model is parametrized by $11$ parameters to be fitted from the real dataset,
\begin{equation}\label{eq:params in real-data model}
    (\theta_{d01},\theta_{d00},\theta_{dw},\theta_{dp},\theta_{s0},\theta_{sp},\theta_{y01},\theta_{y00},\theta_{yd},\theta_{yw},\theta_{ys}).
\end{equation}
In comparison to \citet{munro2021treatment}, our model differs in having $\theta_{ys}$ as an extra parameter, that appears as a coefficient in~\eqref{eq:real data outcome model} before the local exposure $S_i(\bw)$. 

Table~\ref{tab:filmer-params} reports the calibrated parameter values used in the
Monte Carlo experiments, computed from the remote-village moments and the
network-based estimate of $\theta_{ys}$ with the baseline network specification.
\begin{table}[H]
\centering
\begin{tabular}{lr}
\toprule
Parameter & Value \\
\midrule
$\theta_{d01}$ & 3.8870 \\
$\theta_{d00}$ & 4.4875 \\
$\theta_{dw}$  & 0.1896 \\
$\theta_{dp}$  & $-0.3764$ \\
$\theta_{s0}$  & $-0.2053$ \\
$\theta_{sp}$  & 0.3263 \\
$\theta_{y01}$ & $-5.4339$ \\
$\theta_{y00}$ & $-5.5117$ \\
$\theta_{yd}$  & 1.8496 \\
$\theta_{yw}$  & 0.1025 \\
$\theta_{ys}$  & $-1.0871$ \\
\bottomrule
\end{tabular}
\caption{Calibrated parameter values used in the Filmer-based simulations.}
\label{tab:filmer-params}
\end{table}

\subsection{Household network construction from housing characteristics}
\label{subsec:filmer-network}

The original dataset also reports housing characteristics for each household, including roof and
wall quality.  Each household $h$ has a roof status in
\texttt{\{`no-roof', `lightroof', `strongroof'\}} and a wall status in
\texttt{\{`no-wall', `lightwall', `strongwall'\}}.  We encode these as ordinal variables
\[
  \mathsf{roof}_h \in \{0,1,2\},
  \qquad
  \mathsf{wall}_h \in \{0,1,2\},
  \qquad h\in[n_h].
\]

We generate the household network $\bE$ by combining a block component based on
barangay/municipality codes with homophily layers based on housing and socioeconomic
covariates.  Specifically, we build a nonnegative score matrix
\[
  S_{hh'}
  =
  S^{\text{block}}_{hh'}
  +
  \sum_{k} w_k \,\widetilde S^{(k)}_{hh'},
\]
where $S^{\text{block}}_{hh'}=w_{\text{cross}}+w_{\text{mun}}\mathbf{1}\{m_h=m_{h'}\}
w_{\text{bgy}}\mathbf{1}\{b_h=b_{h'}\}$ and each similarity layer $\widetilde S^{(k)}$
(roof, wall, education, assets, water/sanitation, school-age, income, and household size)
is normalized to have off-diagonal mean one.  We optionally apply a triadic-closure
mixing $S\leftarrow (1-\lambda_{\text{tc}})S+\lambda_{\text{tc}}\widetilde S_{\text{tc}}$
with $\widetilde S_{\text{tc}}$ the normalized product $SS$.  Finally we rescale to a
target density $\rho$ and draw edges independently with
\[
  \PP(E_{hh'}=1)=\min\{1,\ \rho\,S_{hh'}/\overline S\},
\]
where $\overline S$ is the off-diagonal mean of $S$.
In the baseline calibration we set
$w_{\text{bgy}}=5.0$, $w_{\text{mun}}=2.0$, $w_{\text{cross}}=0.2$,
$w_{\text{roof}}=1.0$, $w_{\text{wall}}=1.0$, $w_{\text{edu}}=1.0$,
$w_{\text{assets}}=0.8$, $w_{\text{watersan}}=0.7$, $w_{\text{school}}=0.6$,
$w_{\text{income}}=1.0$, $w_{\text{hhsize}}=0.6$,
$(\sigma_{\text{income}},\sigma_{\text{hhsize}},\sigma_{\text{edu}})=(0.7,1.2,1.0)$,
and $\lambda_{\text{tc}}=0.5$, and we target density $\rho=0.02$.

To assess sensitivity to $\rho$, we reran the Monte Carlo experiment with
$\rho\in\{0.01,0.02,0.03\}$ while holding all other calibration choices fixed.
Table~\ref{tab:filmer-rho-sens} reports Monte Carlo means and standard deviations
(200 replications, $n_h=2{,}000$).  The ADE and global AIE are stable across these
densities, while the local AIE varies in magnitude, reflecting the dependence of
local exposure on network density.
\begin{table}[H]
\centering
\scalebox{0.8}{
\begin{tabular}{lrrrrrr}
\toprule
$\rho$ & ADE (mean) & Local AIE (mean) & Global AIE (mean) & ADE (SD) & Local AIE (SD) & Global AIE (SD) \\
\midrule
0.01 & 0.3313 & $-0.9235$ & $-0.1348$ & 0.1496 & 0.4037 & 0.0949 \\
0.02 & 0.3319 & $-1.0111$ & $-0.1347$ & 0.1527 & 0.5644 & 0.0970 \\
0.03 & 0.3308 & $-0.4990$ & $-0.1351$ & 0.1398 & 0.6280 & 0.0884 \\
\bottomrule
\end{tabular}
}
\caption{Sensitivity of Monte Carlo estimates to the target network density $\rho$.}
\label{tab:filmer-rho-sens}
\end{table}

\paragraph{Network construction and covariates.}
A large empirical literature on informal risk sharing in village economies finds that network links
exhibit strong homophily along geography and socio-economic status: most gifts, transfers, and
informal loans occur among neighbors and relatives in the same (or adjacent) villages, and among
households with similar wealth and occupations \citep{fafchamps2003risk}.
Motivated by this evidence, we construct a parsimonious index model for network formation in
which the probability of a tie depends on (i) fine geographic location (the village identifier), (ii)
housing quality (roof and wall materials) as a proxy for wealth and social status, and (iii), in
robustness checks, the education of the household head. In the code, these variables enter a
low-rank stochastic block model that induces homophily in location and socio-economic status,
which is standard in empirical work where the true social network is unobserved but the covariates
governing homophily can be proxied from survey data.

Because the roof and wall covariates span a $3\times3$ grid in each dimension, the
resulting network model is approximately rank~6.  We therefore set the number of principal
components in the PC-balancing estimator to $r=6$ in this calibration.

\subsection{Structural calibration of parameters}

Following \cite{filmer2023cash} and \cite{munro2021treatment}, we estimate the first $10$ parameters in~\eqref{eq:params in real-data model},
\[
  \bigl(
    \theta_{d00}, \theta_{d01}, \theta_{dw}, \theta_{dp},
    \theta_{s0}, \theta_{sp},
    \theta_{y00}, \theta_{y01}, \theta_{yd}, \theta_{yw}
  \bigr)
\]
by matching model-implied moments to sample moments from the subsample of remote villages.
The moments include: mean demand and child outcomes by eligibility status, equilibrium
prices in control and treatment villages, and average treatment effects on demand and
outcomes.  Because the demand, supply, and outcome equations are linear, the resulting
estimators have closed forms that mirror those reported in \cite{munro2021treatment}; we
use these closed-form expressions in the code.  The additional local-exposure term enters
only the outcome equation and does not alter the demand or supply moments, so the closed
forms for the first ten parameters are unchanged.

To incorporate local network spillovers, we augment the structural system with
$\theta_{ys}$, the coefficient on $s_i(\bw)$, and estimate it from the regression
\[
  Y_i = \alpha_0 + \alpha_1 D_i + \theta_{ys} S_i + \alpha_3 P_i + \varepsilon_i,
\]
where $Y_i$ is child height-for-age, $D_i$ is the individual treatment indicator, $S_i$ is
the share of treated neighbors in $\bE$, and $P_i=\log p_{b(i)}$ is the
barangay-level log egg price.  Imposing the moment conditions
\[
  \EE[D_i\varepsilon_i]
  =
  \EE[P_i\varepsilon_i]
  =
  \EE[S_i\varepsilon_i]
  =
  0
\]
yields a method-of-moments estimator for $\theta_{ys}$ in terms of sample covariances
$\hat\sigma_{AB}=\widehat{\Cov}(A_i,B_i)$:
\[
  \hat\theta_{ys}
  =
  \frac{
    \hat\sigma_{YS}
    -
    \hat\sigma_{ZS}^{\top}\hat\Sigma_{ZZ}^{-1}\hat\sigma_{ZY}
  }{
    \hat\sigma_{SS}
    -
    \hat\sigma_{ZS}^{\top}\hat\Sigma_{ZZ}^{-1}\hat\sigma_{ZS}
  },
\]
where $Z_i=(D_i,P_i)^\top$, $\hat\Sigma_{ZZ}$ is the covariance matrix of $Z_i$, and
$\hat\sigma_{ZS},\hat\sigma_{ZY}$ are the corresponding cross-covariances.  In the
simulation we treat $\theta_{ys}=\hat\theta_{ys}$ as fixed at its estimated value, using
a single network draw from the housing-covariate model and holding it fixed across
Monte Carlo replications.

\subsection{Ground truth under the parametric model}
\label{subsec:filmer-ground-truth}

Given a household-randomized policy $\bW\sim\RCT(\pi)$, the equilibrium price $p$ solves
the market-clearing condition
\[
  \EE_{\bW\sim\RCT(\pi)}\bigl[z_i(W_{h(i)}, p)\bigr] = 0.
\]
Using the linear structure of $z_i(w_{h(i)},p)$ and the fact that
$W_{h(i)}\sim\mathrm{Bernoulli}(\pi)$ independently of $E_{h(i)}$, we obtain
\begin{align*}
  0
  &=
  \theta_{d01}\,\mu_{\mathrm{eli}}
  +
  \theta_{d00}\,(1-\mu_{\mathrm{eli}})
  +
  \theta_{dw}\,\pi\,\mu_{\mathrm{eli}}
  -
  \theta_{s0}
  +
  (\theta_{dp}-\theta_{sp})\,p.
\end{align*}
Solving for $p$ yields the population equilibrium price under $\RCT(\pi)$:
\begin{equation}
  p_\pi^\ast
  =
  -\frac{
    \theta_{d01}\,\mu_{\mathrm{eli}}
    +
    \theta_{d00}\,(1-\mu_{\mathrm{eli}})
    +
    \theta_{dw}\,\pi\,\mu_{\mathrm{eli}}
    -
    \theta_{s0}
  }{
    \theta_{dp}-\theta_{sp}
  }.
  \label{eq:filmer-p-star}
\end{equation}

The price elasticities of excess demand and outcomes are
\begin{align}
  \xi_z
  &=
  \nabla_p\EE\bigl[z_i(W_{h(i)}, p_\pi^\ast)\bigr]
  =
  \theta_{dp}-\theta_{sp},
  \\
  \xi_y
  &=
  \nabla_p\EE\bigl[y_i(W_{h(i)}, s_i(\bw), p_\pi^\ast)\bigr]
  =
  \theta_{yd}\,\theta_{dp},
\end{align}
consistent with the definitions in the general equilibrium model.

Using these expressions and the linear structure of $y_i(\cdot)$, the population targets
of interest are
\begin{align}
  \tau_{\ADE}^\ast(\pi)
  &=
  \EE\bigl[y_i(1,\pi,p_\pi^\ast) - y_i(0,\pi,p_\pi^\ast)\bigr]
  =
  \theta_{yd}\,\theta_{dw}\,\mu_{\mathrm{eli}}
  +
  \theta_{yw},
  \\
  \tau_{\AIE}^{\rmL,\ast}(\pi)
  &=
  \EE\bigl[
    \pi\,\nabla_s y_i(1,\pi,p_\pi^\ast)
    +
    (1-\pi)\,\nabla_s y_i(0,\pi,p_\pi^\ast)
  \bigr]
  =
  \theta_{ys},
  \\
  \tau_{\AIE}^{\rmG,\ast}(\pi)
  &=
  -\,\xi_y^\top \xi_z^{-1}
  \EE\bigl[z_i(1,p_\pi^\ast)-z_i(0,p_\pi^\ast)\bigr]
  =
  -\frac{\theta_{yd}\,\theta_{dp}}{\theta_{dp}-\theta_{sp}}
  \,\theta_{dw}\,\mu_{\mathrm{eli}}.
\end{align}
The total marginal policy effect is
\[
  \tau_{\TOT}^\ast(\pi)
  =
  \tau_{\ADE}^\ast(\pi)
  +
  \tau_{\AIE}^{\rmL,\ast}(\pi)
  +
  \tau_{\AIE}^{\rmG,\ast}(\pi).
\]

\subsection{Monte Carlo implementation}

In the Monte Carlo experiments reported in the main text we use $n_h=2{,}000$ households
and $2{,}000$ replications.  For each replication we:

\begin{enumerate}
  \item draw household treatments $\bW\sim\RCT(\pi)$;
  \item generate the household network $\bE$ via the construction in
        Section~\ref{subsec:filmer-network};
  \item simulate the augmented-trial perturbations $\bU$ and solve for the equilibrium
        price $P_n(\bW)+U$;
  \item generate individual-level demand, supply, and outcomes using the calibrated
        structural parameters $(\theta,\theta_{ys})$;
  \item aggregate to household-level variables $(\underline{y}_h,\underline{z}_h)$;
  \item compute
        $\hat\tau_{\ADE}$,
        $\hat\tau_{\AIE}^{\rmL}$ (with $r=6$ principal components), and
        $\hat\tau_{\AIE}^{\rmG}=-\hat\gamma^\top\hat\tau_z$, together with their
        estimated standard errors.
\end{enumerate}

Across replications we summarize finite-sample bias, standard deviation, and mean squared
error for each component and compare them to the population targets derived above; the
summary for the baseline calibration is reported in Table~\ref{tab:decomp-filmer} in the
main text.

\begin{figure}[H]
  \centering
  \begin{subfigure}{.32\textwidth}
    \centering
    \includegraphics[width=\linewidth]{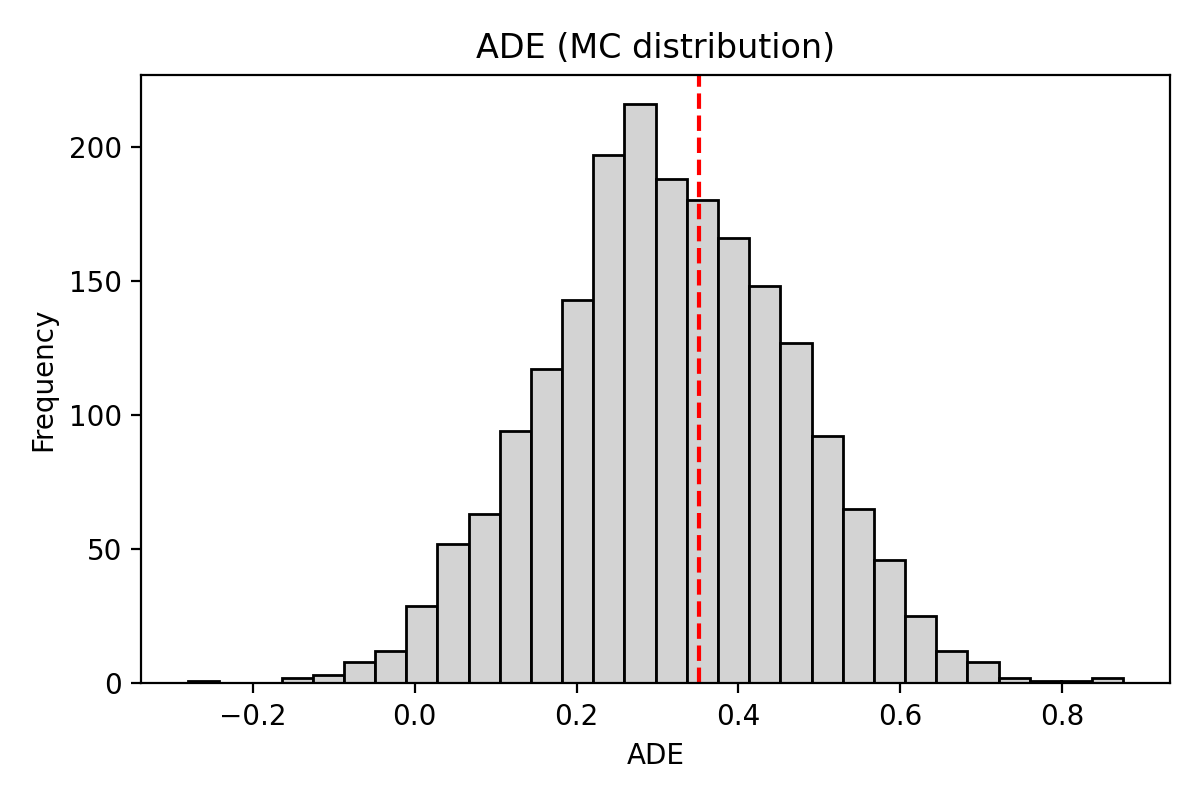}
    \caption{ADE}
  \end{subfigure}
  \hfill
  \begin{subfigure}{.32\textwidth}
    \centering
    \includegraphics[width=\linewidth]{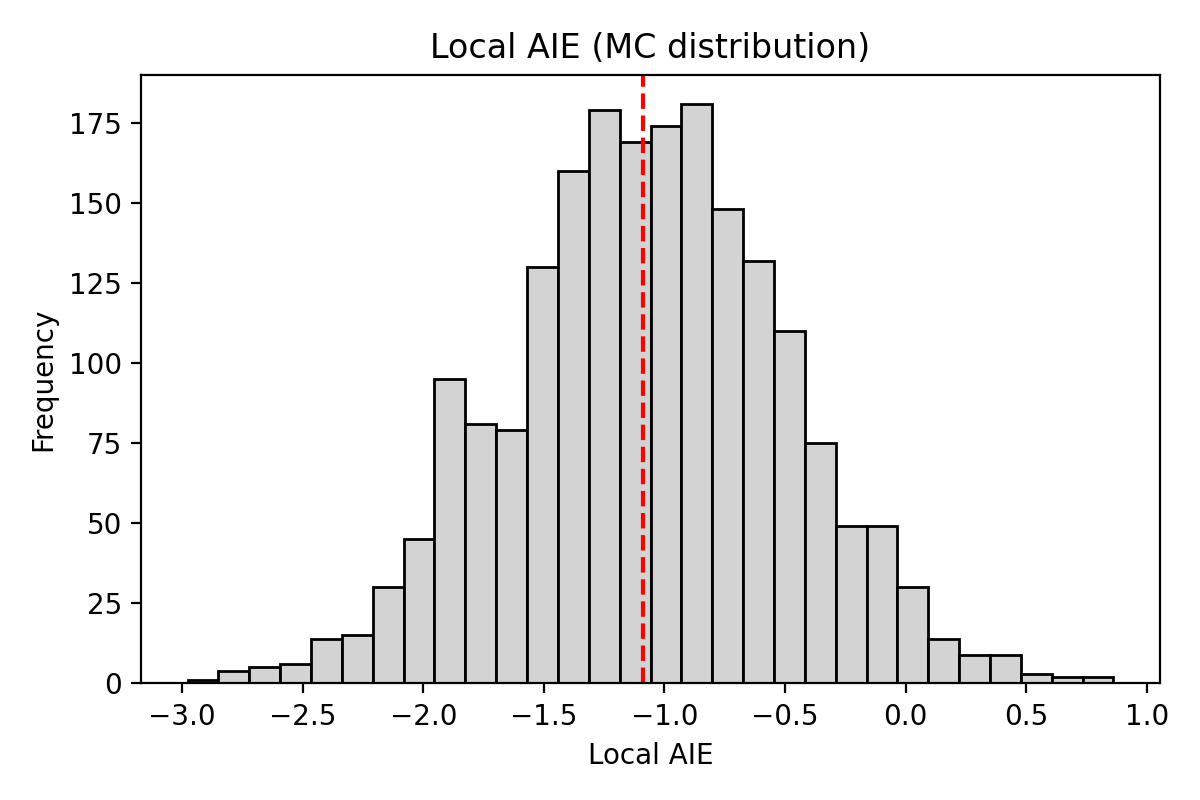}
    \caption{Local AIE}
  \end{subfigure}
  \hfill
  \begin{subfigure}{.32\textwidth}
    \centering
    \includegraphics[width=\linewidth]{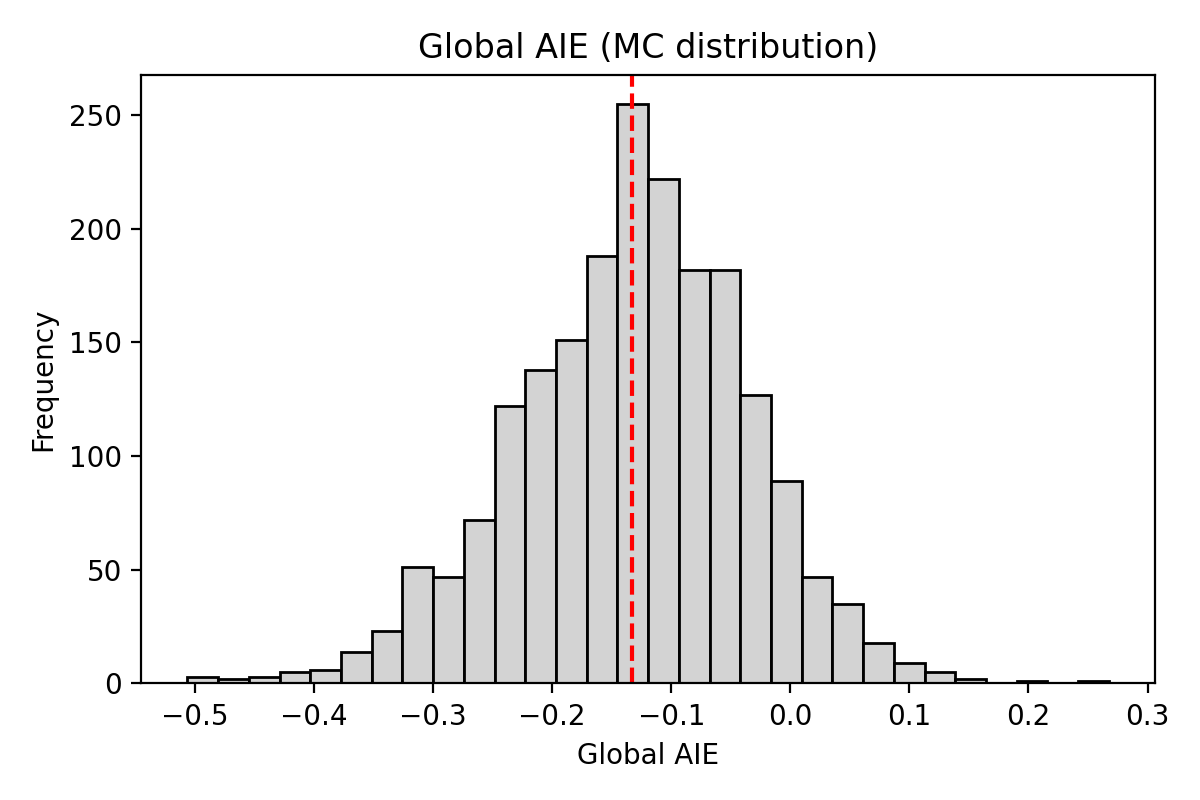}
    \caption{Global AIE}
  \end{subfigure}
  \caption{Monte Carlo sampling distributions for the ADE, local AIE, and global AIE
  estimators in the Filmer-calibrated design.}
  \label{fig:filmer-hist}
\end{figure}

\section{Proofs for Section~\ref{sec:misspecified estimand}}

\subsection{Proof of Theorem~\ref{thm:mpe equal direct plus indirect}}

\begin{proof}[Proof of Theorem~\ref{thm:mpe equal direct plus indirect}]
For convenience, we introduce the following notation.
\begin{align*}
    \tilde{y}_i(\bw;\pi) &= \EE_{\bW^{(2)}\sim\mathrm{RCT}(\pi)}\{ y_i(\bW^{(2)}) |d_i(\bW^{(2)})=d_i(w)\}, \\
    \bar{V}(\bw;\pi) &= \frac{1}{n}\sum_{i=1}^n \tilde{y}_i(\bw;\pi).
\end{align*}
Later, we will use $p(\bw;\pi)$ as the p.m.f. of $\bw$ under $\mathrm{RCT}(\pi)$.
Henceforth, when breaking the expectation over $w^{(1)}$, the empirical sum becomes
\begin{align}
    \frac{1}{n}\sum_{i=1}^n \Bar{y}_i(\pi_1,\pi_2) = \sum_{\bw\in\{0,1\}^n} \bar{V}(\bw;\pi_2) p(\bw;\pi_1).
\end{align}
Since $\frac{\mathrm{d}}{\mathrm{d}\pi} p(\bw;\pi) = p(\bw;\pi)\sum_{i=1}^n \frac{w_i-\pi}{\pi(1-\pi)}$, by taking the partial derivative with respect to $\pi_1$, we find that
\begin{align}
    \tau(\pi) &= \frac{\partial}{\partial\pi_1}\left[\frac{1}{n}\sum_{i=1}^n \Bar{y}_i(\pi_1,\pi_2)\right] \Bigg|_{\pi_1=\pi_2=\pi} \\
    &= \left[ \sum_{\bw\in\{0,1\}^n} \bar{V}(\bw;\pi_2) \frac{\mathrm{d}}{\mathrm{d}\pi_1} p(\bw;\pi_1)\right]\Bigg|_{\pi_1=\pi_2=\pi} \\
    &= \left[ \sum_{\bw\in\{0,1\}^n} \bar{V}(\bw;\pi_2) p(\bw;\pi_1) \sum_{i=1}^n \frac{w_i-\pi_1}{\pi_1(1-\pi_1)} \right]\Bigg|_{\pi_1=\pi_2=\pi} \\
    &= \EE_{\bW\sim\mathrm{RCT}(\pi)}\left[ \bar{V}(\bW;\pi) \sum_{i=1}^n \frac{W_i-\pi}{\pi(1-\pi)} \right].
\end{align}
Furthermore, we should plug in the definition of $\bar{V}(\bw;\pi)$ to find
\begin{align}
    \tau(\pi) &= \frac{1}{n}\EE_{\bW\sim\mathrm{RCT}(\pi)}\left[  \sum_{i=1}^n \sum_{j=1}^n \frac{W_i-\pi}{\pi(1-\pi)} \tilde{y}_j(\bW;\pi)\right] \\
    &= \frac{1}{n}\sum_{i=1}^n\EE_{\bW}\left[\frac{W_i-\pi}{\pi(1-\pi)} \tilde{y}_i(\bW;\pi)\right] + \frac{1}{n}\sum_{j=1}^n\sum_{i\neq j} \EE_{\bW}\left[\frac{W_i-\pi}{\pi(1-\pi)} \tilde{y}_j(\bW;\pi)\right].
\end{align}
Lastly note that for any pair $(i,j)$, there holds that
\begin{equation}
    \EE_{\bW}\left[\frac{W_i-\pi}{\pi(1-\pi)} \tilde{y}_j(\bW;\pi)\right] = \EE_{\bW}\left[ \tilde{y}_j(W_i=1,W_{-i};\pi) - \tilde{y}_j(W_i=0,W_{-i};\pi) \right].
\end{equation}
Consequently, we end up with
\begin{align}
    \tau(\pi) &= \frac{1}{n}\sum_{i=1}^n\EE_{\bW}\left[ \tilde{y}_i(W_i=1,W_{-i};\pi) - \tilde{y}_i(W_i=0,W_{-i};\pi) \right] \\
    &\qquad + \frac{1}{n}\sum_{j=1}^n\sum_{i\neq j}\EE_{\bW}\left[ \tilde{y}_j(W_i=1,W_{-i};\pi) - \tilde{y}_j(W_i=0,W_{-i};\pi) \right].
\end{align}
Call the first term as a definition of direct effect, and the second term as a definition of indirect effect. Note that the definition of $\tilde{y}_j$ depends on the exposure $d_j$ we choose.
\end{proof}

\subsection{Proof of Proposition~\ref{prop:functional estimand being Lipschitz}}\label{sec:reflecting sq loss criterion}

\begin{proof}[Proof of Proposition~\ref{prop:functional estimand being Lipschitz}]
    For a Bernoulli variable $W\in\cbr{0,1}$ with $\pi=\PP(W=1)=1-\PP(W=0)$ and any function $f$ on $\cbr{0,1}$, such an identity holds,
    \begin{equation}
        f(1) - f(0) = \frac{1}{\pi(1-\pi)}\EE[(W-\pi)f(W)].
    \end{equation}
    Exploiting this identity and Assumption~\ref{assump:RCT}, we can transform the oracle and functional definitions in~\eqref{eq:oracle estimand} and~\eqref{eq:functional estimand}, like
    \begin{align}
        \tau_{\ADE}^{\oracle}(\pi) &= \frac{1}{n\pi(1-\pi)}\sum_{i=1}^n\EE_{\bW\sim\RCT(\pi)}\sbr{ (W_i-\pi) y_i\rbr{\bW}  }, \\
        \tau_{\ADE}^{\functional}(f;\pi) &= \frac{1}{n\pi(1-\pi)}\sum_{i=1}^n\EE_{\bW\sim\RCT(\pi)}\sbr{ (W_i-\pi) f_i\rbr{\bW}  }.
    \end{align}
    Therefore, via Cauchy-Schwarz inequality,
    \begin{align}
        &\quad \abr{ \tau_{\ADE}^{\functional}(f;\pi) - \tau_{\ADE}^{\oracle}(\pi) } \\
        &\le \frac{1}{n\pi(1-\pi)}\sum_{i=1}^n\EE_{\bW} \big| (W_i-\pi) \sbr{ y_i\rbr{\bW}-f_i\rbr{\bW}}  \big| \\
        &\le \frac{1}{n\pi(1-\pi)}\sum_{i=1}^n \sqrt{\EE_{\bW}  (W_i-\pi)^2} \sqrt{ \EE_{\bW} \sbr{ y_i\rbr{\bW}-f_i\rbr{\bW}}^2 } \\
        &= \frac{1}{n\sqrt{\pi(1-\pi)}}\sum_{i=1}^n \sqrt{ \EE_{\bW} \sbr{ y_i\rbr{\bW}-f_i\rbr{\bW}}^2 } \\
        &\le \frac{1}{\sqrt{n\pi(1-\pi)}} \sqrt{ \sum_{i=1}^n  \EE_{\bW} \sbr{ y_i\rbr{\bW}-f_i\rbr{\bW}}^2 }.
    \end{align}
    Compared to $\AIE$, it would be easier to directly prove the case for $\MPE$. Just like before, we have that
    \begin{align}
        \tau_{\MPE}^{\oracle}(\pi) &= \frac{1}{n\pi(1-\pi)}\sum_{i=1}^n\sum_{j=1}^n\EE_{\bW\sim\RCT(\pi)}\sbr{ (W_i-\pi) y_j\rbr{\bW}  }, \\
        \tau_{\MPE}^{\functional}(f;\pi) &= \frac{1}{n\pi(1-\pi)}\sum_{i=1}^n\sum_{j=1}^n\EE_{\bW\sim\RCT(\pi)}\sbr{ (W_i-\pi) f_j\rbr{\bW}  }.
    \end{align}
    Then it follows that
    \begin{align}
        &\quad \abr{ \tau_{\MPE}^{\functional}(f;\pi) - \tau_{\MPE}^{\oracle}(\pi) } \\
        &\le \frac{1}{n\pi(1-\pi)} \EE_{\bW}\abr{ \sum_{i=1}^n\sum_{j=1}^n (W_i-\pi) (y_j\rbr{\bW}-f_j\rbr{\bW})  }\\
        &\le \frac{1}{n\pi(1-\pi)} \sqrt{\EE_{\bW}\sbr{\sum_{i}(W_i-\pi)}^2} \, \sqrt{\EE_{\bW}\sbr{\sum_{j}(y_j\rbr{\bW}-f_j\rbr{\bW})}^2} \\
        &= \frac{1}{\sqrt{n\pi(1-\pi)}} \sqrt{\EE_{\bW}\sbr{\sum_{j}(y_j\rbr{\bW}-f_j\rbr{\bW})}^2} \\
        &\le \frac{1}{\sqrt{\pi(1-\pi)}} \sqrt{ \sum_{i=1}^n  \EE_{\bW} \sbr{ y_i\rbr{\bW}-f_i\rbr{\bW}}^2 }. \label{eq:Lipschitz of AIE}
    \end{align}
    Since $\tau_{\AIE}^{\oracle}(\pi)=\tau_{\MPE}^{\oracle}(\pi)-\tau_{\ADE}^{\oracle}(\pi)$ and $\tau_{\AIE}^{\functional}(f;\pi)=\tau_{\MPE}^{\functional}(f;\pi)-\tau_{\ADE}^{\functional}(f;\pi)$, the Lipschitz result holds for $\AIE$ as well.
    Taking $y_i(\bw)=\sum_{j=1}^n (w_j-\pi)$ and $f_i(\bw)=2y_i(\bw)$ suffices to illustrate that the bound in~\eqref{eq:Lipschitz of AIE} is tight.
\end{proof}

\section{Preliminaries for Section~\ref{sec:local vs global example}}

\subsection{Detailed assumptions}\label{sec:assump}

\begin{assumption}[Regularity of potential outcomes]
\label{assump:regularity}
    Assume that for any $y(\cdot)\in\cY$ and $z(\cdot)\in\cZ$, it holds that for any $w\in\{0,1\},s\in[0,1],p\in\cS$
    \begin{align}
        \abr{y(w,s,p)},\, \abr{\nabla_s y(w,s,p)},\, \abr{\nabla_s^2 y(w,s,p)} &\le B,\\
        \abr{\nabla_p y(w,s,p)},\, \abr{\nabla_s\nabla_p y(w,s,p)},\, \abr{\nabla_s^2\nabla_p y(w,s,p)} &\le B,\\
        \abr{\nabla_p^2 y(w,s,p)},\, \abr{\nabla_s^3\nabla_p y(w,s,p)}, \,  \abr{\nabla_s^3 y(w,s,p)} &\le B,\\
        \abr{z(w,p)},\,\abr{\nabla_p z(w,p)},\, \abr{\nabla_p^2 z(w,p)}, \, \abr{\nabla_p^3 z(w,p)} &\le B.
    \end{align}
\end{assumption}

\begin{assumption}
    We assume Condition~\ref{assump:graphon} and Condition~\ref{assump: augmented trial} to be true.
\end{assumption}

Random graph asymptotics such as Condition~\ref{assump:graphon} have yielded prominent results in statistical network analysis \citep{jin2024mixed,deng2024network,yang2025fundamental} to model pairwise interactions.

\begin{assumption}\label{assump:empirical price}
    The market prices satisfy $P_n=P_n(\bW)$, where $P_n(\bw)$ sets the excess demand to approximately $0$ with high probability in the following sense. There exists a sequence $a_n=o(1/n)$ and constants $b,c_1>0$ such that, for every $\bw\in\{0,1\}^n$ and for $U_i$ drawn iid through the augmented randomization design,
    \begin{equation}
        \cS_{\bw} = \cbr{p\in\RR^J:\nbr{\frac{1}{n}\sum_{i=1}^n z_i\rbr{w_i,p+U_i}} \le a_n }
    \end{equation}
    is non-empty with probability at least $1-e^{-c_1n}$ for all $n$. On the event where this set is non-empty, the market price is in this set, $P_n(\bw)\in\cS_{\bw}$.
\end{assumption}

\citet{munro2021treatment} only require $a_n=o(1/\sqrt{n})$ in their Assumption 2. We slightly strengthen it to $a_n=o(1/n)$. The stronger condition facilitates higher-order expansion for $P_n(\bW)-p_{\pi}^{\ast}$ as in Lemma~\ref{lemma: convergence of price}, and finally enable us to establish the convergence of $\tau_{\AIE}^{\rmL}$ and $\tau_{\MPE}^{\oracle}$. This practice is common in the literature of Z-estimators, like \cite[Sections 5.4 and 6.6]{van2000asymptotic}. Empirically, the stronger condition can also be easily satisfied by decreasing the tolerance parameter for numerical optimization.

\begin{assumption}
    Given any randomization policy $\pi$, there exists a unique population-clearing price $p_{\pi}^{\ast}\in\cS$ such that $\EE\sbr{z_i(W_i,p_{\pi}^\ast)}=0$. Moreover, we let the Jacobian $\xi_z=\nabla_p \EE\sbr{z_i\rbr{W_i,p_{\pi}^\ast}}\in\RR^{J\times J}$ to be full-rank, specifically with $\lambda_{\min}\rbr{\xi_z}>0$.
\end{assumption}


\subsection{Notable lemmas}

This section collects several key lemmas from \cite{li2022random,munro2021treatment}, laying foundations for detailed theoretical analysis of our outcome model~\eqref{eq: outcome model}.

    



\begin{lemma}[Part of the proof of Theorem 4 in \cite{li2022random}]\label{lemma:averaging over neighborhood}
    Suppose $\cbr{(\phi_j,Q_j)\in\RR\times[0,1]}_{j\in[n]}$ is drawn i.i.d from a certain distribution such that the marginal of $Q_j$ is uniform in $[0,1]$ and $\phi_j$ is uniformly bounded. And then an adjacency matrix $\bE=(E_{ij})_{1\le i<j\le n}$ is generated from the graphon model. Then for any $i\in[n]$,
    \begin{equation}
        \EE\rbr{\sum_{j\neq i}\frac{E_{ij}\phi_j}{N_j} - \EE\sbr{\frac{E_{ij}\phi_j}{g_n(Q_j)}\Big|Q_i} }^2 = O\rbr{\frac{1}{n\rho_n}}.
    \end{equation}
\end{lemma}
To present the following lemma, we also let
\begin{equation}
    \tilde{\Xi}_z = \EE\sbr{ \pi \nabla^2 z(1,p_{\ast}^{\ast}) + (1-\pi) \nabla^2 z(0,p_{\ast}^{\ast}) } \in \RR^{J\times J\times J}
\end{equation}
be the Hessian of $z$ at the limiting price $p_{\pi}^{\ast}$. It is a $3$-order tensor since $z$ takes values in $\RR^J$. However, it would be more convenient to identify it as a linear map $\Xi_z:\RR^{J\times J}\to\RR^{J}$, so that it facilitates writing Taylor expansion as below
\begin{equation}\label{eq:Taylor for z up to second order}
    \EE\sbr{z(W,p)} = \xi_z(p-p_{\pi}^{\ast}) + \frac{1}{2} \Xi_z\sbr{\rbr{p-p_{\pi}^{\ast}}\rbr{p-p_{\pi}^{\ast}}^\top} + o\rbr{\nbr{p-p_{\pi}^{\ast}}^2}.
\end{equation}
\begin{lemma}[Convergence of the finite-sample price variables]\label{lemma: convergence of price}
    Under all the assumptions mentioned before, the equilibrium market prices $P_n(\bW)$ satisfies the following: as $n\to\infty$,
    \begin{align}\label{eq:second-order convergence of price}
        & P_n(\bW)-p_{\pi}^{\ast} = -\xi_z^{-1}\bar{Z}_n  -  \cbr{\frac{1}{2}\xi_z^{-1}\Xi_z\sbr{ \rbr{\xi_z^{-1} \bar{Z}_n}\rbr{\xi_z^{-1} \bar{Z}_n}^\top } - \xi_z^{-1} \bar{\epsilon}_n\xi_z^{-1}\bar{Z}_n} + o_p(1/n),
    \end{align}
    where we write $\bar{Z}_n := \frac{1}{n}\sum_{i=1}^n z_i(W_i,p_{\pi}^{\ast}+U_i)$ and $\bar{\epsilon}_n=\frac{1}{n} \sum_{i=1}^n \nabla z_i(W_i,p_{\pi}^{\ast}+U_i)-\xi_z$. Inside the bracket $\cbr{\cdot}$ is the term of order $1/n$.
\end{lemma}
\citet{munro2021treatment} also obtained a very similar result. While assuming a slightly weaker condition, they prove that
\begin{equation}\label{eq:first-order convergence of price}
    P_n(\bW)-p_{\pi}^{\ast} = -\xi_z^{-1}\sbr{\frac{1}{n}\sum_{i=1}^n z_i(W_i,p_{\pi}^{\ast}+U_i)} + o_p(1/\sqrt{n}).
\end{equation}
In comparison, our result~\eqref{eq:second-order convergence of price} is just a higher-order expansion. In proving Theorem~\ref{thm:asymptotics of estimand local AIE}, we need the stronger convergence to control $P_n(\bW^{(2)})-P_n(\bW_{\cN_i}^{(1)},\bW_{-\cN_i}^{(2)})$ sharply. In proving Theorem~\ref{thm:asymptotics of estimand total}, the stronger convergence is also needed to control an error term resulting from the quadratic terms in Taylor expansion.

\begin{proof}
    Firstly, we emphasize that our assumptions are only stronger than those in \cite{munro2021treatment}, so we can directly use their Lemma 14 to establish~\eqref{eq:first-order convergence of price}. To strengthen it to~\eqref{eq:second-order convergence of price}, apply Taylor expansion to every $z_i$ so that
    \begin{align}
        z_i(W_i,p+U_i) &= z_i(W_i,p_{\pi}^{\ast}+U_i) + \nabla z_i(W_i,p_{\pi}^{\ast}+U_i) (p-p_{\pi}^{\ast}) \\
        &\qquad + \frac{1}{2} \nabla^2 z_i(W_i,p_{\pi}^{\ast}+U_i) \sbr{(p-p_{\pi}^{\ast})(p-p_{\pi}^{\ast})^\top} + O\rbr{\nbr{p-p_{\pi}^{\ast}}^3}.
    \end{align}
    This Taylor expansion is possible because we have assumed $\nabla^3 z_i$ to exist and uniformly bounded in Assumption~\ref{assump:regularity}. 
    Specify it to $p=P_n(\bW)$, and average over $i\in[n]$, we learn that
    \begin{align}
        0 &= \frac{1}{n}\sum_{i=1}^n z_i(W_i,p_{\pi}^{\ast}+U_i) + \sbr{\frac{1}{n}\sum_{i=1}^n\nabla z_i(W_i,p_{\pi}^{\ast}+U_i)} \rbr{P_n(\bW)-p_{\pi}^{\ast}} \\
        &\qquad + \sbr{\frac{1}{2n}\sum_{i=1}^n\nabla^2 z_i(W_i,p_{\pi}^{\ast}+U_i)} \sbr{(P_n(\bW)-p_{\pi}^{\ast})(P_n(\bW)-p_{\pi}^{\ast})^\top} + o_p(1/n).\label{eq:averaged Taylor expansion}
    \end{align}
    Here the error term $o_p(1/n)$ comes from both $a_n=o(1/n)$ in Assumption~\ref{assump:empirical price} and the known convergence bound that $P_n(\bW)-p_{\pi}^{\ast}=O_p(1/\sqrt{n})$ as in~\eqref{eq:first-order convergence of price}. In the following, let $\epsilon_i=\nabla z_i(W_i,p_{\pi}^{\ast}+U_i)-\xi_z$.
    Lastly, by law of large numbers and central limit theorems,
    \begin{align}
        \frac{1}{n}\sum_{i=1}^n\nabla z_i(W_i,p_{\pi}^{\ast}+U_i) &= \xi_z +\frac{1}{n} \sum_{i=1}^n \epsilon_i + o_p(1/\sqrt{n}), \\
        \frac{1}{n}\sum_{i=1}^n\nabla^2 z_i(W_i,p_{\pi}^{\ast}+U_i) &= \Xi_z + O_p(1/\sqrt{n}).
    \end{align}
    In combination with $P_n(\bW)-p_{\pi}^{\ast}=-\xi_z^{-1} \bar{Z}_n + o_p(1/\sqrt{n})$, equation~\eqref{eq:averaged Taylor expansion} then becomes
    \begin{align}
        0 = \bar{Z}_n + \sbr{\xi_z +\frac{1}{n} \sum_{i=1}^n \epsilon_i} \rbr{P_n(\bW)-p_{\pi}^{\ast}} + \frac{1}{2}\Xi_z\sbr{ \rbr{\xi_z^{-1} \bar{Z}_n}\rbr{\xi_z^{-1} \bar{Z}_n}^\top } + o_p(1/n). 
    \end{align}
    With $\bar{\epsilon}_n=\frac{1}{n} \sum_{i=1}^n \epsilon_i$, it follows that
    \begin{align}
        P_n(\bW)-p_{\pi}^{\ast} &= -\rbr{\xi_z +\bar{\epsilon}_n}^{-1} \cbr{\bar{Z}_n+\frac{1}{2}\Xi_z\sbr{ \rbr{\xi_z^{-1} \bar{Z}_n}\rbr{\xi_z^{-1} \bar{Z}_n}^\top }} + o_p(1/n) \\
        &=- \xi_z^{-1}\bar{Z}_n + \xi_z^{-1} \bar{\epsilon}_n\xi_z^{-1}\bar{Z}_n -\frac{1}{2}  \xi_z^{-1}\Xi_z\sbr{ \rbr{\xi_z^{-1} \bar{Z}_n}\rbr{\xi_z^{-1} \bar{Z}_n}^\top } + o_p(1/n),
    \end{align}
    which concludes this lemma.
\end{proof}


\section{Proofs about estimands in Section~\ref{sec:local vs global example estimand}}
\label{sec:app-th4.5}
\subsection{Remarks before proofs}
In Section~\ref{sec:local vs global example estimator} of the main text, we have introduced an augmented design (Assumption~\ref{assump: augmented trial}) to facilitate effective estimation of price elasticities. In this augmented design, individualized price perturbation $U_i$ is generated onto every unit.

As long as these perturbations are well-conditioned, i.e.
\begin{equation}\label{eq:well-conditioned price pertubation}
    \EE\sbr{U_i}=0 \text{ and } \nbr{U_i}\le h_n =o(n^{-1/4}) \text{ almost surely},
\end{equation}
they will not change the asymptotics of the estimands at all. For conciseness, we will establish Theorem~\ref{thm: asymptotic limit of estimands} under the assumption that random individualized price perturbations exist.
For better presentation, Theorem~\ref{thm: asymptotic limit of estimands} is split into Theorems~\ref{thm:asymptotics of estimand local AIE}, \ref{thm:asymptotics of estimand global AIE}, \ref{thm:asymptotics of estimand total} and \ref{thm:asymptotics of estimand direct}.

\subsection{Proof for the local spillover estimand}\label{sec: local spillover estimand proof}

\begin{theorem}\label{thm:asymptotics of estimand local AIE}
    Suppose that Assumptions~\ref{assump:RCT}, \ref{assump:super-population for units} and the assumptions in Section~\ref{sec:assump} all hold. \textcolor{black}{If $1/3<\kappa<1/2$,} then the estimand $\tau^{\rmL}_{\AIE}$ defined in \eqref{eq:local spillover estimand} converges as follows
    \begin{equation}
        \tau_{\AIE}^{\rmL} \pto \tau_{\AIE}^{\rmL,\ast} := \EE\sbr{\pi \nabla_s y_i(1,\pi,p_{\pi}^{\ast})+(1-\pi)\nabla_s y_i(0,\pi,p_{\pi}^{\ast})}.
    \end{equation}
\end{theorem}

To rigorously prove this result, let's recall the following lemma.
\begin{lemma}[Proposition 1 from \cite{li2022random}]\label{lemma: prop1 li and wager}
    For any $f_i:\cbr{0,1}\times[0,1]\to\RR$ such that
    \begin{equation}
        \abr{f_i(w,s)}, \abr{\nabla_s f_i(w,x)}, \abr{\nabla_s^2 f_i(w,x)}, \abr{\nabla_s^3 f_i(w,x)} \le B, \quad \forall i\in[n],
    \end{equation}
    it holds that
    \begin{align}
        &\qquad \frac{1}{n}\sum_{j=1}^n\sum_{i\neq j}\EE_{\bW\sim\RCT(\pi)}\sbr{ f_i\rbr{ W_i, \frac{E_{ij}+\sum_{k\neq i,j} E_{ik}W_k}{N_i}} - f_i\rbr{ W_i, \frac{\sum_{k\neq i,j} E_{ik}W_k}{N_i}} } \\
        &= \frac{1}{n} \sum_{i=1}^n \sbr{\pi \nabla_s f_i(1,\pi) + (1-\pi)\nabla_s f_i(0,\pi) } + O\rbr{\frac{B}{\sqrt{\min_i N_i}}}.
    \end{align}
\end{lemma}

\begin{proof}[Proof of Theorem~\ref{thm:asymptotics of estimand local AIE}]
In order to effectively use Lemma~\ref{lemma: prop1 li and wager} in our own setup, write
\begin{equation}\label{eq:intermediate f_i local}
    f_i(w,s) := \EE_{\bW^{(2)}}\cbr{ y_i\rbr{ w,s,P_n(\bW^{(2)}) + U_i } }.
\end{equation}
Based on this, we define an intermediate quantity
\begin{equation}\label{eq:intermediate estimand local}
    \tilde{\tau} = \frac{1}{n}\sum_{j=1}^n\sum_{i\neq j}\EE_{\bW\sim\RCT(\pi)}\sbr{ f_i\rbr{W_i,\frac{E_{ij}+\sum_{k\neq i,j} E_{ik}W_k}{N_i}} - f_i\rbr{W_i,\frac{\sum_{k\neq i,j} E_{ik}W_k}{N_i}} }.
\end{equation}
The rest of the proof consists of two steps: (i) showing $\tilde{\tau} = \tau_{\AIE}^{\rmL,\ast}+o_p(1)$; (ii) establishing that $\tau_{\AIE}^{\rmL}$ is asymptotically very close to $\tilde{\tau}$.

\medskip
\textit{Step 1.}
This set of $f_i$'s naturally satisfy the regularity conditions in Lemma~\ref{lemma: prop1 li and wager} under Assumption~\ref{assump:regularity}, so we have
\begin{align}
    \tilde{\tau} &= \frac{1}{n} \sum_{i=1}^n \sbr{\pi \nabla_s f_i(1,\pi) + (1-\pi)\nabla_s f_i(0,\pi) } + O\rbr{\frac{B}{\sqrt{\min_i N_i}}} \\
    &= \frac{1}{n} \sum_{i=1}^n \Bigg[ \pi \EE_{\bW^{(2)}}\cbr{ \nabla_s y_i\rbr{ 1,\pi,P_n(\bW^{(2)})+U_i } } \\
    &\qquad\qquad +(1-\pi)\EE_{\bW^{(2)}}\cbr{ \nabla_s y_i\rbr{ 0,\pi,P_n(\bW^{(2)})+U_i } } \Bigg] + O\rbr{\frac{B}{\sqrt{\min_i N_i}}} \\
    &= \frac{1}{n} \sum_{i=1}^n \sbr{ \pi\nabla_s y_i\rbr{ 1,\pi,p_{\pi}^\ast+U_i} + (1-\pi)\nabla_s y_i\rbr{ 0,\pi,p_{\pi}^\ast+U_i} } \\
    &\qquad\qquad + O\rbr{\nbr{ P_n\rbr{\bW^{(2)}} - p_{\pi}^{\ast} }} + O\rbr{\frac{B}{\sqrt{\min_i N_i}}}.
\end{align}
Since every $\nbr{U_i}\le o(n^{-1/4})$ and $\nabla_s y_i(w,s,p)$ is continuously differentiable in $p$, the effect of $U_i$ is negligible. By Lemma 15 of \cite{li2022random}, $\min_i N_i\to\infty$; by Lemma~\ref{lemma: convergence of price}, $P_n\rbr{\bW^{(2)}}\pto p_{\pi}^\ast$. So we have
\begin{align}
    \tilde{\tau} &= \frac{1}{n} \sum_{i=1}^n \sbr{ \pi\nabla_s y_i\rbr{ 1,\pi,p_{\pi}^\ast} + (1-\pi)\nabla_s y_i\rbr{ 0,\pi,p_{\pi}^\ast} } + o_p(1) \\
    &= \EE\sbr{\pi \nabla_s y(1,\pi,p_{\pi}^{\ast})+(1-\pi)\nabla_s y(0,\pi,p_{\pi}^{\ast})} + o_p(1),
\end{align}
where the last approximation uses law of large numbers driven by independently drawing units from the superpopulation (Assumption~\ref{assump:super-population for units}).

\medskip
\textit{Step 2.}
Plugging \eqref{eq:intermediate f_i local} into the definition of $\tilde{\tau}$ in \eqref{eq:intermediate estimand local}, we have
\begin{align}
    \tilde{\tau} &= \frac{1}{n}\sum_{j=1}^n\sum_{i\neq j} \EE_{\bW^{(1)},\bW^{(2)}}\Bigg\{  y_i\rbr{ W_i^{(1)}, \frac{E_{ij}+\sum_{k\neq i,j}E_{ik}W^{(1)}_k }{N_i}, P_n\rbr{\bW^{(2)}}+U_i }  \\
    & \qquad\qquad\qquad\qquad\qquad\qquad - y_i\rbr{ W_i^{(1)}, \frac{\sum_{k\neq i,j}E_{ik}W^{(1)}_k }{N_i}, P_n\rbr{\bW^{(2)}}+U_i } \Bigg\}. \label{eq:intermediate estimand local another}
\end{align}
Compared to our original definition of the local spillover estimand $\tau^{\rmL}_{\AIE}$ defined from~\eqref{eq:local spillover estimand},
\begin{align}
    \tau^{\rmL}_{\AIE} = \frac{1}{n}\sum_{j=1}^n\sum_{i\neq j} \EE_{\bW^{(1)},\bW^{(2)}} & \Bigg\{  y_i\rbr{ W_i^{(1)}, \frac{E_{ij}+\sum_{k\neq i,j}E_{ik}W^{(1)}_k }{N_i}, P_n\rbr{\bW^{(1)}_{\cN_i},\bW^{(2)}_{-\cN_i}}+U_i }  \\
    & - y_i\rbr{ W_i^{(1)}, \frac{\sum_{k\neq i,j}E_{ik}W^{(1)}_k }{N_i}, P_n\rbr{\bW^{(1)}_{\cN_i},\bW^{(2)}_{-\cN_i}}+U_i } \Bigg\}, \label{eq:local spillover estimand another}
\end{align}
the intermediate quantity $\tilde{\tau}$ mainly differs in that the price $P_n(\bW^{(2)})$ is determined solely by the second set of treatments. 
Since the proportion $\abr{\cN_i}/n\to0$, the price $P_n\rbr{\bW^{(1)}_{\cN_i},\bW^{(2)}_{-\cN_i}}$ should not deviate too much from $P_n(\bW^{(2)})$.

Recall that every $y_i(w,s,p)$ is continuously differentiable in $p$. Note in particular that the summation $\sum_{j=1}^n\sum_{i\neq j}$ appearing in \eqref{eq:intermediate estimand local another} and \eqref{eq:local spillover estimand another} actually only involve around $n^2\rho_n$ effective terms. This is significantly less than the seemingly number $n(n-1)$ because the summand becomes zero as long as $E_{ij}=0$. For any $w=0,1$, let
\begin{align}
    \Delta(w) := \frac{1}{n}\sum_{j=1}^n\sum_{i\in\cN_j} \EE_{\bW^{(1)},\bW^{(2)}} & \Bigg\{  y_i\rbr{ W_i^{(1)}, \frac{E_{ij}w+\sum_{k\neq i,j}E_{ik}W^{(1)}_k }{N_i}, P_n\rbr{\bW^{(1)}_{\cN_i},\bW^{(2)}_{-\cN_i}}+U_i }  \\
    & - y_i\rbr{ W_i^{(1)}, \frac{E_{ij}w+\sum_{k\neq i,j}E_{ik}W^{(1)}_k }{N_i}, P_n\rbr{\bW^{(2)}}+U_i } \Bigg\},
\end{align}
so that $\tau^{\rmL}_{\AIE} - \tilde{\tau} = \Delta(1)-\Delta(0)$. Due to Lemma~\ref{lemma: convergence of price}, we have that
\begin{align}
    P_n\rbr{\bW^{(1)}_{\cN_i},\bW^{(2)}_{-\cN_i}} - P_n\rbr{\bW^{(2)}} &= -\frac{1}{n}\xi_z^{-1}  \sum_{j\in\cN_i} \sbr{ z_{j}\rbr{W_j^{(1)},p_{\pi}^{\ast}} -   z_{j}\rbr{W_j^{(2)},p_{\pi}^{\ast}} }+ O_p(1/n) \\
    &= O_p(\sqrt{\rho_n/n}).
\end{align}
Since every $y_i$ is at least twice continuous differentiable in $p$, we have that
\begin{align}
    &\quad y_i\rbr{ W_i^{(1)}, \frac{E_{ij}w+\sum_{k\neq i,j}E_{ik}W^{(1)}_k }{N_i}, P_n\rbr{\bW^{(1)}_{\cN_i},\bW^{(2)}_{-\cN_i}}+U_i }  \\
    &\quad - y_i\rbr{ W_i^{(1)}, \frac{E_{ij}w+\sum_{k\neq i,j}E_{ik}W^{(1)}_k }{N_i}, P_n\rbr{\bW^{(2)}}+U_i } \\
    &= \zeta(i,j,w) \sbr{ P_n\rbr{\bW^{(1)}_{\cN_i},\bW^{(2)}_{-\cN_i}} - P_n\rbr{\bW^{(2)}} } + O_p(\rho_n/n) \\
    &= -\zeta(i,j,w) \frac{1}{n}\xi_z^{-1}  \sum_{l\in\cN_i} \sbr{ z_{l}\rbr{W_l^{(1)},p_{\pi}^{\ast}} -   z_{l}\rbr{W_l^{(2)},p_{\pi}^{\ast}} } + O_p(1/n),
\end{align}
where we write $\zeta(i,j,w):=\nabla_p y_i\rbr{ W_i^{(1)}, \frac{E_{ij}w+\sum_{k\neq i,j}E_{ik}W^{(1)}_k }{N_i}, P_n\rbr{\bW^{(2)}}+U_i }$ for simplicity. In this way, we would have
\begin{align}
    \Delta(w) &= -\frac{1}{n}\sum_{j=1}^n\sum_{i\in\cN_j} \EE_{\bW}\cbr{  \zeta(i,j,w) \frac{1}{n}\xi_z^{-1}  \sum_{l\in\cN_i} \sbr{ z_{l}\rbr{W_l^{(1)},p_{\pi}^{\ast}} -   z_{l}\rbr{W_l^{(2)},p_{\pi}^{\ast}} } } + O_p(\rho_n).
\end{align}
Using central limit theorem, we have that
\begin{equation}
    \sum_{l\in\cN_i} \sbr{ z_{l}\rbr{W_l^{(1)},p_{\pi}^{\ast}} -   z_{l}\rbr{W_l^{(2)},p_{\pi}^{\ast}} } = O_p(\sqrt{n\rho_n}).
\end{equation}
Therefore, $\abr{\tau^{\rmL}_{\AIE} - \tilde{\tau}}=\abr{\Delta(1)-\Delta(0)}=O_p(\sqrt{n}\rho_n^{3/2})=o_p(1)$ as long as $1/3<\kappa<1/2$.

\medskip
Combined with the previous Step 1, this proof is complete.
\end{proof}

\subsection{Proof for the global spillover estimand}\label{sec: global spillover estimand proof}

To start with, let's formalize the definition in \eqref{eq:global spillover estimand} a bit more. Let $a_n=n^{-\mu}$ be a series of small positive numbers with $0<\mu<1/2$, so that the condition ``$P_n\rbr{\bW^{(2)}} \approx P_n(\bw)$" is quantified to $\nbr{P_n\rbr{\bW^{(2)}} - P_n(\bw)} \le a_n$. Then the definitions become
\begin{align}
    \tilde{y}_i^{\rmG}(\bw;\pi) &= \EE_{\bW^{(2)}\sim\mathrm{RCT}(\pi)}\cbr{ y_i\rbr{ \bW^{(2)} } \Big| \nbr{P_n\rbr{\bW^{(2)}} - P_n(\bw)} \le a_n } \\
    &= \EE_{\bW^{(2)}}\cbr{ y_i\rbr{ w_i,\frac{\sum_{j\neq i}E_{ij}W_j^{(2)}}{\sum_{j\neq i}E_{ij}}, P_n\rbr{\bw}+U_i } \Big| \nbr{P_n\rbr{\bW^{(2)}} - P_n(\bw)} \le a_n },\\
    \tau_{\AIE}^{\rmG} &= \frac{1}{n}\sum_{j=1}^n\sum_{i\neq j}\EE_{\bW\sim\RCT(\pi)}\sbr{ \tilde{ y}_j^{\rmG}\rbr{w_i=1,\bW_{-i};\pi} - \tilde{y}_j^{\rmG}\rbr{w_i=0,\bW_{-i};\pi} }.
\end{align}
Now we proceed to a concise theorem with proof.

\begin{theorem}\label{thm:asymptotics of estimand global AIE}
    Suppose that Assumptions~\ref{assump:RCT}, \ref{assump:super-population for units} and the assumptions in Section~\ref{sec:assump} all hold. 
    The estimand $\tau^{\rmG}_{\AIE}$ defined in \eqref{eq:global spillover estimand} converges as follows
    \begin{equation}
        \tau_{\AIE}^{\rmG} \pto \tau_{\AIE}^{\rmG,\ast} := - \xi_y^{\top}\xi_z^{-1} \EE\sbr{z_i(1,p_{\pi}^{\ast})-z_i(0,p_{\pi}^{\ast})}.
    \end{equation}
\end{theorem}

\begin{proof}[Proof of Theorem~\ref{thm:asymptotics of estimand global AIE}]
As an intermediate quantity, we define
\begin{align}
    \check{y}_i(\bw)
    &= \EE_{\bW^{(2)}}\cbr{
        y_i\!\left(
            w_i,
            \frac{\sum_{j\neq i}E_{ij}W_j^{(2)}}{\sum_{j\neq i}E_{ij}},
            P_n(\bw)+U_i
        \right)
    }, \\
    \tilde{\tau}
    &= \frac{1}{n}\sum_{j=1}^n\sum_{i\neq j}
    \EE_{\bW\sim\RCT(\pi)}\sbr{
        \check{y}_j(w_i=1,\bW_{-i};\pi)
        -
        \check{y}_j(w_i=0,\bW_{-i};\pi)
    }.
\end{align}
The results of \cite{munro2021treatment} are directly applicable to $\tilde{\tau}$, since $\check{y}_i$ satisfies all the required regularity conditions under Assumption~\ref{assump:regularity}.

Therefore, using Theorem~4 of \cite{munro2021treatment}, we obtain
\begin{align}
    \tilde{\tau}
    \pto
    - \xi_y^{\top}\xi_z^{-1}
    \EE\sbr{z_i(1,p_{\pi}^{\ast})-z_i(0,p_{\pi}^{\ast})},
\end{align}
where this limit uses the definition of $\xi_y$ and the continuous differentiability of $y_i$ with respect to the price argument.

It remains to control the difference between $\tau_{\AIE}^{\rmG}$ and $\tilde{\tau}$. Their only difference is whether we condition on the event that the realized equilibrium price under $\bW^{(2)}$ is close to the equilibrium price induced by $(w_i,\bW^{(1)}_{-i})$. 
Suppose that we fix $i\neq j\in[n]$ and $w_i\in\{0,1\}$ for now. Denote
\begin{align}
    \mathsf{Y}
    &=
    y_j\!\left(
        W_j^{(1)},
        \frac{\sum_{k\neq j}E_{jk}W_k^{(2)}}{\sum_{k\neq j}E_{jk}},
        P_n(w_i;\bW^{(1)}_{-i})+U_j
    \right), \\
    \mathsf{A}
    &=
    \cbr{
        \nbr{
            P_n(\bW^{(2)}) - P_n(w_i;\bW^{(1)}_{-i})
        }
        \le a_n
    },
\end{align}
respectively as the variable to be taken expectation and the event that will be conditioned on. Then
\begin{align}
    &\quad
    \check{y}_j(w_i,\bW^{(1)}_{-i})
    -
    \tilde{y}_j^{\rmG}(w_i,\bW^{(1)}_{-i}) \\
    &=
    \EE_{\bW^{(2)}}[\mathsf{Y}]
    -
    \EE_{\bW^{(2)}}[\mathsf{Y}\mid \mathsf{A}] \\
    &=
    \EE_{\bW^{(2)}}[\mathsf{Y}\mid \mathsf{A}]\,\PP_{\bW^{(2)}}(\mathsf{A})
    +
    \EE_{\bW^{(2)}}[\mathsf{Y}\mid \mathsf{A}^c]\,\PP_{\bW^{(2)}}(\mathsf{A}^c)
    -
    \EE_{\bW^{(2)}}[\mathsf{Y}\mid \mathsf{A}] \\
    &=
    \PP_{\bW^{(2)}}(\mathsf{A}^c)
    \sbr{
        \EE_{\bW^{(2)}}[\mathsf{Y}\mid \mathsf{A}^c]
        -
        \EE_{\bW^{(2)}}[\mathsf{Y}\mid \mathsf{A}]
    }.
\end{align}
Since $|\mathsf{Y}|\le B$ almost surely, it follows that
\begin{equation}
    \abr{
        \check{y}_j(w_i,\bW^{(1)}_{-i})
        -
        \tilde{y}_j^{\rmG}(w_i,\bW^{(1)}_{-i})
    }
    \le
    2B\,
    \PP_{\bW^{(2)}}\rbr{
        \nbr{
            P_n(\bW^{(2)}) - P_n(w_i;\bW^{(1)}_{-i})
        }
        \ge a_n
    }.
\end{equation}
by plugging in the definitions of $\mathsf{Y}$ and $\mathsf{A}$. This inequality holds for any $i\neq j\in[n]$ and $w_i\in\{0,1\}$.
Now let $a_n=n^{-\mu}$ with $0<\mu<1/2$. By the triangle inequality, for every realization of $\bW^{(1)}$,
\begin{align}
    &\PP_{\bW^{(2)}}\rbr{
        \nbr{
            P_n(\bW^{(2)}) - P_n(w_i;\bW^{(1)}_{-i})
        }
        \ge a_n
    } \\
    &\le
    \PP_{\bW^{(2)}}\rbr{
        \nbr{
            P_n(\bW^{(2)}) - p_{\pi}^\ast
        }
        \ge a_n/2
    }
    +
    \one\cbr{
        \nbr{
            P_n(w_i;\bW^{(1)}_{-i}) - p_{\pi}^\ast
        }
        \ge a_n/2
    }.
\end{align}
Taking expectation with respect to $\bW^{(1)}$, we obtain
\begin{align}
    &\EE_{\bW^{(1)}}
    \PP_{\bW^{(2)}}\rbr{
        \nbr{
            P_n(\bW^{(2)}) - P_n(w_i;\bW^{(1)}_{-i})
        }
        \ge a_n
    } \\
    &\le
    \PP\rbr{
        \nbr{
            P_n(\bW^{(2)}) - p_{\pi}^\ast
        }
        \ge a_n/2
    }
    +
    \PP\rbr{
        \nbr{
            P_n(w_i;\bW^{(1)}_{-i}) - p_{\pi}^\ast
        }
        \ge a_n/2
    }.
\end{align}
Since $a_n/2=(1/2)n^{-\mu}$ is of the same order as $n^{-\mu}$, the asymptotic bound~(44) in \cite{munro2021treatment} yields, uniformly in $i$ and $w_i\in\{0,1\}$,
\begin{align}
    \PP\rbr{
        \nbr{
            P_n(\bW^{(2)}) - p_{\pi}^\ast
        }
        \ge a_n/2
    }
    &= o(1/n), \\
    \PP\rbr{
        \nbr{
            P_n(w_i;\bW^{(1)}_{-i}) - p_{\pi}^\ast
        }
        \ge a_n/2
    }
    &= o(1/n).
\end{align}
Therefore,
\begin{align}
    &\quad
    \EE\abr{\tau_{\AIE}^{\rmG}-\tilde{\tau}} \\
    &\le
    \frac{2B}{n}\sum_{j=1}^n\sum_{i\neq j}
    \Bigg[
        \EE_{\bW^{(1)}}
        \PP_{\bW^{(2)}}\rbr{
            \nbr{
                P_n(\bW^{(2)}) - P_n(w_i=1;\bW^{(1)}_{-i})
            }
            \ge a_n
        } \\
    &\qquad\qquad\qquad\qquad\qquad\qquad
        +
        \EE_{\bW^{(1)}}
        \PP_{\bW^{(2)}}\rbr{
            \nbr{
                P_n(\bW^{(2)}) - P_n(w_i=0;\bW^{(1)}_{-i})
            }
            \ge a_n
        }
    \Bigg] \\
    &= o(1).
\end{align}
Hence $\tau_{\AIE}^{\rmG}-\tilde{\tau}=o_p(1)$ by Markov's inequality. Combining this with the limit for $\tilde{\tau}$, we conclude that
\begin{equation}
    \tau_{\AIE}^{\rmG}
    =
    - \xi_y^{\top}\xi_z^{-1}
    \EE\sbr{z_i(1,p_{\pi}^{\ast})-z_i(0,p_{\pi}^{\ast})}
    + o_p(1),
\end{equation}
which finishes the proof.
\end{proof}

\subsection{Proof for the direct/total effect estimand}\label{sec: direct/total estimand proof}

\begin{theorem}\label{thm:asymptotics of estimand total}
    Suppose that Assumptions~\ref{assump:RCT}, \ref{assump:super-population for units} and the assumptions in Section~\ref{sec:assump} all hold.
    The estimand $\tau_{\MPE}^{\oracle}$ defined in \eqref{eq: define total effect}, converges as follows
    \begin{align}
        \tau_{\MPE}^{\oracle} &\pto \tau_{\MPE}^{\oracle,\ast} := \tau_{\ADE}^{\oracle,\ast} + \tau_{\AIE}^{\rmL,\ast} + \tau_{\AIE}^{\rmG,\ast}.
    \end{align}
\end{theorem}

\begin{proof}
    Since every $W_k$ would be independent under RCT (Assumption~\ref{assump:RCT}), it follows that
    \begin{align}
        \tau_{\MPE}^{\oracle} &= \frac{1}{\pi(1-\pi)} \EE_{\pi}\sbr{\frac{1}{n}\sum_{i=1}^nY_i(\bW)\times\sum_{k=1}^n\rbr{W_k-\pi}} \\
        &= \frac{1}{\pi(1-\pi)} \EE_{\pi}\sbr{\frac{1}{n}\sum_{i=1}^ny_i(W_i,S_i,P_n(\bW)+U_i)\times\sum_{k=1}^n\rbr{W_k-\pi}}, \label{eq:tau_TOT decomp}
    \end{align}
    where we have plugged in our model specification~\eqref{eq: outcome model}.

    \vspace*{2mm}
    \noindent\textit{Step 1. Expansion with respect to the market prices.} In this step, we will (i) expand $y_i$ with respect to the market prices $P_n(\bW)$ and only keep the linear terms; 
    (ii) replace $P_n(\bW)-p_{\pi}^{\ast}$ by its major term $-\xi_z^{-1}\bar{Z}_n$ as suggested in Lemma~\ref{lemma: convergence of price}.

    To do (i), we use the regularity of each $y_i$ in Assumption~\ref{assump:regularity}. With some $\tilde{P}_i$ being an interpolation between $p_{\pi}^{\ast}$ and $P_n(\bW)$ for each $i\in[n]$, the expansion takes the form of
    \begin{align}
        &\quad y_i(W_i,S_i,P_n(\bW)+U_i) \\
        &= y_i(W_i,S_i,p_{\pi}^{\ast}+U_i) + (P_n(\bW)-p_{\pi}^{\ast})^{\top}\nabla_p y_i(W_i,S_i,p_{\pi}^{\ast}+U_i) \\
        &\qquad +\frac{1}{2} (P_n(\bW)-p_{\pi}^{\ast})^{\top}\nabla_p^2 y_i(W_i,S_i,\tilde{P}_i+U_i)(P_n(\bW)-p_{\pi}^{\ast}). \label{eq:expansion w.r.t. price}
    \end{align}
    Now going back to~\eqref{eq:tau_TOT decomp}, the last quadratic term in~\eqref{eq:expansion w.r.t. price} appears in $\tau_{\MPE}^{\oracle}$ as an error term as follows
    \begin{align}
        &\quad \Delta\\ 
        &= \EE_{\pi}\sbr{ \frac{1}{n}\sum_{i=1}^n (P_n(\bW)-p_{\pi}^{\ast})^{\top}\nabla_p^2 y_i(W_i,S_i,\tilde{P}_i+U_i)(P_n(\bW)-p_{\pi}^{\ast}) \times \sum_{k=1}^n\rbr{W_k-\pi} } \\
        &= \EE_{\pi}\cbr{ (P_n(\bW)-p_{\pi}^{\ast})^{\top} \sbr{ \frac{1}{n}\sum_{i=1}^n \nabla_p^2 y_i(W_i,S_i,\tilde{P}_i+U_i) } (P_n(\bW)-p_{\pi}^{\ast}) \times \sum_{k=1}^n\rbr{W_k-\pi} } \label{eq:quadratic error term in MPE for price}
    \end{align}
    To proceed, we can use such a trivial bound
    \begin{equation}
        \abr{ \frac{1}{n}\sum_{i=1}^n \nabla_p^2 y_i(W_i,S_i,\tilde{P}_i+U_i) } \le B
    \end{equation}
    for the Hessians, and a joint central limit theorem that
\begin{align}
    \begin{bmatrix}
        \sqrt{n}(P_n(\bW)-p_{\pi}^{\ast}) \\
        \sum_{k=1}^n\rbr{W_k-\pi} / \sqrt{n}
    \end{bmatrix}
    &= \frac{1}{\sqrt{n}} \sum_{k=1}^n \begin{bmatrix}
        -\xi_{z}^{-1} z_k(W_k,p_{\pi}^{\ast}+U_k) \\
        W_k - \pi
    \end{bmatrix}
    + O_p(1/\sqrt{n})
    \Rightarrow
    \begin{bmatrix}
        \mathsf{Z} \\
        \mathsf{W}
    \end{bmatrix}.
\end{align}
    where $(\mathsf{Z},\mathsf{W})$ denotes a mean-zero joint normal distribution with a non-zero covariance. In the first equality, we have used the approximation in Lemma~\ref{lemma: convergence of price}. In this way, we can derive a valid upper bound on $\Delta$ in~\eqref{eq:quadratic error term in MPE for price} by
    \begin{align}
        \EE\abr{\Delta} &\le B \EE\sbr{ \nbr{\sqrt{n}(P_n(\bW)-p_{\pi}^{\ast}) }^2 \abr{\frac{1}{n}\sum_{k=1}^n\rbr{W_k-\pi}}} \\
        &\le \frac{B}{\sqrt{n}} \EE\sbr{\mathsf{Z}^2|\mathsf{W}|} + o(1) = o(1).
    \end{align}
    Therefore, we can asymptotically ignore the last term in~\eqref{eq:expansion w.r.t. price}, when we use~\eqref{eq:tau_TOT decomp}. To ease notations, we write
    \begin{equation}
        \bar{y}_i(W_i,S_i,U_i) := y_i(W_i,S_i,p_{\pi}^{\ast}+U_i) + (P_n(\bW)-p_{\pi}^{\ast})^{\top}\nabla_p y_i(W_i,S_i,p_{\pi}^{\ast}+U_i),
    \end{equation}
    which is the linear component in \eqref{eq:expansion w.r.t. price}, so that
    \begin{equation}
        \tau_{\MPE}^{\oracle} = \frac{1}{\pi(1-\pi)} \EE_{\pi}\sbr{\frac{1}{n}\sum_{i=1}^n\bar{y}_i(W_i,S_i,U_i)\times\sum_{k=1}^n\rbr{W_k-\pi}} + o_p(1).
    \end{equation}
    \begin{remark}
        We have noticed a small gap in \cite{munro2021treatment}. In the last part of proving their Lemma 16, the authors mistakenly used $\nbr{P_n(\bW)-p_{\pi}^{\ast}}^2=o_p(1/n)$ to control the same error term $\Delta$ in \eqref{eq:quadratic error term in MPE for price}.

        To overcome this issue, we (a) make stronger regularity assumptions so that a higher-order expansion of $P_n(\bW)-p_{\pi}^{\ast}$ can be established in Lemma~\ref{lemma: convergence of price}; (b) effectively use the fact that $\sum_{k=1}^n\rbr{W_k-\pi} / \sqrt{n}$ converges to a normal distribution so that $\sum_{k=1}^n\rbr{W_k-\pi}$ is of order $\sqrt{n}$ asymptotically.
    \end{remark}

    Then we continue to conduct procedure (ii), which further simplifies $\bar{y}_i$ with the help of Lemma~\ref{lemma: convergence of price}. Let
    \begin{align}
        \check{y}_i(W_i,S_i,U_i) :&= y_i(W_i,S_i,p_{\pi}^{\ast}+U_i) \\
        &\quad-\sbr{\frac{1}{n}\sum_{i=1}^n z_i(W_i,p_{\pi}^{\ast}+U_i)}^{\top}\xi_z^{-\top}\nabla_p y_i(W_i,S_i,p_{\pi}^{\ast}+U_i). \label{eq:y first order expansion to price}
    \end{align}
Since $\abr{\nabla_p y_i}\le B$ and Lemma~\ref{lemma: convergence of price} yields a second-order expansion remainder of order $n^{-1}$, there exists a universal constant $C>0$ such that
\[
\EE \abr{\bar{y}_i(W_i,S_i,U_i)-\check{y}_i(W_i,S_i,U_i)}^2 \le \frac{C}{n^2}.
\]
In particular,
\[
\bar{y}_i(W_i,S_i,U_i)-\check{y}_i(W_i,S_i,U_i)=O_p(1/n).
\]
    Then by Cauchy-Schwarz inequality,
    \begin{align}
        &\quad \EE\abr{ \frac{1}{n}\EE_{\pi}\sbr{\sum_{i=1}^n\rbr{\bar{y}_i(W_i,S_i,U_i)-\check{y}_i(W_i,S_i,U_i)}\times\sum_{k=1}^n\rbr{W_k-\pi}} }\\
        &\le \frac{1}{n} \sum_{i=1}^n \EE\abr{ \rbr{\bar{y}_i(W_i,S_i,U_i)-\check{y}_i(W_i,S_i,U_i)}\times\sum_{k=1}^n\rbr{W_k-\pi} } \\
        &\le \frac{1}{n} \sum_{i=1}^n \sqrt{\EE \abr{\bar{y}_i(W_i,S_i,U_i)-\check{y}_i(W_i,S_i,U_i)}^2 } \times \sqrt{\EE \sbr{\sum_{k=1}^n\rbr{W_k-\pi}}^2} \le \frac{C}{\sqrt{n}}.
    \end{align}
    Therefore,
    \begin{equation}\label{eq: MPE after expansion w.r.t. price}
        \tau_{\MPE}^{\oracle} = \frac{1}{\pi(1-\pi)} \EE_{\pi}\sbr{\frac{1}{n}\sum_{i=1}^n\check{y}_i(W_i,S_i,U_i)\times\sum_{k=1}^n\rbr{W_k-\pi}} + o_p(1).
    \end{equation}

\vspace*{2mm}
\noindent\textit{Step 2. Expansion with respect to local interactions.}
Now we expand $\check{y}_i$ with respect to the local interaction argument $S_i$ around $s=\pi$:
\begin{align}
    &\quad \check{y}_i(W_i,S_i,U_i)
    \\
    &= \check{y}_i(W_i,\pi,U_i)
    + \rbr{S_i-\pi}\nabla_s \check{y}_i(W_i,\pi,U_i)
    + \frac{1}{2}\rbr{S_i-\pi}^2\nabla_s^2 \check{y}_i(W_i,\tilde{S}_i,U_i), \\
    &= y_i(W_i,\pi,p_{\pi}^{\ast}+U_i)
    + \rbr{S_i-\pi}\nabla_s y_i(W_i,\pi,p_{\pi}^{\ast}+U_i) \\
    &\quad -\bar{Z}_n^{\top} \xi_z^{-\top}\nabla_p y_i(W_i,\pi,p_{\pi}^{\ast}+U_i) - (S_i-\pi) \bar{Z}_n^{\top} \xi_z^{-\top}\nabla_s \nabla_p y_i(W_i,\pi,p_{\pi}^{\ast}+U_i) \\
    &\quad + \frac{1}{2}\rbr{S_i-\pi}^2\nabla_s^2 y_i(W_i,\tilde{S}_i,p_{\pi}^{\ast}+U_i) \\
    &\quad - \frac{1}{2} \rbr{S_i-\pi}^2\bar{Z}_n^{\top}\xi_z^{-\top}\nabla_s^2 \nabla_p y_i(W_i,\tilde{S}_i,p_{\pi}^{\ast}+U_i) := \sum_{\iota=1}^6 \mathcal{T}_i^{(\iota)}
    \label{eq:expansion w.r.t. proportion}
\end{align}
where $\tilde{S}_i$ lies between $S_i$ and $\pi$. There are $6$ terms in the last expression~\eqref{eq:expansion w.r.t. proportion}; and we will refer to them by $\mathcal{T}^{(1)}_i,\ldots,\mathcal{T}^{(6)}_i$ for simplicity. Our estimand is thus denoted as
\begin{equation}
    \tau_{\MPE}^{\oracle} = \frac{1}{\pi(1-\pi)} \EE_{\pi}\sbr{\frac{1}{n}\sum_{i=1}^n \sum_{\iota=1}^6  \mathcal{T}^{(\iota)}_i \times\sum_{k=1}^n\rbr{W_k-\pi}} + o_p(1).
\end{equation}

\vspace*{2mm}
\noindent\textit{Step 3. Cancellation of the last three terms.} Via discussing each term case by case, we find that the higher-order derivatives $\mathcal{T}^{(4)}_i$ to $\mathcal{T}^{(6)}_i$ are only bringing vanishing minor terms. We slightly adjust the ordering of dealing the last three terms, based on the complicacy.
\begin{description}
    \item[(v)] Since $\mathcal{T}^{(5)}_i$ actually only depends on the its neighboring treatment $\cbr{W_j:j\in \cN_i}$, the expectation $\EE_{\pi}$ greatly simplifies the summation:
    \begin{align}
        &\quad \EE_{\pi}\sbr{\frac{1}{n}\sum_{i=1}^n  \mathcal{T}^{(5)}_i \times\sum_{k=1}^n\rbr{W_k-\pi}} \\
        &= \frac{1}{2n} \sum_{i=1}^n \EE_{\pi}\sbr{ \rbr{S_i-\pi}^2\nabla_s^2 y_i(W_i,\tilde{S}_i,p_{\pi}^{\ast}+U_i) \times (M_i-\pi N_i)} \label{eq:MPE estimand term 5}
    \end{align}
    Specifically, the expectation of the second-order term can be controlled as follows:
    \begin{align}
        &\quad \abr{\EE_{\pi}\sbr{\rbr{S_i-\pi}^2\nabla_s^2 y_i(W_i,\tilde{S}_i,p_{\pi}^{\ast}+U_i) \rbr{M_i-\pi N_i}}} \\
        &\le \frac{B}{N_i^2}\EE_{\pi}\abr{M_i-\pi N_i}^3
        \le \frac{C}{\sqrt{N_i}},
    \end{align}
    for some constant $C>0$, where we used $S_i-\pi=(M_i-\pi N_i)/N_i$ and the standard third-moment bound 
    $
    \EE_\pi |M_i-\pi N_i|^3 = O(N_i^{3/2})
    $ for sums of centered independent Bernoulli random variables. Consequently, averaging over $i\in[n]$ yields
    \begin{align}
        &\quad \frac{1}{n}\sum_{i=1}^n
        \abr{\EE_{\pi}\sbr{\rbr{S_i-\pi}^2\nabla_s^2 y_i(W_i,\tilde{S}_i,U_i)\rbr{M_i-\pi N_i}}}
        \\
        &\le \cO\rbr{\frac{1}{\sqrt{\min_i N_i}}}
        = \cO\rbr{\frac{1}{\sqrt{n\rho_n}}}
        = o(1).
    \end{align}
    Consequently, the expression in~\eqref{eq:MPE estimand term 5} is vanishing.
    
    \item[(iv)] To deal with the other two terms, it is essential to utilize the following decomposition
    \begin{equation}\label{eq:Zbar neighbor split}
        \bar{Z}_n=\frac{1}{n}\sum_{\ell\in\cN_i\cup\{i\}} z_\ell(W_\ell,p_{\pi}^{\ast}+U_\ell)+\frac{1}{n}\sum_{\ell\notin\cN_i\cup\{i\}} z_\ell(W_\ell,p_{\pi}^{\ast}+U_\ell) := \bar{Z}_{n,\cN_i} + \bar{Z}_{n,\cN_i^c}.
    \end{equation}
    Therefore, we proceed to further split $\mathcal{T}^{(4)}_i$ into
    \begin{align}
        \mathcal{T}^{(4)}_i &= -(S_i-\pi) \; \bar{Z}_{n,\cN_i}^{\top} \; \xi_z^{-\top}\nabla_s \nabla_p y_i(W_i,\pi,p_{\pi}^{\ast}+U_i) \\
        &\quad -(S_i-\pi) \; \bar{Z}_{n,\cN_i^c}^{\top} \; \xi_z^{-\top}\nabla_s \nabla_p y_i(W_i,\pi,p_{\pi}^{\ast}+U_i) := \mathcal{T}^{(41)}_i + \mathcal{T}^{(42)}_i.
    \end{align}
    Note that $\mathcal{T}^{(41)}_i$ now only depends on the neighboring treatments as well; so it holds that
    \begin{align}
        &\quad \EE_{\pi}\sbr{\frac{1}{n}\sum_{i=1}^n \mathcal{T}_i^{(41)} \times\sum_{k=1}^n\rbr{W_k-\pi}} \\
        &= -\EE_{\pi}\sbr{\frac{1}{n}\sum_{i=1}^n \rbr{W_i-\pi}\rbr{S_i-\pi}\bar{Z}_{n,\cN_i}^{\top}\xi_z^{-\top}\nabla_s\nabla_p y_i(W_i,\pi,p_{\pi}^{\ast}+U_i)} \\
        &\quad -\EE_{\pi}\sbr{\frac{1}{n}\sum_{i=1}^n \frac{\rbr{M_i-N_i\pi}^2}{N_i}\bar{Z}_{n,\cN_i}^{\top}\xi_z^{-\top}\nabla_s\nabla_p y_i(W_i,\pi,p_{\pi}^{\ast}+U_i)}.
    \end{align}
    Furthermore, using the fact that $\nbr{\bar{Z}_{n,\cN_i}}\le C\frac{\abr{\cN_i}}{n}$, it follows by a straightforward bound that
    \begin{equation}
        \abr{\EE_{\pi}\sbr{\frac{1}{n}\sum_{i=1}^n \mathcal{T}_i^{(41)} \times\sum_{k=1}^n\rbr{W_k-\pi}}} \le \frac{C}{n}\sum_{i=1}^n \frac{\abr{\cN_i}}{n} = O_p(\rho_n) =o_p(1),
    \end{equation}
    where we have used a trivial fact that $\EE_{\pi}(M_i-N_i\pi)^2=\pi(1-\pi)N_i$.
    We shall focus primarily on $\mathcal{T}^{(42)}_i$ next. As a starting point, using the decomposition that
    \begin{equation}
        \sum_{k=1}^n (W_k-\pi) = (W_i-\pi) + (M_i-N_i\pi) + \sum_{k\notin \cN_i\cup\{i\}}(W_k-\pi),
    \end{equation}
    we can reorganize the summation involving $\mathcal{T}_i^{(42)}$ as follows
        \begin{align}
            &\quad \EE_{\pi}\sbr{\frac{1}{n}\sum_{i=1}^n \mathcal{T}_i^{(42)} \times\sum_{k=1}^n\rbr{W_k-\pi}} \\
            &= -\EE_{\pi}\sbr{\frac{1}{n}\sum_{i=1}^n \rbr{W_i-\pi}\rbr{S_i-\pi}\bar{Z}_{n,\cN_i^c}^{\top}\xi_z^{-\top}\nabla_s\nabla_p y_i(W_i,\pi,p_{\pi}^{\ast}+U_i)} \\
            &\quad -\EE_{\pi}\sbr{\frac{1}{n}\sum_{i=1}^n \frac{\rbr{M_i-N_i\pi}^2}{N_i}\bar{Z}_{n,\cN_i^c}^{\top}\xi_z^{-\top}\nabla_s\nabla_p y_i(W_i,\pi,p_{\pi}^{\ast}+U_i)} \\
            &\quad -\EE_{\pi}\sbr{\frac{1}{n}\sum_{i=1}^n \sum_{k\notin \cN_i\cup\{i\}}(W_k-\pi)(S_i-\pi)\bar{Z}_{n,\cN_i^c}^{\top}\xi_z^{-\top}\nabla_s\nabla_p y_i(W_i,\pi,p_{\pi}^{\ast}+U_i)} \label{eq:MPE estimand term 4}
        \end{align}
        After careful inspection, every term here is actually (asymptotically) vanishing.
        \begin{itemize}
            \item The first term in~\eqref{eq:MPE estimand term 4} is exactly zero since the factor $S_i-\pi$ is independent of every other factor, so
            \begin{equation}
                \EE_{\pi}\sbr{ \rbr{W_i-\pi}\rbr{S_i-\pi}\bar{Z}_{n,\cN_i^c}^{\top}\xi_z^{-\top}\nabla_s\nabla_p y_i(W_i,\pi,p_{\pi}^{\ast}+U_i)} = 0.
            \end{equation}

            \item For the second term, we used again that $\frac{(M_i-N_i\pi)^2}{N_i}$ is independent of everything else, so it holds that
            \begin{align}
                &\quad -\EE_{\pi}\sbr{\frac{1}{n}\sum_{i=1}^n \frac{\rbr{M_i-N_i\pi}^2}{N_i}\bar{Z}_{n,\cN_i^c}^{\top}\xi_z^{-\top}\nabla_s\nabla_p y_i(W_i,\pi,p_{\pi}^{\ast}+U_i)} \\
                &= -\frac{1}{n} \sum_{i=1}^n \EE_{\pi}\rbr{ \frac{\rbr{M_i-N_i\pi}^2}{N_i}} \EE_{\pi}\cbr{\bar{Z}_{n,\cN_i^c}^{\top}\xi_z^{-\top}\nabla_s\nabla_p y_i(W_i,\pi,p_{\pi}^{\ast}+U_i)} \\
                &= -\frac{\pi(1-\pi)}{n} \sum_{i=1}^n \EE_{\pi}\cbr{\bar{Z}_{n,\cN_i^c}^{\top}\xi_z^{-\top}\nabla_s\nabla_p y_i(W_i,\pi,p_{\pi}^{\ast}+U_i)}.
            \end{align}
            Since $\abr{\cN_i}/n = O_p(\rho_n)\to0$, each complement $\cN_i^c$ retains all but a vanishing fraction of the units. The summands $z_\ell(W_\ell,p_{\pi}^{\ast}+U_\ell)$ are uniformly bounded by Assumption~\ref{assump:regularity}, and they are centered up to a negligible perturbation bias: since $\EE\sbr{z_\ell(W_\ell,p_{\pi}^{\ast})}=0$ and $\EE[U_\ell]=0$, a second-order expansion gives $\EE\sbr{z_\ell(W_\ell,p_{\pi}^{\ast}+U_\ell)}=O(h_n^2)=o(n^{-1/2})$ because $\alpha>1/4$. Hoeffding's inequality together with a union bound over $i\in[n]$ therefore yields
            \begin{equation}
                \max_{i\in[n]}\nbr{\bar{Z}_{n,\cN_i^c}} = O_p\rbr{\sqrt{\frac{\log n}{n}}} \pto 0,
            \end{equation}
            so the expression above is $o_p(1)$.

            \item The third term is exactly zero as well, again due to the fact that $S_i-\pi$ is independent of everything else in every summand, and of mean zero.
        \end{itemize}
        By combining these discussions, we would have
        \begin{equation}
            \EE_{\pi}\sbr{\frac{1}{n}\sum_{i=1}^n \mathcal{T}_i^{(4)} \times\sum_{k=1}^n\rbr{W_k-\pi}} = o_p(1).
        \end{equation}

        \item[(vi)] For the last term $\cT_i^{(6)}$, we need to reuse the split of $\bar{Z}_n$ from~\eqref{eq:Zbar neighbor split} to obtain
        \begin{align}
            \mathcal{T}^{(6)}_i &= -\frac{1}{2}(S_i-\pi)^2 \; \bar{Z}_{n,\cN_i}^{\top} \; \xi_z^{-\top}\nabla_s^2 \nabla_p y_i(W_i,\pi,p_{\pi}^{\ast}+U_i) \\
            &\quad -\frac{1}{2}(S_i-\pi)^2 \; \bar{Z}_{n,\cN_i^c}^{\top} \; \xi_z^{-\top}\nabla_s^2 \nabla_p y_i(W_i,\pi,p_{\pi}^{\ast}+U_i) := \mathcal{T}^{(61)}_i + \mathcal{T}^{(62)}_i.
        \end{align}
        We can use the same argument in \textbf{(v)} for $\mathcal{T}^{(61)}_i$; and the same argument of $\mathcal{T}^{(42)}_i$ for $\mathcal{T}^{(62)}_i$. So we omit the detailed derivation here. We will also have $\EE_{\pi}\sbr{\frac{1}{n}\sum_{i=1}^n \mathcal{T}_i^{(6)} \times\sum_{k=1}^n\rbr{W_k-\pi}} = o_p(1)$.

\end{description}


\vspace*{2mm}
\noindent\textit{Step 4. Non-vanishing terms from~\eqref{eq:expansion w.r.t. proportion}}
Departing from~\eqref{eq: MPE after expansion w.r.t. price}, it now suffices to work on $\cT_i^{(1)}$-$\cT_i^{(3)}$ to conclude the whole lemma. Each of them actually leads to a non-vanishing interpretable estimand.
    \begin{description}
        \item[(i)] Using standard law of large numbers, we have that
        \begin{equation}
            \EE_{\pi}\sbr{\frac{1}{n\pi(1-\pi)}\sum_{i=1}^n \cT_i^{(1)}\times\sum_{k=1}^n\rbr{W_k-\pi}} \pto \EE\sbr{y_i(1,\pi,p_{\pi}^{\ast})-y_i(0,\pi,p_{\pi}^{\ast})}.
        \end{equation}

        \item[(ii)] The sum involving $\cT_i^{(2)}$ can be rephrased as
        \begin{align}
            &\quad \EE_{\pi}\sbr{\frac{1}{n\pi(1-\pi)}\sum_{i=1}^n \cT_i^{(2)}\times\sum_{k=1}^n\rbr{W_k-\pi}}
        \end{align}
        is in fact the source of local spillover effect. The effect of global market interference \spillovertype{MAR} is ruled out already. So we can use \citet[Proposition 1]{li2022random} to deduce that this term converges in probability to $\EE\sbr{\pi \nabla_s y_i(1,\pi,p_{\pi}^{\ast})+(1-\pi)\nabla_s y_i(0,\pi,p_{\pi}^{\ast})}$.

        \item[(iii)] The sum involving $\cT_i^{(3)}$ takes the form as
        \begin{align}
            &\quad \EE_{\pi}\sbr{\frac{1}{n\pi(1-\pi)}\sum_{i=1}^n \cT_i^{(3)}\times\sum_{k=1}^n\rbr{W_k-\pi}} \\
            &= -\frac{1}{n\pi(1-\pi)} \EE_{\pi}\cbr{ \sbr{\sum_{k=1}^n\rbr{W_k-\pi}} \sbr{ \sum_{i=1}^n z_i(W_i,p_{\pi}^{\ast}+U_i) }^{\top}\xi_z^{-\top} \sbr{\frac{1}{n}\sum_{i=1}^n \nabla_p y_i(W_i,\pi,p_{\pi}^{\ast}+U_i) } }.
        \end{align}
        Using law of large numbers, the empirical average
        \begin{equation}
            \frac{1}{n}\sum_{i=1}^n \nabla_p y_i(W_i,\pi,p_{\pi}^{\ast}+U_i) \to \xi_y \text{ a.s.}
        \end{equation}
        Then using continuous mapping theorem, the limit of this quantity is identical to that of
        \begin{align}
            &\quad -\frac{1}{n\pi(1-\pi)} \EE_{\pi}\cbr{ \sbr{\sum_{k=1}^n\rbr{W_k-\pi}} \sbr{ \sum_{i=1}^n z_i(W_i,p_{\pi}^{\ast}+U_i) }^{\top}\xi_z^{-\top} \xi_y } \\
            &= -\frac{1}{n} \EE_{\pi}\cbr{ \sbr{ \sum_{i=1}^n z_i(1,p_{\pi}^{\ast}+U_i) - z_i(0,p_{\pi}^{\ast}+U_i) }^{\top}\xi_z^{-\top} \xi_y }.
        \end{align}
        The final limit is then $-\xi_y^{\top}\xi_z^{-1} \EE\sbr{z_i(1,p_{\pi}^{\ast})-z_i(0,p_{\pi}^{\ast})}$.
        
    \end{description}
    After all these derivations, we can finally conclude the asymptotic limit for $\tau_{\MPE}^{\oracle}$.
\end{proof}

\begin{theorem}\label{thm:asymptotics of estimand direct}
    Suppose that Assumptions~\ref{assump:RCT}, \ref{assump:super-population for units} and the assumptions in Section~\ref{sec:assump} all hold.
    The estimand  $\tau_{\ADE}^{\oracle}$ defined in \eqref{eq: define direct effect} converges as follows
    \begin{align}
        \tau_{\ADE}^{\oracle} &\pto \tau_{\ADE}^{\oracle,\ast} := \EE\sbr{y_i(1,\pi,p_{\pi}^{\ast})-y_i(0,\pi,p_{\pi}^{\ast})}.
    \end{align}
\end{theorem}

\begin{proof}
    Based on the last proof, this one only gets easier. Recall that this estimand can be transformed into
    \begin{align}
        \tau_{\ADE}^{\oracle} &= \frac{1}{\pi(1-\pi)} \EE_{\pi}\sbr{\frac{1}{n}\sum_{i=1}^nY_i\rbr{W_i-\pi}} \\
        &= \frac{1}{\pi(1-\pi)} \EE_{\pi}\sbr{\frac{1}{n}\sum_{i=1}^ny_i(W_i,S_i,P_n(\bW)+U_i)\rbr{W_i-\pi}},
    \end{align}
    Expanding $y_i$ linearly on the argument $p$ to find
    \begin{align}
        &\quad y_i(W_i,S_i,P_n(\bW)+U_i) \\
        &= y_i(W_i,S_i,p_{\pi}^{\ast}+U_i) + (P_n(\bW)-p_{\pi}^{\ast})^{\top}\nabla_p y_i(W_i,S_i,\tilde{P}_i+U_i).
    \end{align}
    Via Cauchy-Schwarz inequality,
    \begin{align}
        &\quad \EE_{\pi}\abr{ \frac{1}{n}\sum_{i=1}^n (P_n(\bW)-p_{\pi}^{\ast})^{\top}\nabla_p y_i(W_i,S_i,\tilde{P}_i+U_i) \rbr{W_i-\pi} } \\
        &\le 2\EE_{\pi} \nbr{P_n(\bW)-p_{\pi}^{\ast}} = o_p(1). 
    \end{align}
    Therefore,
    \begin{equation}
        \tau_{\ADE}^{\oracle} = \frac{1}{\pi(1-\pi)} \EE_{\pi}\sbr{\frac{1}{n}\sum_{i=1}^ny_i(W_i,S_i,p_{\pi}^{\ast}+U_i)\rbr{W_i-\pi}} + o_p(1).
    \end{equation}
    Now since the market interference \spillovertype{MAR} has been taken out, we can use Proposition 1 of \cite{li2022random} to deduce that
    \begin{equation}
        \tau_{\ADE}^{\oracle} = \frac{1}{n} \sum_{i=1}^n \sbr{y_i(1,\pi,p_{\pi}^{\ast}+U_i)-y_i(0,\pi,p_{\pi}^{\ast}+U_i)} + o_p(1).
    \end{equation}
    A first-order Taylor expansion around $p_\pi^\ast$ gives
    \begin{align}
        &\quad
        y_i(1,\pi,p_\pi^\ast+U_i)-y_i(0,\pi,p_\pi^\ast+U_i)
        - y_i(1,\pi,p_\pi^\ast)+y_i(0,\pi,p_\pi^\ast) \\
        &=
        \sbr{\nabla_p y_i(1,\pi,p_\pi^\ast)-\nabla_p y_i(0,\pi,p_\pi^\ast)}^\top U_i
        + R_i,
    \end{align}
    where $\abr{R_i}\le B\|U_i\|^2\le B h_n^2$. Since $\nabla_p y_i(1,\pi,p_\pi^\ast)-\nabla_p y_i(0,\pi,p_\pi^\ast)$ is uniformly bounded across $i$, and $U_i$ is independent with zero mean,
    \begin{align}
        &\quad \abr{ \tau_{\ADE}^{\oracle}-\frac{1}{n}\sum_{i=1}^n \sbr{y_i(1,\pi,p_\pi^\ast)-y_i(0,\pi,p_\pi^\ast)} } \\
        &\le \abr{ 
        \frac{1}{n}\sum_{i=1}^n
        \sbr{\nabla_p y_i(1,\pi,p_\pi^\ast)-\nabla_p y_i(0,\pi,p_\pi^\ast)}^\top U_i
        } + \frac{1}{n}\sum_{i=1}^n \abr{R_i} \\
        &= O_p(h_n/\sqrt n) + O_p(h_n^2)
        = o_p(1/\sqrt n).
    \end{align}
    Lastly, by law of large numbers, we have $\frac{1}{n}\sum_{i=1}^n \sbr{y_i(1,\pi,p_\pi^\ast)-y_i(0,\pi,p_\pi^\ast)} = \tau_{\ADE}^{\oracle,\ast}+o_p(1)$ thus concluding the proof.
\end{proof}

\section{Proofs about estimators in Section~\ref{sec:local vs global example estimator}}

\subsection{Estimation of direct effects}
\begin{proof}[Proof of Theorem~\ref{thm:direct effect estimator}]
    By plugging our generative model assumption into our Horvitz Thompson estimator, it becomes
    \begin{align}
         \hat{\tau}_{\ADE}^{\oracle} = \frac{1}{n}\sum_{i=1}^n \rbr{\frac{W_i}{\pi}-\frac{1-W_i}{1-\pi}} y_i\rbr{W_i, S_i, P_n(\bW)+U_i}.
    \end{align}
    To start, we expand $y_i$'s over the price argument,
    \begin{align}
        &\quad y_i(W_i,S_i,P_n(\bW)+U_i) \\
            &= y_i(W_i,S_i,p_{\pi}^{\ast}+U_i) + \rbr{ P_n(\bW)-p_{\pi}^{\ast} }^{\top}\nabla_p y_i(W_i,S_i,p_{\pi}^{\ast}+U_i) \\
            &\qquad +\frac{1}{2} (P_n(\bW)-p_{\pi}^{\ast})^{\top}\nabla_p^2 y_i(W_i,S_i,\tilde{P}_i+U_i)(P_n(\bW)-p_{\pi}^{\ast}). 
    \end{align}
    Then our estimator can be decomposed into several separate terms,
    \begin{align}
        \hat{\tau}_{\ADE}^{\oracle} &=\sbr{ \frac{1}{n}\sum_{i=1}^n \rbr{\frac{W_i}{\pi}-\frac{1-W_i}{1-\pi}} y_i\rbr{W_i, S_i, p_\pi^\ast+U_i} } \\
        &\quad +\rbr{ P_n(\bW)-p_{\pi}^{\ast} }^{\top} \sbr{\frac{1}{n}\sum_{i=1}^n \rbr{\frac{W_i}{\pi}-\frac{1-W_i}{1-\pi}} \nabla_p y_i(W_i,S_i,p_{\pi}^{\ast}+U_i) } \\
        &\quad + \frac{1}{2} \rbr{ P_n(\bW)-p_{\pi}^{\ast} }^{\top} \sbr{\frac{1}{n}\sum_{i=1}^n \rbr{\frac{W_i}{\pi}-\frac{1-W_i}{1-\pi}} \nabla_p^2 y_i(W_i,S_i,\tilde{p}_i+U_i) } \rbr{ P_n(\bW)-p_{\pi}^{\ast} } \\
        &= \mathcal{A}+\mathcal{B}+\mathcal{C}. \label{eq: direct eff estimator expansion over price}
    \end{align}
    As consistent to our theorem statement, we define two set of random variables:
    \begin{align}
        \mathcal{V}_{i}^{(1)} &= y_i(1,\pi,p_\pi^\ast)-y_i(0,\pi,p_\pi^\ast), \\
        \mathcal{V}_{i}^{(2)} &= -\nabla_p\sbr{\EE y(1,\pi,p_\pi^\ast)-\EE y(0,\pi,p_\pi^\ast)}^\top \xi_z^{-1} \sbr{\pi z_i(1,p_\pi^\ast)+(1-\pi)z_i(0,p_\pi^\ast)}
    \end{align}
    and also
    \begin{align}
        V^{(1)}_i &= \frac{y_i(1,\pi,p_\pi^\ast)}{\pi}+\frac{y_i(0,\pi,p_\pi^\ast)}{1-\pi}, \\
        V^{(2)}_i &= \EE_{Q_j,y_j}\sbr{\frac{G(Q_i,Q_j)\sbr{\nabla_s y_j(1,\pi,p_\pi^\ast) - \nabla_s y_j(0,\pi,p_\pi^\ast)}}{g(Q_j)}\bigg|Q_i}, \\
        V^{(3)}_i &= -\nabla_p\sbr{\EE y(1,\pi,p_\pi^\ast)-\EE y(0,\pi,p_\pi^\ast)}^\top \xi_z^{-1} \sbr{z_i(1,p_\pi^\ast)-z_i(0,p_\pi^\ast)}.
    \end{align}
    Subsequently, we will take separate steps to show that
    \begin{align}
        \mathcal{A} &= \frac{1}{n}\sum_{i=1}^n \mathcal{V}_{i}^{(1)} + \frac{1}{n} \sum_{i=1}^n \sbr{ (W_i-\pi)\rbr{ V^{(1)}_i + V^{(2)}_i } } + o_p(1/\sqrt{n}), \label{eq:ADE estimator term A}\\
        \mathcal{B} &= \frac{1}{n}\sum_{i=1}^n \mathcal{V}_{i}^{(2)} + \frac{1}{n} \sum_{i=1}^n \sbr{ (W_i-\pi) V^{(3)}_i } + o_p(1/\sqrt{n}), \label{eq:ADE estimator term B}\\
        \mathcal{C} &= o_p(1/\sqrt{n}). \label{eq:ADE estimator term C}
    \end{align}

    \vspace{2mm}
    \noindent\textit{Step 1.} Directly apply Theorem 4 in \cite{li2022random} to the term $\mathcal{A}$ to obtain that
    \begin{align}
        \mathcal{A} &= \frac{1}{n}\sum_{i=1}^n \mathcal{\bar{V}}_{i}^{(1)} + \frac{1}{n} \sum_{i=1}^n \sbr{ (W_i-\pi)\rbr{ \bar{V}^{(1)}_i + \bar{V}^{(2)}_i } } + o_p(1/\sqrt{n}),
    \end{align}
    where we can formally write $\mathcal{\bar{V}}^{(1)}_i,\bar{V}^{(1)}_i,\bar{V}^{(2)}_i$ as follows:
    \begin{align}
        \mathcal{\bar{V}}_{i}^{(1)} &= y_i(1,\pi,p_\pi^\ast+U_i)-y_i(0,\pi,p_\pi^\ast+U_i), \\
        \bar{V}^{(1)}_i &= \frac{y_i(1,\pi,p_\pi^\ast+U_i)}{\pi}+\frac{y_i(0,\pi,p_\pi^\ast+U_i)}{1-\pi}, \\
        \bar{V}^{(2)}_i
            &= \EE_{Q_j,y_j}\sbr{
            \frac{G(Q_i,Q_j)\sbr{\nabla_s y_j(1,\pi,p_\pi^\ast+U_j)-\nabla_s y_j(0,\pi,p_\pi^\ast+U_j)}}{g(Q_j)}
            \bigg| Q_i }.
    \end{align}
    These new notions are meant to differ from the ones used before, by including those individualized price perturbations $U_i$.
    On one hand, a first-order Taylor expansion around $p_\pi^\ast$ gives
    \begin{align}
        &\quad
        y_i(1,\pi,p_\pi^\ast+U_i)-y_i(0,\pi,p_\pi^\ast+U_i)
        - y_i(1,\pi,p_\pi^\ast)+y_i(0,\pi,p_\pi^\ast) \\
        &=
        \sbr{\nabla_p y_i(1,\pi,p_\pi^\ast)-\nabla_p y_i(0,\pi,p_\pi^\ast)}^\top U_i
        + R_i,
    \end{align}
    where $\abr{R_i}\le B\|U_i\|^2\le B h_n^2$. Since $\nabla_p y_i(1,\pi,p_\pi^\ast)-\nabla_p y_i(0,\pi,p_\pi^\ast)$ is uniformly bounded across $i$, and $U_i$ is independent with zero mean,
    \begin{align}
        &\quad \abr{ \bar\tau_{\ADE}-\frac{1}{n}\sum_{i=1}^n \sbr{y_i(1,\pi,p_\pi^\ast)-y_i(0,\pi,p_\pi^\ast)} } \\
        &\le \abr{ 
        \frac{1}{n}\sum_{i=1}^n
        \sbr{\nabla_p y_i(1,\pi,p_\pi^\ast)-\nabla_p y_i(0,\pi,p_\pi^\ast)}^\top U_i
        } + \frac{1}{n}\sum_{i=1}^n \abr{R_i} \\
        &= O_p(h_n/\sqrt n) + O_p(h_n^2)
        = o_p(1/\sqrt n).
    \end{align}
    On the other hand, there exists a universal constant $C>0$ such that for $k=1,2$,
    \begin{equation}
        \abr{\bar V_i^{(k)}-V_i^{(k)}} \le C\|U_i\|
        \quad \text{almost surely},
    \end{equation}
    by simply using the regularity in Assumption~\ref{assump:regularity}.
    Hence, using $\EE[W_i-\pi]=0$ and independence across $i$,
    \begin{align}
        \EE\abr{
            \frac{1}{n}\sum_{i=1}^n (W_i-\pi)\rbr{\bar V_i^{(k)}-V_i^{(k)}}
        }^2
        &\le \frac{1}{n^2}\sum_{i=1}^n \EE\sbr{(W_i-\pi)^2\abr{\bar V_i^{(k)}-V_i^{(k)}}^2} \\
        &\le \frac{C h_n^2}{n}
        = o(1/n).
    \end{align}
    Therefore, for each $k=1,2$,
    \begin{equation}
        \frac{1}{n}\sum_{i=1}^n (W_i-\pi)\rbr{\bar V_i^{(k)}-V_i^{(k)}}
        = o_p(1/\sqrt n).
    \end{equation}
    In conclusion, we have managed to establish~\eqref{eq:ADE estimator term A}.

    \vspace{2mm}
    \noindent\textit{Step 2.} To deal with the first-order term $\mathcal{B}$, we introduce two intermediate quantities
    \begin{align}
        \tilde{\mathcal{B}} &= \rbr{ P_n(\bW)-p_{\pi}^{\ast} }^{\top} \sbr{\frac{1}{n}\sum_{i=1}^n \rbr{\frac{W_i}{\pi}-\frac{1-W_i}{1-\pi}} \nabla_p y_i(W_i,\pi,p_{\pi}^{\ast}+U_i) }, 
    \end{align}
    The intermediate quantity $\tilde{\mathcal{B}}$ differs from $\mathcal{B}$ by replacing $S_i$ with $\pi$ in the second argument.
    Since $\nbr{\nabla_s\nabla_p y_i}\le B$, the difference is upper bounded directly,
    \begin{align}
        \EE\abr{\mathcal{B}-\tilde{\mathcal{B}}}
        &\le \frac{B}{\pi(1-\pi)}
        \sqrt{\EE \nbr{P_n(\bW)-p_\pi^\ast}^2}
        \cdot
        \sqrt{\frac{1}{n}\sum_{i=1}^n \EE\abr{S_i-\pi}^2} \\
        &\le \frac{B}{\pi(1-\pi)}
        \cdot\sqrt{\frac{C}{n}}
        \cdot\sqrt{\frac{C}{n\rho_n}}
        = o\rbr{1/\sqrt{n}},
    \end{align}
    where we have used Lemma~\ref{lemma: convergence of price} and Lemma 15 (c) in \cite{li2022random}. By Markov's inequality, $\mathcal{B}-\tilde{\mathcal{B}}=o_p(1/\sqrt{n})$. Now focusing on $\tilde{\mathcal{B}}$ itself, we note that
    \begin{align}
        &\quad \frac{1}{n}\sum_{i=1}^n \rbr{\frac{W_i}{\pi}-\frac{1-W_i}{1-\pi}} \nabla_p y_i(W_i,\pi,p_{\pi}^{\ast}+U_i) \\
        &= \EE\sbr{\nabla_p y_i(1,\pi,p_{\pi}^{\ast}) - \nabla_p y_i(0,\pi,p_{\pi}^{\ast})} + O_p(1/\sqrt{n}),
    \end{align}
    which is deduced from law of large numbers and the same argument for getting rid of $U_i$ as in \textit{Step 1}. From Lemma~\ref{lemma: convergence of price}, we learn that $P_n(\bW)-p_{\pi}^{\ast}$ alone admits a decomposition in the form of
    \begin{align}
        P_n(\bW)-p_\pi^\ast = -\xi_z^{-1}\sbr{\frac{1}{n}\sum_{i=1}^n z_i(W_i,p_{\pi}^{\ast}+U_i)} + o_p(1/\sqrt{n}).
    \end{align}
    In addition, since $z_i$ is continuously differentiable in the price argument,
    \begin{equation}
        z_i(W_i,p_{\pi}^{\ast}+U_i) = z_i(W_i,p_{\pi}^{\ast}) + \nabla z_i(W_i,p_{\pi}^{\ast})^\top U_i + R_i,
    \end{equation}
    where $\abr{R_i}\le B\nbr{U_i}^2\le Bh_n^2$. Therefore, it holds that
    \begin{align}
        \frac{1}{n}\sum_{i=1}^n z_i(W_i,p_{\pi}^{\ast}+U_i) &= \frac{1}{n}\sum_{i=1}^n z_i(W_i,p_{\pi}^{\ast}) +O_p(h_n/\sqrt{n})+O(h_n^2) \\
        &= \frac{1}{n}\sum_{i=1}^n z_i(W_i,p_{\pi}^{\ast}) +o_p(1/\sqrt{n}) \\
        &= \frac{1}{n}\sum_{i=1}^n (W_i-\pi)\sbr{z_i(1,p_{\pi}^{\ast})-z_i(0,p_{\pi}^{\ast})} + \frac{1}{n}\sum_{i=1}^n m_i +o_p(1/\sqrt{n}), \\
        \text{for } m_i &:= \pi z_i(1,p_{\pi}^{\ast}) + (1-\pi)z_i(0,p_{\pi}^{\ast}).
    \end{align}
    Our second step is due to $h_n< n^{-1/4}$ and the last step uses the exact identity $z_i(W_i,p_{\pi}^{\ast}) = (W_i-\pi)\sbr{z_i(1,p_{\pi}^{\ast})-z_i(0,p_{\pi}^{\ast})} + m_i$.
    The definition of $p_{\pi}^{\ast}$ is exactly the population-clearing condition 
    \begin{equation}
        \EE\sbr{ \pi z_i(1,p_{\pi}^{\ast}) + (1-\pi)z_i(0,p_{\pi}^{\ast}) } = \E\sbr{m_i} = 0.
    \end{equation}
    We can finally conclude that
    \begin{align}
        \tilde{\mathcal{B}} &= \frac{1}{n} \sum_{i=1}^n \mathcal{V}_{i}^{(2)} + \frac{1}{n} \sum_{i=1}^n (W_i-\pi) V_{i}^{(3)} + o_p(1/\sqrt{n}).
    \end{align}
    So~\eqref{eq:ADE estimator term B} also holds. The fact that $\EE[m_i]=0$ also leads to $\EE[\mathcal{V}_{i}^{(2)}]=0$ which we will use later.

    \vspace{2mm}
    \noindent\textit{Step 3.} Our last step is to show that $\mathcal{C}$ is negligible as $n\to\infty$. Since $\nbr{\nabla_p^2 y(w,s,p)}\le B$ is uniformly bounded for any $(w,s,p)$, we find that
    \begin{align}
        &\quad \EE\abr{\mathcal{C}} \\
        &= \frac{1}{2} \EE\abr{ \rbr{ P_n(\bW)-p_{\pi}^{\ast} }^{\top}
        \sbr{\frac{1}{n}\sum_{i=1}^n \rbr{\frac{W_i}{\pi}-\frac{1-W_i}{1-\pi}} \nabla_p^2 y_i(W_i,S_i,\tilde{p}_i+U_i) }
        \rbr{ P_n(\bW)-p_{\pi}^{\ast} } }\\
        &\le C_0\,\EE\nbr{P_n(\bW)-p_{\pi}^{\ast}}^2
        = O\rbr{1/n},
    \end{align}
    where we have used Lemma~\ref{lemma: convergence of price}. By Markov's inequality, $\mathcal{C}=O_p(1/n)=o_p(1/\sqrt{n})$, which is~\eqref{eq:ADE estimator term C}.

    \vspace*{2mm}\noindent\textbf{Concluding this theorem.} By plugging \eqref{eq:ADE estimator term A}-\eqref{eq:ADE estimator term C} into~\eqref{eq: direct eff estimator expansion over price}, it holds that
    \begin{equation}
        \hat{\tau}_{\ADE}^{\oracle} = \frac{1}{n} \sum_{i=1}^n \sbr{ \mathcal{V}^{(1)}_i+\mathcal{V}^{(2)}_i+ (W_i-\pi)\rbr{ V^{(1)}_i + V^{(2)}_i + V^{(3)}_i } } + o_p(1/\sqrt{n}).
    \end{equation}
    The following statements are immediately correct:
    \begin{enumerate}
        \item Since $(y_i,z_i)$ is jointly drawn i.i.d from a fixed population, and $p_{\pi}^{\ast}$ is a market-clearing price, it holds that
        \begin{align}
            \EE \sbr{ \mathcal{V}^{(1)}_i } &= \EE \sbr{y(1,\pi,p_\pi^\ast)-y(0,\pi,p_\pi^\ast)} = \tau^{\oracle,\ast}_{\ADE}, \\
            \EE \sbr{\mathcal{V}^{(2)}_i} &= 0.
        \end{align}
        And each $\mathcal{V}^{(1)}_i+\mathcal{V}^{(2)}_i$ has bounded variance.
        \item The other sum $(W_i-\pi)\rbr{ V^{(1)}_i + V^{(2)}_i + V^{(3)}_i }$ is of mean zero and uncorrelated to $\mathcal{V}^{(1)}_i+\mathcal{V}^{(2)}_i$ because it has a $W_i-\pi$ factor; and $W_i$ is drawn independently from the other random functionals/latent-variables $(y_i,z_i,Q_i)$. Therefore, this sum only adds additional variance to our estimator.
    \end{enumerate}
    Consequently, we can conclude this theorem statement.
\end{proof}

\subsection{Estimation of local spillover effects}\label{sec: estimator proof local}

The following is a formal presentation of Theorem~\ref{thm:local spillover effect estimator informal} in the main text.
\begin{theorem}\label{thm:local spillover effect estimator formal}
    Under assumptions detailed in Section~\ref{sec:assump}, the PC-balancing estimator $\hat{\tau}_{\AIE}^{\rmL}$ has a limiting Gaussian distribution around the asymptotic local spillover estimand $\tau_{\AIE}^{\rmL,\ast}$,
    \begin{equation}
        \frac{1}{\sqrt{\rho_n}} \rbr{ \hat{\tau}_{\AIE}^{\rmL} - \tau_{\AIE}^{\rmL,\ast} } \Rightarrow \cN\rbr{0, \mathsf{V}_{\rmL}},
    \end{equation}
    where the variance $\mathsf{V}_{\rmL}$ is given as
    \begin{align}
        \mathsf{V}_{\rmL} &= \EE\sbr{G(Q_1,Q_2)(\alpha_1^2+\alpha_1\alpha_2)} + \EE\sbr{g(Q_1)\eta_1^2}/(\pi(1-\pi)), \\
        \alpha_i &= y_i(1,\pi,p_{\pi}^{\ast}) - y_i(0,\pi,p_{\pi}^{\ast}), \\
        b_i &= \pi y_i(1,\pi,p_{\pi}^{\ast}) + (1-\pi) y_i(0,\pi,p_{\pi}^{\ast}), \\
        \eta_i &= b_i - \sum_{k=1}^r \EE\sbr{b_i\psi_k(Q_i)} \psi_k(Q_i),
    \end{align}
    where $g(Q):=\EE_{Q_1} G(Q,Q_1)$ is the marginal of $G$. Recall that $\rho_n=cn^{-\kappa}$ is introduced in Condition~\ref{assump:graphon} as network density.
\end{theorem}

Before showing proof, we recall some notations. Compute $\hat{\bPsi}\in\RR^{n\times r}$ as the (normalized)\footnote{The eigenvectors are normalized so that $\hat{\bPsi}^\top\hat{\bPsi}=\Ib_r$.} top-$r$ eigenvectors of the observed adjacency matrix $\bE=(E_{ij})$. Form a raw weight vector $\bnu\in\RR^n$ with $\nu_i=\frac{M_i}{\pi}-\frac{N_i-M_i}{1-\pi}$. Then derive a PC-balancing weight vector
\begin{equation}
    \bnu^{\PC} = \rbr{\Ib_n-\hat{\bPsi}\hat{\bPsi}^\top}\bnu \in\RR^n,
\end{equation}
and output a weighted average
\begin{equation}
    \hat{\tau}_{\AIE}^{\rmL} = \frac{1}{n}\sum_{i=1}^n \nu_i^{\PC} Y_i \in\RR.
\end{equation}

\begin{proof}[Proof of Theorem~\ref{thm:local spillover effect estimator formal}]
    By a Taylor expansion on the potential outcomes, the estimator is then decomposed into three different terms,
    \begin{align}
        \hat{\tau}_{\AIE}^{\rmL} &= \frac{1}{n}\sum_{i=1}^n \nu_i^{\PC} y_i(W_i,S_i,P_n(\bW)+U_i) \\
        &= \sbr{\frac{1}{n}\sum_{i=1}^n \nu_i^{\PC} y_i(W_i,S_i,p_\pi^\ast+U_i)} \\
        &\quad +\rbr{P_n(\bW)-p_\pi^\ast}^\top\sbr{\frac{1}{n}\sum_{i=1}^n \nu_i^{\PC} \nabla_p y_i(W_i,S_i,p_\pi^\ast+U_i)} \\
        &\quad +\frac{1}{2}\rbr{P_n(\bW)-p_\pi^\ast}^\top\sbr{\frac{1}{n}\sum_{i=1}^n \nu_i^{\PC} \nabla_p^2 y_i(W_i,S_i,\tilde{p}_i+U_i)}\rbr{P_n(\bW)-p_\pi^\ast} \\
        &= \mathcal{A}+\mathcal{B}+\mathcal{C}.
    \end{align}
    
    \vspace{2mm}
    \noindent\textit{Step 1.} Directly apply Theorem 6 and Proposition 13 in \cite{li2022random} to the term $\mathcal{A}$ to obtain that
    \begin{align}
        \mathcal{A}-\tau_{\AIE}^{\rmL,\ast} &= \frac{1}{n\pi(1-\pi)} \sum_{i,j:i\neq j} (W_i-\pi) E_{ij} \beta_j + o_p(\sqrt{\rho_n}), \\
        \beta_j &= (W_j-\pi)\sbr{y_j(1,\pi,p_\pi^\ast)-y_j(0,\pi,p_\pi^\ast)} + \tilde{\eta}_j
    \end{align}
    where $\tilde{\eta}_j$ is a residual term (derived by regressing $\pi y(1,\pi,p_\pi^\ast) + (1-\pi)y(0,\pi,p_\pi^\ast)$ onto the population-level principal components). Moreover, the summand has such an asymptotic normal distribution,
    \begin{equation}
        \frac{1}{\sqrt{\rho_n}}\cdot\frac{1}{n\pi(1-\pi)} \sum_{i,j:i\neq j} (W_i-\pi) E_{ij} \beta_j \Rightarrow \cN(0,\mathsf{V}_{\rmL}),
    \end{equation}
    where the variance term is given as
    \begin{align}
        \mathsf{V}_{\rmL} &= \EE\sbr{G(Q_1,Q_2)(\alpha_1^2+\alpha_1\alpha_2)} + \EE\sbr{g(Q_1)\eta_1^2}/(\pi(1-\pi)), \\
        \alpha_i &= y_i(1,\pi,p_{\pi}^{\ast}) - y_i(0,\pi,p_{\pi}^{\ast}), \\
        b_i &= \pi y_i(1,\pi,p_{\pi}^{\ast}) + (1-\pi) y_i(0,\pi,p_{\pi}^{\ast}), \\
        \eta_i &= b_i - \sum_{k=1}^r \EE\sbr{b_i\psi_k(Q_i)} \psi_k(Q_i).
    \end{align}
    
    \vspace{2mm}
    \noindent\textit{Step 2.} Using the same tools on $\mathcal{A}$, we find that
    \begin{align}
        &\quad \frac{1}{n}\sum_{i=1}^n \nu_i^{\PC} \nabla_p y_i(W_i,S_i,p_\pi^\ast+U_i) \\
        &= \EE\sbr{\pi \nabla_s\nabla_p y_i(1,\pi,p_{\pi}^{\ast})+(1-\pi)\nabla_s\nabla_p y_i(0,\pi,p_{\pi}^{\ast})} + O_p(\sqrt{\rho_n}).
    \end{align}
    In the meantime, Lemma~\ref{lemma: convergence of price} implies that $P_n(\bW)-p_\pi^\ast=O_p(1/\sqrt{n})$. Therefore, $\mathcal{B}=O_p(1/\sqrt{n})=o_p(\sqrt{\rho_n})$ is in fact negligible. Note that Assumption~\ref{assump:regularity} has made every $\nabla_p y_i$ to be sufficiently smooth in $s$ so that the tool in \textit{Step 1.} is indeed applicable.
    
    \vspace{2mm}
    \noindent\textit{Step 3.}
    Since $\nbr{\nabla_p y_i}$ is uniformly bounded by $B$, we can use Cauchy-Schwarz inequality to obtain
    \begin{align}
        \EE\abr{\mathcal{C}} &\le \frac{B}{2n}\sum_{i=1}^n \EE\sbr{\abr{\nu_i^\PC} \cdot\nbr{P_n(\bW)-p_\pi^\ast}^2} \\
        &\le \frac{B}{n} \sum_{i=1}^n \sqrt{\EE \abr{\nu_i^\PC}^2} \cdot \sqrt{\EE \nbr{P_n(\bW)-p_\pi^\ast}^4} \\
        &\le B\sqrt{\frac{1}{n}\EE\nbr{\bnu^{\PC}}^2} \cdot \sqrt{\EE \nbr{P_n(\bW)-p_\pi^\ast}^4}
    \end{align}
    With Lemma 15 in \cite{li2022random} providing an upper bound on each separate entry of $\bnu$, the norm of $\bnu^{\PC}$ can be controlled as follows
    \begin{align}
        \EE\nbr{\bnu^{\PC}}^2 &\le \EE\nbr{\bnu}^2=\sum_{i=1}^n \EE\rbr{\frac{M_i}{\pi}-\frac{N_i-M_i}{1-\pi}}^2 \\
        &= \sum_{i=1}^n \EE \sbr{\frac{N_i}{\pi(1-\pi)}} = O(n^2\rho_n).  
    \end{align}
    Moreover, by Lemma~\ref{lemma: convergence of price}, we know $\EE \nbr{P_n(\bW)-p_\pi^\ast}^4 = O(1/n^2)$. Therefore, $\mathcal{C} = O_p(\sqrt{\rho_n/n})=o_p(\sqrt{\rho_n})$.
\end{proof}

\subsection{Estimation of global spillover effects}\label{sec: estimator proof global}
The following is a formal presentation of Theorem~\ref{thm:global spillover effect estimator informal} in the main text.
\begin{theorem}\label{thm:global spillover effect estimator formal}
    Under assumptions detailed in Section~\ref{sec:assump}, the estimator $\hat{\tau}_{\AIE}^{\rmG}$ has a limiting Gaussian distribution around the asymptotic global spillover estimand $\tau_{\AIE}^{\rmG,\ast}$,
    \begin{equation}
        h_n\sqrt{n} \rbr{ \hat{\tau}_{\AIE}^{\rmG} - \tau_{\AIE}^{\rmG,\ast} } \Rightarrow \cN\rbr{0, \mathsf{V}_{\rmG}},
    \end{equation}
    where the variance $\mathsf{V}_{\rmG}$ is given as
    \begin{align}
        \mathsf{V}_{\rmG} &= \psi^\top \EE\sbr{ \rbr{y(W,\pi,p_{\pi}^{\ast})-z(W,p_{\pi}^{\ast})^\top\xi_z^{-\top}\xi_y}^2 I_J } \psi, \\
        \psi &= \xi_z^{-1}\EE\sbr{ z(1,p_{\pi}^{\ast}) - z(0,p_{\pi}^{\ast}) }.
    \end{align}
    Recall that $h_n=cn^{-\alpha}$ is introduced in Condition~\ref{assump: augmented trial} as the magnitude of individualized price perturbations.
\end{theorem}

\begin{proof}
To study the global-spillover part, the observed outcomes come in the form of estimating the price elasticities $\xi_z^{-\top}\xi_y$ via $\hat{\gamma} = \rbr{\bU^\top\bZ}^{-1}\rbr{\bU^\top\bY}$.
For convenience, we denote $U_i=h_n\tilde{U}_i$, then each $\tilde{U}_{ij}\sim\mathrm{Unif}\rbr{\cbr{\pm 1}}$ for any $i\in[n]$ and $j\in[J]$.
By plugging our definitions of the potential outcomes
\begin{align}
    \rbr{\tilde{\bU}^\top\tilde{\bU}}^{-1}\tilde{\bU}^\top\bY &=\rbr{\tilde{\bU}^\top\tilde{\bU}}^{-1}\tilde{\bU}^\top\vec\sbr{ y_i\rbr{W_i,S_i,P_n(\bW)+h_n\tilde{U}_i} }\\
    &= \rbr{\tilde{\bU}^\top\tilde{\bU}}^{-1}\tilde{\bU}^\top\vec\sbr{ y_i\rbr{W_i,\pi,P_n(\bW)+h_n\tilde{U}_i} } \\
    &\quad + \rbr{\tilde{\bU}^\top\tilde{\bU}}^{-1}\tilde{\bU}^\top\vec\sbr{ (S_i-\pi)\nabla_s y_i\rbr{W_i,\pi,P_n(\bW)+h_n\tilde{U}_i} } \\
    &\quad +\frac{1}{2} \rbr{\tilde{\bU}^\top\tilde{\bU}}^{-1}\tilde{\bU}^\top\vec\sbr{ (S_i-\pi)^2\nabla_s^2 y_i\rbr{W_i,\tilde{S}_i,P_n(\bW)+h_n\tilde{U}_i} } \\
    &=\mathcal{A}+\mathcal{B}+\mathcal{C}\in\RR^J.
\end{align}
For convenience, we also let
\begin{equation}
    \mathcal{D} = \rbr{\tilde{\bU}^\top\tilde{\bU}}^{-1}\tilde{\bU}^\top\bZ = \rbr{\tilde{\bU}^\top\tilde{\bU}}^{-1}\tilde{\bU}^\top\vec\sbr{ z_i\rbr{W_i,P_n(\bW)+h_n\tilde{U}_i} } \in \RR^{J \times J}.
\end{equation}
Henceforth, we can write $\hat{\gamma}=\mathcal{D}^{-1}\rbr{\mathcal{A}+\mathcal{B}+\mathcal{C}}$.

\vspace{2mm}
\noindent\textit{Step 1.} Directly using the results in section B.5 (which is the proof of their main Theorem 7) from \cite{munro2021treatment}, we find that
\begin{align}
    \mathcal{A} &= h_n\xi_y +  \rbr{\tilde{\bU}^\top\tilde{\bU}}^{-1}\tilde{\bU}^\top \Vec\sbr{y_i(W_i,\pi,p_\pi^\ast)} + o_p(1/\sqrt{n}),\\
    \mathcal{D} &= h_n\xi_z^{\top} +  \rbr{\tilde{\bU}^\top\tilde{\bU}}^{-1}\tilde{\bU}^\top \Vec\sbr{z_i(W_i,p_\pi^\ast)} + o_p(1/\sqrt{n}).
\end{align}
Recall that in Assumption~\ref{assump:RCT}, we have set $h_n=cn^{-\alpha}$ with $\frac{1}{4}<\alpha<\frac{1}{2}$.

\vspace{2mm}
\noindent\textit{Step 2.} Start from
\begin{align}
    &\quad\frac{1}{n}\tilde{\bU}^\top\vec\sbr{ (S_i-\pi)\nabla_s y_i\rbr{W_i,\pi,P_n(\bW)+h_n\tilde{U}_i} } \\
    &= \frac{1}{n}\tilde{\bU}^\top\vec\sbr{ (S_i-\pi)\nabla_s y_i\rbr{W_i,\pi,p_\pi^\ast+h_n\tilde{U}_i} } +o_p(1/\sqrt{n}) \\
    &= \frac{1}{n}\sum_{i=1}^n (S_i-\pi)\tilde{U}_i \nabla_s y_i\rbr{W_i,\pi,p_\pi^\ast+h_n\tilde{U}_i} +o_p(1/\sqrt{n}).
\end{align}
We denote the major term as
\begin{equation}
    \zeta_1 := \frac{1}{n}\sum_{i=1}^n \rbr{\frac{M_i}{N_i}-\pi}\tilde{U}_i \nabla_s y_i\rbr{W_i,\pi,p_\pi^\ast+h_n\tilde{U}_i}.
\end{equation}
and a tight approximation of it by
\begin{equation}
    \zeta_2 := \frac{1}{n}\sum_{i=1}^n \rbr{\frac{M_i}{N_i}-\pi}\tilde{U}_i \nabla_s y_i\rbr{W_i,\pi,p_\pi^\ast},
\end{equation}
which cancels out the influence of $h_n\tilde{U}_i$ onto every $\nabla_s y_i(\cdot)$. Their difference is then denoted as
\begin{align}
    \zeta_1-\zeta_2 &= \frac{1}{n}\sum_{i=1}^n (S_i-\pi)\,\tilde{U}_i\, d_i, \\
    d_i &:= \nabla_s y_i\rbr{W_i,\pi,p_\pi^\ast+h_n\tilde{U}_i} - \nabla_s y_i\rbr{W_i,\pi,p_\pi^\ast}.
\end{align}
Here $\abr{d_i}\le B\nbr{h_n\tilde{U}_i}=O(h_n)$ by Assumption~\ref{assump:regularity}. Further by the Cauchy--Schwarz inequality,
\begin{align}
    \abr{\zeta_1-\zeta_2}
    &\le \rbr{\frac{1}{n}\sum_{i=1}^n (S_i-\pi)^2}^{1/2}
    \rbr{\frac{1}{n}\sum_{i=1}^n \nbr{\tilde{U}_i d_i}^2}^{1/2}
    \\
    &= O_p\rbr{\frac{1}{\sqrt{n\rho_n}}}\cdot O_p(h_n)
    = O_p\rbr{\frac{h_n}{\sqrt{n\rho_n}}},
\end{align}
where $\frac{1}{n}\sum_i(S_i-\pi)^2=O_p(1/(n\rho_n))$ because $\EE[(S_i-\pi)^2]=\EE[\pi(1-\pi)/N_i]=O(1/(n\rho_n))$ by Lemma~15 of \cite{li2022random} (which uses $g(q)\ge c_g>0$ from Condition~\ref{assump:graphon}). Since $h_n/\sqrt{\rho_n}=n^{-(\alpha-\kappa/2)}\to0$ under $\alpha>1/4$ and $\kappa<1/2$ (so that $\alpha>\kappa/2$), it follows that $\zeta_1-\zeta_2=o_p(1/\sqrt{n})$.

To proceed, we can plug in $S_i=M_i/N_i$ so that $\zeta_2$ can be reorganized as
\begin{equation}
    \zeta_2 = \frac{1}{n}\sum_{j=1}^n (W_j-\pi)\sum_{i\neq j} \frac{E_{ij}}{\sum_{k\neq i}E_{ik}} \tilde{U}_i \nabla_s y_i\rbr{W_i,\pi,p_\pi^\ast}.
\end{equation}
Lastly, since
\begin{equation}
    \chi_{j} = \EE\sbr{\frac{E_{ij}}{g_n(Q_i)} \tilde{U}_i \nabla_s y_i\rbr{W_i,\pi,p_\pi^\ast}\bigg|Q_j} = 0,
\end{equation}
Lemma~\ref{lemma:averaging over neighborhood} implies that for any $j\in[n]$,
\begin{equation}
    \EE\!\sbr{\rbr{\sum_{i\neq j} \frac{E_{ij}}{\sum_{k\neq i}E_{ik}} \tilde{U}_i \nabla_s y_i\rbr{W_i,\pi,p_\pi^\ast}}^2}
    = O\rbr{\frac{1}{n\rho_n}}.
\end{equation}
we can conclude that
\begin{equation}
    \zeta_2 = O_p\rbr{\frac{1}{\sqrt{n}\sqrt{n\rho_n}}} = o_p(1/\sqrt{n}).
\end{equation}
Since $\tilde{\bU}^\top\tilde{\bU} / n \pto \Ib_{J}$, we end up with $\mathcal{B} = o_p(1/\sqrt{n})$.

\vspace{2mm}
\noindent\textit{Step 3.} To deal with the last term, we go from
\begin{align}
    &\quad \EE\!\nbr{\frac{1}{n}\sum_{i=1}^n \tilde{U}_i (S_i-\pi)^2\nabla_s^2 y_i\rbr{W_i,\tilde{S}_i,P_n(\bW)+h_n\tilde{U}_i}} \\
    &\le \frac{1}{n}\sum_{i=1}^n \EE (S_i-\pi)^2
    = O\rbr{\frac{1}{n\rho_n}}
    = o\rbr{1/\sqrt{n}},
\end{align}
where we have used Lemma 15 in \cite{li2022random}. Therefore, it also holds that $\mathcal{C}=o_p(1/\sqrt{n})$.

As a result, we can conclude the asymptotic characterization for $\hat{\gamma}$ as
\begin{align}
    \hat{\gamma} &= \mathcal{D}^{-1}\rbr{\mathcal{A}+\mathcal{B}+\mathcal{C}} \\
    &= \xi_z^{-\top}\xi_y +\frac{1}{nh_n} \sum_{i=1}^n \sbr{y_i(W_i,\pi,p_\pi^\ast) - z_i(W_i,p_\pi^\ast)^\top\xi_z^{-\top}\xi_y}  \xi_{z}^{-\top} \tilde{U}_i +o_p\rbr{\frac{1}{\sqrt{n}h_n}},
\end{align}
where we have used the fact that $\frac{1}{n}\tilde{\bU}^\top\tilde{\bU}=I_J+O_p(1/\sqrt{n})$. Lastly, as suggested by Theorem 5 in \cite{munro2021treatment},
\begin{equation}
    \hat{\tau}_z = \EE\sbr{ z(1,p_{\pi}^{\ast}) - z(0,p_{\pi}^{\ast}) } + O_p(1/\sqrt{n}),
\end{equation}
where for simplicity we can denote $\tau_z^{\ast}:=\EE\sbr{ z(1,p_{\pi}^{\ast}) - z(0,p_{\pi}^{\ast}) }$. Henceforth, 
\begin{align}
    \hat{\tau}_{\AIE}^{\rmG} &= -\hat{\gamma}^{\top} \hat{\tau}_z \\
    &= -\xi_y^{\top}\xi_z^{-1} \tau_z^{\ast} -\frac{1}{nh_n} \sum_{i=1}^n \sbr{y_i(W_i,\pi,p_\pi^\ast) - z_i(W_i,p_\pi^\ast)^\top\xi_z^{-\top}\xi_y} \tilde{U}_i^\top \xi_{z}^{-1} \tau_z^{\ast} +o_p\rbr{\frac{1}{\sqrt{n}h_n}},
\end{align}
which yields the final asymptotic normal distribution.
\end{proof}

\section{Notation}

Throughout the draft, we use $C,c>0$ for positive constants that do not depend on $n$.
We write $O(\cdot),o(\cdot)$ and $O_p(\cdot),o_p(\cdot)$ in the following sense:
$a_n=O(b_n)$ if there exists some $C>0$ such that $|a_n|\le C|b_n|$;
$a_n=o(b_n)$ if $\lim_{n\to\infty}|a_n|/|b_n|=0$;
$X_n=O_p(b_n)$ if for any $\delta>0$, there exist $M,N>0$ such that
$\PP(|X_n|\ge M|b_n|)\le\delta$ for all $n>N$;
and $X_n=o_p(b_n)$ if
$\PP(|X_n|\ge \epsilon |b_n|)\to 0$ as $n\to\infty$ for any $\epsilon>0$.

\begin{table}[h!]
\centering
\renewcommand{\arraystretch}{1.15}
\scalebox{0.8}{
\begin{tabular}{l p{0.75\textwidth}}
\hline
\textbf{Symbol} & \textbf{Meaning} \\
\hline
$n$ & Number of units in the sample. \\[2pt]

$i,j \in [n]$ & Unit indices, with $[n]:=\{1,\dots,n\}$. \\[2pt]

$\{0,1\}$ & Binary treatment space, where $0$ denotes control and $1$ denotes treatment. \\[2pt]

$\bw=(w_1,\dots,w_n)\in\{0,1\}^n$ & Generic (possibly counterfactual) treatment assignment vector. \\[2pt]

$\bW=(W_1,\dots,W_n)$ & Random treatment assignment vector generated by the design. \\[2pt]

$\RCT(\pi)$ & Bernoulli randomized controlled trial with $\PP(W_i=1)=\pi$ independently across $i$. \\[2pt]

$\pi\in(0,1)$ & Baseline treatment probability under $\RCT(\pi)$. \\[2pt]

$\bW^{(1)},\bW^{(2)}$ & Independent treatment assignments used in the two-copy construction. \\[4pt]

$y_i:\{0,1\}^n\to\RR$ & Potential outcome function of unit $i$ under assignment $\bw$. \\[2pt]

$Y_i=y_i(\bW)$ & Realized outcome of unit $i$ under the realized assignment $\bW$. \\[2pt]

$\tau_{\MPE}^{\oracle}(\pi)$ & Oracle marginal policy effect. \\[2pt]

$\tau_{\ADE}^{\oracle}(\pi)$ & Oracle average direct effect. \\[2pt]

$\tau_{\AIE}^{\oracle}(\pi)$ & Oracle average indirect effect. \\[2pt]

\hline
\end{tabular}
}
\caption{General design and oracle potential-outcome notation.}
\label{tab:notation-general}
\end{table}

\begin{table}[h!]
\centering
\renewcommand{\arraystretch}{1.15}
\scalebox{0.8}{
\begin{tabular}{l p{0.75\textwidth}}
\hline
\textbf{Symbol} & \textbf{Meaning} \\
\hline
$d_i(\bw)$ & Researcher-chosen exposure mapping for unit $i$ as a function of the full assignment vector $\bw$. \\[2pt]

$\cD_i$ & Codomain of the exposure mapping $d_i:\{0,1\}^n\to\cD_i$. \\[2pt]

$h_i(d_i(\bw))$ & Generic exposure-based outcome model that depends on $\bw$ only through $d_i(\bw)$. \\[2pt]

$h_i^*(d;\pi)$ & Pseudo-true projection under $\RCT(\pi)$:
$h_i^*(d;\pi)=\EE_{\bW\sim\RCT(\pi)}[\,y_i(\bW)\mid d_i(\bW)=d\,]$. \\[2pt]

$\tilde y_i(\bw;\pi)$ & Design-induced pseudo-true potential outcome:
$\tilde y_i(\bw;\pi)=h_i^*(d_i(\bw);\pi)
=\EE_{\bW^{(2)}\sim\RCT(\pi)}
[\,y_i(\bW^{(2)})\mid d_i(\bW^{(2)})=d_i(\bw)\,]$. \\[4pt]

$\mu(\pi_1,\pi_2)$ & Two-copy population criterion:
$\frac{1}{n}\sum_{i=1}^n
\EE_{\bW^{(1)}\sim\RCT(\pi_1)}
\!\Big[
\EE_{\bW^{(2)}\sim\RCT(\pi_2)}
\big[
y_i(\bW^{(2)})\mid d_i(\bW^{(2)})=d_i(\bW^{(1)})
\big]
\Big]$. \\[2pt]

$\tau_{\MPE}(\pi)$ & Pseudo-true marginal policy effect induced by $\{d_i\}$. \\[2pt]

$\tau_{\ADE}(\pi)$ & Pseudo-true average direct effect induced by $\{d_i\}$. \\[2pt]

$\tau_{\AIE}(\pi)$ & Pseudo-true average indirect effect induced by $\{d_i\}$. \\[2pt]

$f=\{f_i\}_{i=1}^n$ & Generic candidate collection of outcome approximations $f_i:\{0,1\}^n\to\RR$. \\[2pt]

$\tau_{\star}^{\mathrm{func}}(f;\pi)$ & Functional induced by $f$ for $\star\in\{\MPE,\ADE,\AIE\}$, as defined in \eqref{eq:functional estimand}. \\[2pt]

\hline
\end{tabular}
}
\caption{Exposure mappings and pseudo-true objects in Section~\ref{sec:misspecified estimand}.}
\label{tab:notation-pseudotrue}
\end{table}

\begin{table}[h!]
\centering
\renewcommand{\arraystretch}{1.15}
\scalebox{0.8}{
\begin{tabular}{l p{0.75\textwidth}}
\hline
\textbf{Symbol} & \textbf{Meaning} \\
\hline
$\bE=(E_{ij})\in\{0,1\}^{n\times n}$ & Symmetric adjacency matrix of the observed network, where $E_{ij}=1$ indicates that $i$ and $j$ are linked. \\[2pt]

$\cN_i=\{j\neq i:E_{ij}=1\}$ & Neighborhood of unit $i$. \\[2pt]

$M_i(\bw)=\sum_{j\neq i}E_{ij}w_j$ & Number of treated neighbors of unit $i$ under assignment $\bw$. \\[4pt]

$d_i^{\rmL}(\bw)$ & Local exposure mapping; in Section~\ref{sec:local vs global example estimand},
$d_i^{\rmL}(\bw)=\{w_j:j\in\cN_i\}$. \\[2pt]

$d_i^{\rmG}(\bw)$ & Global exposure mapping; in Section~\ref{sec:local vs global example estimand},
$d_i^{\rmG}(\bw)=P_n(\bw)$. \\[2pt]

$\tau_{\AIE}^{\rmL}(\pi)$ & Local spillover effect in Section~\ref{sec:local vs global example estimand}, induced by $d_i^{\rmL}$. \\[2pt]

$\tau_{\AIE}^{\rmG}(\pi)$ & Global spillover effect in Section~\ref{sec:local vs global example estimand}, induced by $d_i^{\rmG}$. \\[4pt]

$Q_i\in\cQ$ & Latent heterogeneity for unit $i$ entering the graphon and equilibrium primitives. \\[2pt]

$G_n(u,v)$ & Graphon sequence governing link probabilities, with $E_{ij}\sim\mathrm{Bernoulli}(G_n(Q_i,Q_j))$ for $i<j$. \\[2pt]

$\rho_n$ & Sparsity parameter of the graphon sequence. \\[4pt]

$P_n(\bw)\in\RR^J$ & Equilibrium state (e.g.\ price vector) induced by assignment $\bw$. \\[2pt]

$z_i(w_i,p)\in\RR^J$ & Excess demand of unit $i$ at price $p$ under own treatment $w_i$. \\[2pt]

$Z_i=z_i(W_i,P_n(\bW))$ & Realized excess demand of unit $i$. \\[2pt]

$p_\pi^*$ & Population-clearing equilibrium under $\RCT(\pi)$. \\[2pt]

$\xi_z\in\RR^{J\times J}$ & Population price derivative of excess demand. \\[2pt]

$\xi_y\in\RR^J$ & Population price derivative of outcomes. \\[4pt]

$\tau_{\ADE}^{\oracle,*},\tau_{\AIE}^{\rmL,*},\tau_{\AIE}^{\rmG,*},\tau_{\MPE}^{\oracle,*}$ & Population limits in Theorem~\ref{thm: asymptotic limit of estimands}, satisfying
$\tau_{\MPE}^{\oracle,*}
=
\tau_{\ADE}^{\oracle,*}
+
\tau_{\AIE}^{\rmL,*}
+
\tau_{\AIE}^{\rmG,*}$. \\[2pt]

\hline
\end{tabular}
}
\caption{Local and global spillover notation in Section~\ref{sec:local vs global example}.}
\label{tab:notation-localglobal}
\end{table}

\begin{table}[h!]
\centering
\renewcommand{\arraystretch}{1.15}
\scalebox{0.8}{
\begin{tabular}{l p{0.75\textwidth}}
\hline
\textbf{Symbol} & \textbf{Meaning} \\
\hline
$\hat\tau_{\ADE}^{\oracle}$ & Horvitz--Thompson estimator for the oracle direct effect $\tau_{\ADE}^{\oracle}$ in Section~\ref{sec:local vs global example estimator}:
$\frac{1}{n}\sum_{i=1}^n
\left(\frac{W_i}{\pi}-\frac{1-W_i}{1-\pi}\right)Y_i$. \\[2pt]

$\hat\tau_{\AIE}^{\rmL}$ & PC-balancing estimator of the local spillover effect $\tau_{\AIE}^{\rmL}$. \\[2pt]

$\hat\tau_{\AIE}^{\rmG}$ & Augmented-trial estimator of the global spillover effect $\tau_{\AIE}^{\rmG}$. \\[2pt]

$\hat\tau_{\MPE}^{\oracle}$ & Estimator of the oracle marginal policy effect formed by
$\hat\tau_{\MPE}^{\oracle}
=
\hat\tau_{\ADE}+\hat\tau_{\AIE}^{\rmL}+\hat\tau_{\AIE}^{\rmG}$. \\[4pt]

$\bnu=(\nu_i)_{i=1}^n$ & Raw network Horvitz--Thompson weights,
$\nu_i=\frac{M_i}{\pi}-\frac{N_i-M_i}{1-\pi}
=\sum_{j\in\cN_i}\left(\frac{W_j}{\pi}-\frac{1-W_j}{1-\pi}\right)$. \\[2pt]

$\hat{\bPsi}\in\RR^{n\times r}$ & Matrix of top-$r$ normalized eigenvectors of $\bE$, satisfying $\hat{\bPsi}^\top\hat{\bPsi}=I_r$. \\[4pt]

$\bU\in\RR^{n\times J}$ & Individualized perturbations in the augmented trial, with $U_{ij}\sim\mathrm{Unif}(\{\pm h_n\})$. \\[2pt]

$h_n$ & Magnitude of the augmented-trial perturbations. \\[2pt]

$\hat\gamma$ & Estimator of the equilibrium-price derivative ratio:
$\hat\gamma=(\bU^\top\bZ)^{-1}(\bU^\top\bY)$. \\[2pt]

$\hat\tau_z$ & Horvitz--Thompson estimator of the direct effect on excess demand:
$\hat\tau_z=
\frac{1}{n}\sum_{i=1}^n
\left(\frac{W_i}{\pi}-\frac{1-W_i}{1-\pi}\right)Z_i$. \\[2pt]

$V^{(1)},V^{(2)},V^{(3)}$ & Components of the asymptotic variance decomposition for $\hat\tau_{\ADE}$ in Theorem~\ref{thm:direct effect estimator}. \\[2pt]

$\mathsf V_{\rmL},\mathsf V_{\rmG}$ & Asymptotic variance constants for $\hat\tau_{\AIE}^{\rmL}$ and $\hat\tau_{\AIE}^{\rmG}$, respectively. \\[2pt]

\hline
\end{tabular}
}
\caption{Estimator and asymptotic notation in Section~\ref{sec:local vs global example estimator}.}
\label{tab:notation-estimators}
\end{table}

\end{document}